\documentclass[12pt]{amsart}
\usepackage{amscd,amssymb,amsthm,verbatim}
\usepackage{enumerate}
\usepackage{bm}

\usepackage{mathrsfs, graphicx} 
\usepackage{stmaryrd}
\usepackage{slashed}

\usepackage[%
bookmarks=true,
colorlinks,
linkcolor=blue,
urlcolor=blue,
citecolor=blue,
plainpages=false,
pdfpagelabels,
final,
breaklinks=true\between
]{hyperref}
\usepackage{hyperref}
\usepackage[numbers,sort&compress]{natbib}

\numberwithin{equation}{section}

\usepackage{enumitem}
\theoremstyle{plain}
\newtheorem{theorem}{Theorem}[section]
\newtheorem{proposition}[theorem]{Proposition}
\newtheorem{lemma}[theorem]{Lemma}
\newtheorem{corollary}[theorem]{Corollary}

\theoremstyle{definition}
\newtheorem{definition}[theorem]{Definition}
\newtheorem{example}[theorem]{Example}\theoremstyle{remark}
\newtheorem{remark}[theorem]{Remark}
\usepackage{csquotes}
\usepackage{graphicx}
\usepackage[dvipsnames]{xcolor}
\usepackage{float}

\colorlet{mycolor}{RedViolet}

\colorlet{mycolor2}{RedViolet}
\usepackage{accents}
\usepackage{pgfplots}
\pgfplotsset{compat=1.18}

\newcounter{Counter}
\theoremstyle{plain}

\newcommand{\Id}{\operatorname{Id}}
\newcommand{\tr}{\operatorname{tr}}
\newcommand{\Ric}{\mathrm{Ric}}

\newcommand{\R}{\operatorname{R}}

\usepackage{MnSymbol}

\DeclareFontFamily{U}{MnSymbolC}{}
\DeclareSymbolFont{MnSyC}{U}{MnSymbolC}{m}{n}
\DeclareFontShape{U}{MnSymbolC}{m}{n}{
	<-6>  MnSymbolC5
	<6-7>  MnSymbolC6
	<7-8>  MnSymbolC7
	<8-9>  MnSymbolC8
	<9-10> MnSymbolC9
	<10-12> MnSymbolC10
	<12->   MnSymbolC12}{}
\DeclareMathSymbol{\intprod}{\mathbin}{MnSyC}{'270}

\begin{document}
\title[Causal character of imaginary Killing spinors]{Causal character of imaginary Killing spinors and spinorial slicings}
\author{Sven Hirsch}
\address{Columbia University, 2990 Broadway, New York NY 10027, USA}
\email{sven.hirsch@columbia.edu }
\author{Yiyue Zhang}
\address{Beijing Institute of Mathematical Sciences and Applications, Beijing, 101408, China}
\email{zhangyiyue@bimsa.cn}

\maketitle

\begin{abstract}
We characterize spin initial data sets that saturate the BPS bound in the asymptotically AdS setting. 
This includes both gravitational waves and rotating black holes in higher dimensions, and we establish a sharp dimension threshold in each case.
A key ingredient in our argument is a theorem providing a general criterion for when an imaginary Killing spinor of mixed causal type can be replaced by one that is strictly timelike or null.
Moreover, in analogy with the minimal surface method, we demonstrate that spinors can be used to construct a codimension-$2$ slicing.
\end{abstract}

\tableofcontents

\section{Introduction}

A central theme in mathematical relativity is the study of conserved quantities—energy, linear momentum, angular momentum, and center of mass—and the geometric structures that encode them. The geometry of mass-minimizing initial data sets with vanishing cosmological constant is well understood, both in the asymptotically flat \cite{SchoenYau79PMT, Witten1981, SY1981, EichmairHuangLeeSchoen11, HKK, BeigChrusciel, Huang-Lee:2024, HuangLee20, HirschHuang25, BHKKZ, HZ23, ChruscielMaerten05, HuangLee23EqualityII} and asymptotically hyperboloidal settings \cite{Wang2001, CJL, HirschJangZhang24, ACG2008, HJM2020, Sakovich2021, CG2021, CD2019,Lundberg2023, CWY2016, Zhang2004}. In sharp contrast, the case of a negative cosmological constant (asymptotically Anti-de Sitter) remains largely unexplored despite its growing importance due to the AdS/CFT correspondence \cite{Maldacena97}. 
The BPS bound, established in the pioneering work \cite{Maerten2006, ChruscielMaertenTod,WXZ2015}, which replaces the well-known positive mass theorem (PMT), is significantly more intricate:

    \begin{theorem}\label{Thm Intro PMT}
        Let $(M^n,g,k)$ be an asymptotically AdS spin initial data set satisfying the dominant energy condition $$\mu\ge |J|.$$ 
   Then for each imaginary Killing spinor $\psi^\infty\in\overline{\mathcal S}(\mathbb H^n)$ we have the mass formula
\begin{align}\label{mass formula intro}
     \mathcal H(\psi^\infty):= \mathcal EN+\langle   \mathcal P,Y\rangle+\langle   \mathcal C,X\rangle+\langle \omega,  \mathcal  A\rangle\ge0.
\end{align}
    \end{theorem}

  Here $\mu$ and $J$ are the energy and momentum densities; $\mathcal E$ is the energy; $\mathcal P$ is the momentum vector; $\mathcal C$ is the center of mass; and $\mathcal A$ is the angular momentum two-form of $(M^n,g,k)$. 
  Moreover, the quantities $N,X,Y,\omega$ are determined from the spinor $\psi^\infty$ and are defined in Section \ref{S:preliminaries}.
  This is in sharp contrast to the usual PMT, which only encodes the energy and the linear momentum but not the other two charges \cite{HuangSchoenWang2011}.
  We state a precise version of Theorem \ref{Thm Intro PMT} in Theorem \ref{pmt maerten} and Corollary \ref{cor pmt maerten}.
  We also point out that Theorem \ref{Thm Intro PMT} is completely open in the non-spin setting.

\medskip

Geometric inequalities in the setting of a vanishing cosmological constant are remarkably rigid:
The case of equality in the asymptotically flat PMT is realized by pp-wave spacetimes\footnote{Planar waves with parallel rays which model gravitational waves.} \cite{HZ24, HirschHuang25}, the asymptotically hyperboloidal PMT by Minkowski space \cite{HJM2020, HirschJangZhang24}, the Riemannian Penrose inequality by the Schwarzschild metrics \cite{Bray01, HuiskenIlmanen97}, and the Riemannian charged PMT by the extremal Reissner-Nordstr\"om metric \cite{BHKKZ}.
In contrast, introducing a negative cosmological constant leads to a far richer landscape.
In particular, \eqref{mass formula intro} is saturated by initial data sets contained in AdS space, Siklos wave spacetimes (which model both gravitational and, in some cases, electromagnetic waves),  and ultraspinning black holes.

\medskip

In this paper we fully characterize mass-minimizing initial data sets: 
\begin{theorem}\label{Thm Intro rigidity}
    In the same setting as Theorem \ref{Thm Intro PMT}, suppose additionally that the total charges vanish, i.e., $  \mathcal H(\psi^\infty)=0$.
    Then $(M^n,g,k)$ isometrically embeds into a spacetime $(\overline M,\overline g)$ with second fundamental form equal to $k$.
    Moreover, $(\overline M^{n+1},\overline g)$ admits an imaginary Killing spinor $\overline \psi$ such that one of the following two scenarios occurs: 
    \begin{enumerate}
        \item $\overline \psi$ is null. In this case the spacetime is a Siklos wave. Moreover, if $n=3$ or $n=4$, or if the AdS decay rate\footnote{See Definition \ref{Def:AH}.} $q$ satisfies $q>n-2$, the Siklos wave must be trivial, i.e. $(\overline M^{n+1},\overline g)$ is the AdS spacetime. 
        \item $\overline \psi$ is timelike. In this case the spacetime is stationary and vacuum. Moreover, if $n=3$ the spacetime must be trivial, i.e. $(\overline M^{n+1},\overline g)$ is the AdS spacetime.
    \end{enumerate}
\end{theorem}

In particular, if $n=3$, there are no non-trivial spacetimes $(\overline M^{4},\overline g)$.
The condition that $(\overline M^{n+1},\overline g)$ is stationary and vacuum in the timelike case should be compared to the Bartnik Stationary Vacuum Conjecture which was resolved recently \cite{HirschHuang25}.

\medskip

Here Siklos waves $(\overline M^{n+1},\overline g)$ \cite{Siklos85, Calvaruso2022Siklos} are a general class of wave-like solutions to the Einstein equations with negative cosmological constant which in particular model gravitational waves.
The negative cosmological constant causes the propagation direction of the waves to rotate \cite{Podolsky1998SiklosAdSWaves}.
They generalize the Kaigorodov spacetime \cite{Kaigorodov1962}, Ozsváth’s homogeneous solution to the Einstein-Maxwell equations \cite{Ozsvath65} (i.e. an electromagnetic wave traveling in AdS), and Defrise's pure radiation solution \cite{Defrise69}.
A Siklos wave admits global AdS-Brinkmann coordinates\footnote{Here we use Einstein summation, $\alpha,\beta=2,...,n$.}
\begin{align*}
    \overline g=\frac{1}{u^{2}_n}\Big( 2du_1dt + Ldu_1^{2}+\delta_{\alpha\beta}du_{\alpha}du_{\beta}\Big)
\end{align*}
for some function $L$.
In particular, they are conformal to the better-known pp-wave spacetimes with vanishing cosmological constant.
In general, Siklos waves have non-vanishing center of mass $\mathcal C$ and angular momentum $\mathcal A$.
Moreover, in contrast to their pp-wave counterpart, they do not have vanishing mass (i.e., $\mathcal E\ne|\mathcal P|$ for their energy and momentum).
A detailed overview of Siklos waves will be given in Section \ref{SS: Siklos}, where we establish several important geometric theorems and examine explicit examples.

\medskip

Previously, Maerten showed in \cite[Theorem 1.4]{Maerten2006} that $(M^n,g,k)$ is contained in the AdS spacetime in case $|\mathcal E|=|\mathcal P|=|\mathcal C|=|\mathcal A|=0$.
The same conclusion has been drawn by Chrusciel, Maerten and Tod \cite[Theorems 3.1 \& 3.9]{ChruscielMaertenTod} in case either $\mathcal E=0$, 
or $n=3$ and $\mathcal E=|\mathcal P|$,
or $n=3$ and $(\ast \mathcal A)\times\mathcal C=0$. 
Moreover, Chen, Hung, Wang, and Yau \cite[Theorem 1.5]{chen2017rest} established  a related rigidity statement in dimension $n=3$ when the rest mass is vanishing. 

\medskip

The next result shows that both dimensional bounds in Theorem \ref{Thm Intro rigidity} are optimal:

\begin{theorem}\label{Thm Intro examples}
There exist asymptotically AdS spin initial data sets $(M^n,g,k)$ satisfying the dominant energy condition and saturating formula \ref{mass formula intro} with the following properties:
\begin{enumerate}
    \item For $n\ge5$ there exist initial data sets contained in non-trivial Siklos wave spacetimes admitting a null imaginary Killing spinor.
    \item For $n=4$ there exist initial data sets contained in non-trivial ultraspinning black hole spacetimes admitting a timelike imaginary Killing spinor.
\end{enumerate}
\end{theorem}

Ultraspinning black holes are a certain type of extremal Kerr-AdS spacetimes and will be introduced in Section \ref{S:ultraspinning}.
Such supersymmetric black holes are of great interest in theoretical physics  \cite{kostelecky1996solitonic, gutowski2004supersymmetric, cvetivc2005supersymmetric} and typically exhibit naked singularities when there is no Maxwell field present.

\medskip

The properties of $(\overline M^{n+1},\overline g)$  depend strongly on the causal character of the underlying imaginary Killing spinor.
The causal type can be determined from the global charges:

\begin{theorem}\label{Thm Intro determining causal type}
In the same setting as Theorem \ref{Thm Intro rigidity}, we have the following:
    \begin{enumerate}[label=(\alph*)]
    \item If $\mathcal P=0$ and $  \mathcal A$ is a simple two-form, case (1) of Theorem \ref{Thm Intro rigidity} occurs.
        \item If $\mathcal P=0$, $\mathcal C=0$ and $\mathcal A $ is non-simple, case (2) of Theorem \ref{Thm Intro rigidity} occurs.
    \end{enumerate}
\end{theorem}
We point out that every two-form in dimension $n = 3$ is simple. Furthermore, by applying a boost (cf. Proposition \ref{Prop boost}), we may set $\mathcal{P}=0$.

\medskip

The proof of Theorem \ref{Thm Intro rigidity} relies on several deep results which are of independent interest.
The first result shows that we can always construct an imaginary Killing spinor which is strictly null or strictly timelike.

\begin{theorem}\label{Thm Intro cannot change causal type}
     Let $(M,g,k)$ be an asymptotically AdS spin initial data set satisfying the dominant energy condition and suppose that there exists a spinor $ \psi^\infty\in \overline {\mathcal S}(\mathbb H^n)$ such that the mass functional vanishes, i.e.
        \begin{align*}
           \mathcal H(\psi^\infty)=0.
    \end{align*}
    Then there exists another spinor $\phi^\infty\in \overline {\mathcal S}(\mathbb H^n)$ such that the spinor $\phi$ which asymptotes to $\phi^\infty$ and solves the PDE $$\slashed D\phi=\frac12\tr_g(k)e_0\phi+\frac 12n\mathbf{i}\phi$$is either strictly null or strictly timelike.
    Moreover, we also have $\mathcal H(\phi^\infty)=0$.
\end{theorem}

This is a striking contrast to the fact that even hyperbolic space itself (i.e. the $t=0$ slice of AdS) admits spinors of mixed causal type, see Proposition \ref{Prop: existence spinor of mixed type}.
The proof of Theorem \ref{Thm Intro cannot change causal type} is highly non-trivial and relies on a new monotonicity formula for the two-form
$$N(\psi)\omega(\psi)-X(\psi)\wedge Y(\psi)$$
and a conservation law for the quantity
$$N(\psi)^2+\frac12|\omega(\psi)|^2-|X(\psi)|^2-|Y(\psi)|^2.$$
Here $N(\psi),Y(\psi),X(\psi),\omega(\psi)$ will be defined in Section \ref{S:preliminaries}.

\medskip

Like Killing vector fields, imaginary Killing spinors describe infinitesimal symmetries.
The next result shows that there is a hidden symmetry at the spinorial level:

\begin{theorem}\label{Thm Intro hidden symmetry}
        Suppose that a spinor $\psi\in \overline{\mathcal S}(M) $ solves $\nabla_i\psi=-\frac12k_{ij}e_je_0\psi-\frac12 \mathbf{i}e_i\psi$.
    Then $\varphi=\varphi(\psi)$, defined by 
    \begin{align*}
        \varphi=e_0(N(\psi)-\mathbf{i}Y(\psi)+e_0X(\psi)-\frac12 \mathbf{i}e_0\omega(\psi))\psi
    \end{align*}
    also solves
    \begin{align*}
        \nabla_i\varphi=-\frac12k_{ij}e_je_0\varphi-\frac12 \mathbf{i}e_i\varphi.
    \end{align*}
\end{theorem}

The spinor $\varphi $ allows us to show that the Killing development of $(M,g,k)$ saturating \eqref{mass formula intro} admits an imaginary Killing spinor.
More precisely,  to extend $\psi$ from the initial data set to the Lorentzian spacetime, we need to solve the equation $\partial_\tau\psi(\tau)=\varphi(\psi(\tau))$.  

\medskip

Lorentzian manifolds admitting Killing (or parallel) spinors play a central role in theoretical physics where they correspond to (purely gravitational) supersymmetry \cite{Figueroa-OFarrill}.
Classifying spacetimes admitting parallel or Killing spinors is an active field of research, see for instance \cite{tod1983all, SUSY11, 5SUSY, SUSY6} and the survey \cite{SUSYsurvey}.

\medskip

Besides spinors, the other central technique to study scalar-curvature-related questions is the use of minimal surfaces \cite{SchoenYau79} or more generally minimal slicings \cite{BrendleHirschJohne24}.
It has been of significant interest to understand precisely how both methods relate to each other.
We will show that spinors can be used to construct a certain type of minimal 2-slicings:

\begin{theorem}\label{Thm Intro minimal slicing}
   Suppose that a null type I spinor $\psi\in\mathcal{\overline{S}}(M)$ solves $$\nabla_i\psi=-\frac12k_{ij}e_je_0\psi-\frac12\mathbf{i}e_i\psi.$$ 
   Then there exists a foliation $\Sigma^1_t$ of hypersurfaces, and each $\Sigma_t^1$ is foliated by hypersurfaces $\Sigma^2_s$ which are flat.
\end{theorem}

This is a crucial ingredient when constructing global AdS-Brinkmann coordinates for Siklos waves in the proof of Theorem \ref{Thm Intro rigidity}.
This construction turns out to be significantly more technical than its pp-wave counterpart \cite{HZ24}.

\medskip

\noindent \textbf{Acknowledgements:}  SH was supported in part by the National Science Foundation under Grant No. DMS-1926686, and by the IAS School of Mathematics. YZ was partially supported by NSFC grant NO.12501070 and the startup fund from BIMSA. The authors would like to thank Piotr Chru\'sciel and Hari Kunduri for helpful discussions and interest in this article.

\section{Preliminaries}\label{S:preliminaries}

\subsection{Hyperbolic and AdS space}

We begin by recalling three common coordinate charts for the standard hyperbolic space.

\medskip

  The hyperboloid in $\mathbb{R}^{1,n}$ is given by
    \begin{align*}
        \mathbb{H}^n=\{(t,x)\in\mathbb{R}^{1,n}\,|\, x_1^2+\cdots x_n^2-t^2=-1,t>0\}.
    \end{align*}
    We call $(x_1,\cdots,x_n)$ the hyperboloidal coordinates. In this chart, the hyperbolic metric $b$ is written as
    \begin{align*}
        b=\frac{dr^2}{1+r^2}+r^2g_{\mathbb{S}^{n-1}}    
    \end{align*}
    where $r=|x|$ and $g_{\mathbb{S}^{n-1}}$ is the standard unit sphere metric.

    \medskip 
    
    In the upper half-space model $\mathbb H^n=\{(y_1,\dots,y_n)\in\mathbb R^n : y_n>0\}$ the hyperbolic metric $b$ takes the form
        \begin{align*}
            b=\frac{1}{y_n^2}(dy_1^2+\cdots+dy_n^2).
        \end{align*}

        \medskip

In the Poincar\'e ball model, $\mathbb H^n=B_1^n=\{(z_1,\dots,z_n)\in\mathbb R^n:|z|<1\}$ the hyperbolic metric $b$ is written as
        \begin{align*}
            b=\left(\frac{2}{1-|z|^2}\right)^2(dz_1^2+\cdots+dz_n^2).
        \end{align*}
        
\medskip

We also define
\begin{align*}
    &r=|x|,\quad t=\sqrt{1+r^2} \text{ in the hyperboloidal coordinates,}\\
    &\rho=\sqrt{y_1^2+\cdots+ y_{n-1}^2}\text{ in the upper half-space coordinates,}\\
    &{\varrho}=|z|\text{ in the ball coordinates.}
\end{align*}

\medskip

We have the following transformation formulas: 
\begin{align*}
     y_i=&\frac{x_i}{t+x_n}\text{ for }i=1,2,\ldots,n-1, && y_n=\frac{1}{t+x_n},\\
    x=&\frac{2}{1-{\varrho}^2}z, &&r=\frac{2{\varrho}}{1-{\varrho}^2},\\
     t=&\frac{1+{\varrho}^2}{1-{\varrho}^2}, &&t+x_n=\frac{1}{y_n}.
\end{align*}

\medskip

The Anti-de Sitter (AdS) space metric in upper half-space coordinates is given by
     \begin{align*}
           \overline g=\frac{1}{y_n^2}(-dt^2 +dy_1^2+\cdots+dy_n^2)
        \end{align*}
and solves the Einstein vacuum equation with cosmological constant $\Lambda=-\frac12n(n-1)$
\begin{align*}
    \overline R_{ij}-\frac12 \overline R \overline g_{ij}+\Lambda \overline g_{ij}=0.
\end{align*}
The minus sign in the metric has several geometric consequences.
For instance, in hyperbolic space $\mathbb H^n$ geodesics satisfy the equation $\rho''(t)=\rho(t)$.
This implies that geodesics which start at the same point will move away from each other at an exponential pace. 
On the other hand, timelike geodesics in AdS satisfy the ODE $\rho''(t)=-\rho(t)$ which leads to periodic solutions. 
Hence two geodesics which start at the same point will intersect again after a finite time.

\subsection{Asymptotically AdS initial data sets}

\begin{definition}[Weighted function spaces] \label{Def:weighted} 

 For $ a \in (0, 1)$, $s=0, 1, 2, \dots$, and $q\in \mathbb{R}$, the weighted H\"older space $C^{s,a}_{-q} (\mathbb{H}^n\setminus B_2)$ is  the collection of $C^{s,a}_{\mathrm{loc}}(\mathbb{H}^n\setminus B_2)$-functions $f$ that satisfy
\begin{align*}
		\| f \|_{C^{s,a}_{-q}(\mathbb{H}^n\setminus B_2)} := \sum_{|I|=0, 1,\dots, k}\sup_{x\in \mathbb{H}^n\setminus B_2} r^q |\nabla^I f(x)| +\sup_{x\in \mathbb{H}^n\setminus B_2} r^{q} [\nabla^s f]_{a; B_1(x)} <\infty,
\end{align*}
where the covariant derivatives and norms are taken with respect to the reference metric $b$, $B_1(x)\subset \mathbb{H}^n$ is a geodesic unit ball centered at $x$, and  
\[
	[ \nabla^s f]_{a; B_1(x)} := \sup_{1\le i_1, \dots, i_s\le n} \sup_{y\neq z\in B_1(x)} \frac{|e_{i_1} \cdots e_{i_s}(f) (y) - e_{i_1}\dots e_{i_s}(f) (z)|}{ d(y, z)^a}
\] 
with respect to a fixed local orthonormal frame $\{ e_1,\dots, e_n\}$ on $(\mathbb{H}^n\setminus \{0\}, b)$.
\end{definition}

Occasionally, we also write $f=O_{s,\alpha}(r^{-q})$ to indicate $f\in C^{s,\alpha}_{-q}$. 

\begin{definition}[Asymptotically AdS initial data sets] \label{Def:AH}
Let $(M^n,g)$ be a connected complete Riemannian manifold without boundary, and let $k$ be a symmetric $2$-tensor on $M^n$. 
We say $(M^n, g, k)$ is a $C^{s,a}_{-q}$-asymptotically hyperboloidal initial data set of decay rate $q\in (\tfrac {n}2,n]$ if it is satisfying the following conditions:  
\begin{enumerate}
    \item There exists a compact set $\mathcal{C}\subseteq M$ such that  $M\setminus \mathcal{C}$ decomposes as a disjoint union $\cup_{\ell=1}^{\ell_0}M_{end}^{\ell}$ where each end $M_{end}^\ell$ 
  is diffeomorphic to the complement of a ball $\mathbb{H}^n \setminus B$ via a diffeomorphism $\phi^\ell$.
    \item On each end, let
    \begin{equation*}
        e^\ell:=\phi_\ast^\ell g- b \quad\text{and}\quad \eta^\ell:=\phi_\ast^\ell k.
    \end{equation*}
    Then we have for each $\ell$
    \begin{equation*}
        (e^\ell,\eta^\ell)\in C^{s,a}_{-q}(M^n)\times C^{s-1,a}_{-q}(M^n).
    \end{equation*}
    \item The energy density $\mu$ and the momentum density $J$ given by
    \begin{equation*}
            \mu=\frac{1}{2}\left( R+|\operatorname{tr}_g(k)|^2-|k|_g^2+n(n-1)\right),\qquad
            J=\operatorname{div}\left(k-(\operatorname{tr}_gk)g\right)
    \end{equation*}
   satisfying
    \begin{equation*}
           \mu\bar{r}\in L^1(M^n),\quad |J|_g\,\bar{r}\in L^1(M^n)
    \end{equation*}
    where $\bar{r}=r\circ\phi^\ell$ on each end.
\end{enumerate}
\end{definition}

We emphasize that $\mu$ has an additional $\frac12n(n-1)$ term in its definition coming from the cosmological constant and the Einstein equations compared to the asymptotically flat setting.

\subsection{Null and imaginary Killing spinors}

We denote with $\mathcal S$ the spinor bundle of $M^n$, with $\nabla$ the induced connection, and with $\slashed D=e_i\nabla_i$ the Dirac operator. In addition, we denote by $\mathring{\nabla}$ the connection with respect to $b$ on the spinor bundle of $\mathbb{H}^n$.

\medskip

   We write $\overline{\mathcal S}=\mathcal S\oplus \mathcal S$ for the spacetime spinor bundle. Given $\phi=(\phi_1,\phi_2)\in \overline{\mathcal S}$, Minkowski space $\mathbb R^{1,n}=\operatorname{span}\{e_0,e_1,\dots,e_n\}$ acts on $\overline{\mathcal S}$ via Clifford multiplication 
   \begin{align*}
        e_l(\psi_1,\psi_2)=(e_l\psi_1,-e_l\psi_2),\qquad e_0(\psi_1,\psi_2)=(\psi_2,\psi_1).
    \end{align*}
    The corresponding connection and Dirac operator will still be denoted with $\nabla, \slashed D$.

\begin{definition}
Given an arbitrary spinor $\phi\in\overline{\mathcal S}$, we set 
$$N(\phi)=|\phi|^2,\qquad X(\phi)=\langle e_ie_0\phi,\phi\rangle e_i.$$
$(N(\phi),X(\phi))\in\mathbb R^{1,n}$ is often referred to as the Dirac current of $\phi$.
We also define the vector field $Y$ and two-form $\omega$ by
$$Y(\phi)=\langle \mathbf{i}e_i\phi,\phi\rangle e_i,\qquad \omega_{ij}(\phi)=\operatorname{Im}\langle e_ie_je_0\phi,\phi\rangle.$$
\end{definition}

When no confusion arises, we write $N$, $X$, $Y$ and $\omega$ for $N(\psi)$, $X(\psi)$, $Y(\psi)$ and $\omega(\psi)$.
Note that $N,X,Y,\omega$ are real-valued. Moreover, $N\ge|X|$ and $N\ge|Y|$, i.e. $(N,X)\in\mathbb R^{1,n}$ and $(N,Y)\in\mathbb R^{1,n}$ are both causal vectors.

\begin{definition}
    We say a spinor $\psi\in\overline{\mathcal S}$ is null if $|X(\psi)|=N(\psi)$, or equivalently, if the Dirac current $(N,X)\in \mathbb R^{1,n}$ is a null vector.
    If $\psi$ is not null, we say $\psi$ is timelike.
\end{definition}

\begin{definition}
   We say a spinor $\psi\in\overline{\mathcal S}$ is type I, if $N(\psi)=|Y(\psi)|$, or equivalently, if $(N,Y)\in\mathbb R^{1,n}$ is a null vector.
    If $\psi$ is not type I, we say $\psi$ is type II.
\end{definition}

Type I and type II spinors naturally arise in the study of imaginary Killing spinors \cite{Baum1989}.

\begin{proposition}\label{Prop N=X implies Nw=XY}
Let $(M^n,g,k)$ be an initial data set and let $\phi\in\overline{\mathcal S}(M)$.
    Suppose that $N=|X|$ at a point $p\in  M$.
    Then $$N\omega=X\wedge Y,\qquad \text{and}\qquad X\perp Y$$ at $p$. 
    The same conclusion holds if the assumption $N=|X|$ is replaced by $N=|Y|$. Moreover, in both cases, $$N^2-|X|^2-|Y|^2+\frac{1}{2}|\omega|^2=0$$ at $p$. 
\end{proposition}

\begin{proof} 
Choose coordinates such that $X=|X|e_1$ at $p$. 
This implies $e_1e_0\phi=\phi$ since $N=|X|$. 
We have
\begin{align*}
    Y_1=\langle \mathbf{i}e_1\phi,\phi\rangle =-\langle \mathbf{i}e_0e_1e_0\phi,\phi\rangle= -\mathbf{i}\langle e_0\phi,\phi\rangle. 
\end{align*}
Note that $Y_1$ is real while $-\mathbf{i}\langle e_0\phi,\phi\rangle$ is imaginary. 
Hence, both terms vanish.
Moreover, we have for $e_\alpha,e_\beta\perp e_1$
\begin{align*}
    \omega_{1\alpha}=\operatorname{Im}\langle e_0e_1e_\alpha \phi,\phi\rangle=
    \operatorname{Im}\langle -e_\alpha \phi,\phi\rangle=Y_\alpha
\end{align*}
and
\begin{align*}
    \omega_{\alpha\beta}=\operatorname{Im}\langle e_\alpha e_\beta e_0\phi,\phi\rangle=-\operatorname{Im}\langle e_\alpha e_\beta e_1\phi,\phi\rangle=0.
\end{align*}
Therefore, $N\omega=X\wedge Y$. 

\medskip

Similarly, if $|Y|=N$, we have $Y=e_2|Y|$ and $\mathbf{i}e_2\phi=\phi$ after choosing appropriate coordinates.
Moreover,
\begin{align*}
    X_2= \langle -e_0e_2\phi,\phi\rangle =\mathbf{i} \langle e_0\phi,\phi\rangle
\end{align*}
which implies that $X\perp Y$.
Finally, we have for $e_\alpha,e_\beta\perp e_2$
\begin{align*}
    \omega_{2\alpha}=-\operatorname{Im}\langle e_0e_\alpha e_2\phi,\phi\rangle=\operatorname{Im}\langle \mathbf{i}e_0e_\alpha\phi,\phi\rangle=- X_\alpha,
\end{align*}
and
\begin{align*}
    \omega_{\alpha\beta}=\operatorname{Im}\langle e_0e_\alpha e_\beta\phi,\phi\rangle=-\operatorname{Re}\langle e_0e_\alpha e_\beta e_2\phi,\phi\rangle=0.
\end{align*}
Therefore, we also obtain $N\omega=X\wedge Y$. 

\medskip

Furthermore, we have in both cases, $\frac{1}{2}|\omega|^2=N^{-2}|X|^2|Y|^2$. 
Therefore,
\begin{equation*}
    N^2-|X|^2-|Y|^2+\frac{1}{2}|\omega|^2=N^{-2}(N^2-|X|^2)(N^2-|Y|^2)=0
\end{equation*}
which finishes the proof. 
\end{proof}

\begin{lemma}
    Given any two perpendicular unit vectors $e_1, e_2 \in \mathbb{R}^n$, there exist 
    $2^{\lfloor \frac{n-2}{2} \rfloor}$ linearly independent parallel spinors 
    $u \in \overline{\mathcal{S}}(\mathbb{R}^n)$ such that $X(u)=e_1$ and $Y(u)=e_2$.
\end{lemma}
\begin{proof}
    Note that $e_1 e_0$ and $\mathbf{i} e_2$ commute and that 
    $(e_1 e_0)^2 = (\mathbf{i} e_2)^2 = \operatorname{id}$. 
    Hence, there exist $\tfrac{1}{4}\dim\overline{\mathcal{S}}(\mathbb{R}^n)$ linearly independent
    unit-length parallel spinors $u \in \overline{S}(\mathbb{R}^n)$ satisfying
    \[
        e_1 e_0 u = u, \qquad \mathbf{i} e_2 u = u.
    \]
    These equations imply that $X(u)=e_1$ and $Y(u)=e_2$, as claimed.
\end{proof}

Imaginary Killing spinors in hyperbolic space correspond to parallel spinors in flat space.

\begin{proposition}\label{Prop construction killing spinor}
    Let $u=(u_1,u_2)\in \overline{\mathcal S}(B_1)$ be a constant spinor in the ball $B_1^n\subseteq \mathbb R^n$. 
    Let $\xi=\frac{1-{\varrho}^2}2$ and $b=\xi^{-2}\delta$ be the hyperbolic metric and define
\begin{align*}
  \psi^\infty=  \xi^{-\frac12}(1-\mathbf{i}{\varrho}e_{\varrho})u
\end{align*}
where $e_{\varrho}$ denotes the radial unit vector field (with respect to $g_{\mathbb R^n}$). 
Then $\psi^\infty$ is an imaginary Killing spinor with eigenvalue $\lambda=-\frac12\mathbf i$.
\end{proposition}

Although this result is well known, see for instance \cite{ChruscielMaertenTod}, we provide a proof for completeness' sake.

\begin{proof}
It is well known \cite{BHMMM} that for a metric $\overline g=e^{-2f}g$ and a spinor $\phi$, we have
\begin{align*}
    \overline{\nabla}_i\phi =\nabla_i\phi+\frac12 e_i \nabla f\phi +\frac12\nabla_i f\phi
\end{align*}
under conformal transformations.
   Applying this identity in our setting, we obtain
    \begin{align*}
    \begin{split}
        \mathring\nabla_i\psi^\infty=&
       \frac12 \xi^{-\frac12}x_i(1-\mathbf{i}{\varrho}e_{\varrho})u-\xi^{\frac12}\mathbf{i}e_iu-\frac12({\varrho}e_ie_{\varrho}+x_i)\xi^{-\frac12}(1-\mathbf{i}{\varrho}e_{\varrho})u+0\\
    = & \frac12 \xi^{-\frac12}(x_i-\mathbf{i}x_i{\varrho}e_{\varrho})u-\frac12\xi^{-\frac12}(\mathbf{i}e_i-{\varrho}^2\mathbf{i}e_i)u\\
    &-\frac12\xi^{-\frac12}({\varrho}e_ie_{\varrho}+x_i+i{\varrho}^2e_i-\mathbf{i}{\varrho}x_ie_{\varrho})u\\
        =&-\frac12 i e_i \xi^{-\frac12}(1-\mathbf{i}{\varrho}e_{\varrho})u.
        \end{split}
    \end{align*}
    Therefore, $\mathring\nabla_i\psi^\infty=-\frac12\mathbf i\psi^\infty$.
\end{proof}

\subsection{The AdS positive mass inequality}

Recall that in the upper half-space coordinates the hyperbolic metric takes the form $$b=y_n^{-2}(dy_1^2+\cdots+dy_n^2).$$
We first show that given any type I null spinor, we can adjust the upper half-space coordinates accordingly.

\begin{proposition} \label{infinty}
    Let $\psi^\infty\in\overline{\mathcal S}(\mathbb H^n)$ be a type I null imaginary Killing spinor.
    Then there exist upper half-space coordinates such that 
    \begin{equation*}
        N(\psi^\infty)=y_n^{-1}, \qquad
X(\psi^\infty)=y_n^{-2}\mathring{\nabla}y_1, \qquad Y(\psi^\infty)=y_n^{-2}\mathring{\nabla}y_n.
    \end{equation*}
\end{proposition}

\begin{proof}
Since $\psi^\infty$ is an imaginary Killing spinor, i.e. $\mathring{\nabla}_i\psi^\infty=-\frac12\mathbf i e_i\psi^\infty$, a short computation yields
$$\mathring{\nabla}_i N(\psi^\infty)=-Y_i(\psi^\infty), \qquad
\mathring{\nabla}_i Y_j(\psi^\infty)=-b_{ij}N(\psi^\infty), \qquad \text{and}\qquad \mathring{\nabla}_i X_j(\psi^\infty)=\omega_{ij}(\psi^\infty).$$
This calculation is carried out in more generality in Theorem \ref{Thm: N,X,Y,omega}.
Moreover, since $\psi^\infty$ is type I null, we have $$N(\psi^\infty)=|X(\psi^\infty)|_b=|Y(\psi^\infty)|_b.$$
Consequently, $\mathring{\nabla}_{ij}N(\psi^\infty)=N(\psi^\infty) b_{ij}$ and $|\mathring{\nabla} N(\psi^\infty)|_b=N(\psi^\infty)$. 
Therefore, $N(\psi^\infty)=y_n^{-1}$ and $Y(\psi^\infty)=\mathring{\nabla}y_n^{-1}$ for some upper half-space coordinates, i.e. $b=y_n^{-2}g_{\mathbb R^n}$. 
It remains to determine the coordinate function $y_1$.

\medskip

By Proposition \ref{Prop N=X implies Nw=XY}, we have $\omega(\psi^\infty)=N(\psi^\infty)^{-1}X(\psi^\infty)\wedge Y(\psi^\infty)$.
This implies $$\mathring{\nabla}_i X_j(\psi^\infty)=y_n(X_j(\psi^\infty)\mathring{\nabla}_iy_n^{-1}-X_i(\psi^\infty)\mathring{\nabla}_jy_n^{-1}),$$
and thus, 
\begin{equation*}
 \mathring{\nabla}_i (y_n^{2}X_j(\psi^\infty))=y_n X_j(\psi^\infty)\mathring{\nabla}_iy_n+y_nX_i(\psi^\infty)\mathring{\nabla}_jy_n   
\end{equation*}
Hence, there exists a function $\mathbf{f}$ such that $y_n^{2}X(\psi^\infty)=\mathring{\nabla}\mathbf{f}$. 
Recall that the flat metric is given by $g_{\mathbb R^n}=y_n^2b$.
Consequently, $\nabla^2_{\mathbb{R}^n}\mathbf{f}=\langle d\log y_n,d\mathbf{f} \rangle_{\mathbb{R}^n}g_{\mathbb R^n}$$=0$ with the help of $X(\psi^\infty)\perp Y(\psi^\infty)$. 
Hence, $\mathbf{f}$ is linear on $\mathbb{R}^{n-1}$, and we can choose a coordinate function $y_1$ such that $X(\psi^\infty)=y_n^{-2}\mathring{\nabla}y_1$. 
\end{proof}

Next, we give a more precise version of Theorem \ref{Thm Intro PMT}:

    \begin{theorem}\label{pmt maerten}
        Let $(M^n,g,k)$ be an asymptotically AdS of order $q\in(\frac n2, n]$ spin initial data set satisfying the dominant energy condition $\mu\ge|J|$.
        Then there exists to each imaginary Killing spinor $\psi^\infty\in\overline{\mathcal S}(\mathbb H^n)$ a spinor $\psi\in \overline{\mathcal S}(M)$ solving $$\slashed D\psi=\frac 12\tr_g(k)e_0\psi+\frac12n\mathbf{i}\psi$$ which asymptotes to $\psi^\infty$ with decay and regularity
\begin{align*}
    \sigma\in H^1(M^n)\cap C^{2,a}(M^n)
\end{align*}
       where $\sigma=\psi-\psi^\infty$. Moreover, we have the mass formula
        \begin{align*}
            &\mathcal H(N(\psi^\infty),X(\psi^\infty))
            =\int_M \left( \left|\nabla_i\psi+\frac12k_{ij}e_je_0\psi+\frac12\mathbf{i}e_i\psi\right|^2+\frac12\mu|\psi|^2+\frac12\langle \psi, Je_0\psi\rangle\right).
        \end{align*}
        where 
        \begin{align*}
  \mathcal  H(N,X) = \lim_{R \to \infty} \frac{1}{2(n-1)\omega_{n-1}} \int_{S_R} \left( \mathbb{U}^i(N) + \mathbb{V}^i(X) \right) \nu_i dA,
\end{align*}
with
\begin{align*}
    \mathbb{U}^i(N) =&2 \left( Ng^{i[k} g^{j]l} \mathring\nabla_j g_{kl} + \mathring{\nabla}^{[i}Ng^{j]k} (g_{jk} - b_{jk}) \right),\\
    \mathbb{V}^i(X) =& 2 \left( k_{ij} - \tr_g(k) g_{ij} \right) X^j.
\end{align*}
Here $S_R$ denotes a sphere of radius $R$ in the asymptotic end with outward-pointing normal $\nu$.
    \end{theorem}

We further analyze this positive mass inequality, or more precisely \emph{BPS bound}, in Section \ref{SS:charges}.
There we will also relate the positivity of $\mathcal H$ to the energy, linear momentum, center of mass and angular momentum of the initial data set $(M^n,g,k)$.

    \medskip

The above result corresponds to setting the cosmological constant $\Lambda=-\frac{n(n-1)}2$.
A similar inequality also holds for other choices of $\Lambda<0$.
For $\Lambda\to0$ the boundary term recovers the well-known quantity
$$EN(\psi^\infty)+\langle P,X(\psi^\infty)\rangle$$
for the ADM energy-momentum $E,P$ given by
     \begin{align*}
\lim_{R \to \infty} \frac{1}{2(n-1)\omega_{n-1}} \int_{r=R} \mathbb{U}^i(N)  \nu_i dA=NE,\quad\quad
 \lim_{R \to \infty} \frac{1}{2(n-1)\omega_{n-1}} \int_{r=R} \ \mathbb{V}^i(X)  \nu_i dA=\langle X,P\rangle.
\end{align*}
In particular, one immediately has $E\ge|P|$.
However, in the asymptotically AdS setting, the situation is more complicated since $\psi^\infty$ and therefore also $N(\psi^\infty),X(\psi^\infty)$ are not constant at infinity.

    \begin{proof}
        The inequality is due to Maerten \cite{Maerten2006, ChruscielMaertenTod} and can be derived with the help of the divergence identity
        \begin{align*}
         &\left|\nabla_i\psi+\frac12k_{ij}e_je_0\psi+\frac12\mathbf{i}e_i\psi\right|^2+\frac12\mu|\psi|^2+\frac12\langle Je_0\psi,\psi\rangle\\
         &-\left|\slashed D\psi-\frac 12\tr_g(k)e_0\psi-\frac12n\mathbf{i}\psi\right|^2\\
         =&\nabla_i(\langle \psi, \nabla_i\psi\rangle-\langle e_i\psi, \slashed D\psi\rangle)+\frac12\nabla_i\langle \psi,(k_{ij}e_j-\tr_g(k)e_i)e_0\psi-(n-1)\mathbf{i}e_i\psi\rangle.
    \end{align*}
        However, some comments are in order regarding the decay and regularity of the spinor $\psi$.

\medskip

        The existence theory of the spinor $\psi\in H^1(\overline{\mathcal S}(M))$ solving $\slashed D\psi=\frac12\tr_g(k)e_0+\frac n2\mathbf{i}\psi$ is standard, see e.g. \cite{HirschJangZhang24, Maerten2006}.
        To improve the regularity, we argue as in \cite{HZ24}. The operator $\slashed D-\frac12\operatorname{tr}_g(k)e_0-\frac12n\mathbf i$ is a linear elliptic system of first order. 
    Writing the equation in local coordinates and applying the Calderon-Zygmund estimates for elliptic systems \cite[Theorem 6.2.5]{Morrey} yields $W^{1,p}_{loc}$ regularity for any $p>1$ which can be bootstrapped to $W^{2,p}_{loc}$. 
    Using Sobolev and Schauder estimates for elliptic systems, finishes the proof.
  \end{proof}

\begin{definition}
We call $\psi^\infty\in \overline{\mathcal S}(\mathbb H^n)$ a mass-minimizing spinor if $\psi^\infty$ minimizes $\mathcal H(N(\psi^\infty),X(\psi^\infty))$ among all other spinors $\phi^\infty\in \overline{\mathcal S}(\mathbb H^n)$ with $|u|=|v|$ where $u,v$ are the constant spinors in Euclidean space corresponding to $\psi^\infty,\phi^\infty$ as in Proposition \ref{Prop construction killing spinor}.
\end{definition}

Note that by a simple compactness argument such a mass-minimizing spinor $\psi^\infty$ (together with $u,\psi$) always exists.


\section{Global charges and characterization of null mass-minimizing spinors}

Recall that imaginary Killing spinors in $\mathbb H^n$ can be constructed from parallel spinors in $\mathbb R^n$ via 
\begin{align*}
  \psi^\infty=  \xi^{-\frac12}(1-\mathbf{i}{\varrho}e_{\varrho})u
\end{align*}
where $\xi=\frac{1-{\varrho}^2}2$ and $b=\xi^{-2}\delta$ is the hyperbolic metric.
In this section we analyze the relationship between $\psi^\infty $ and $u$ with a focus on $N,X,Y,\omega$.
This will allow us to define the charges $\mathcal E,\mathcal P,\mathcal C,\mathcal A$ in Section \ref{SS:charges}.

\medskip

Before doing so, we remark that imaginary Killing spinors can also be constructed via the upper half-space model.
For instance, let $\phi$ be a constant spinor in $\mathbb R^n$ equipped with the flat metric $\delta =dy_1^2+\dots+dy_n^2$ and suppose that $N(\phi)=|Y(\phi)|$, i.e. $e_n\phi=-\mathbf{i}\phi$ for some unit vector $e_n=Y(\phi)|Y(\phi)|^{-1}$.
Consider the hyperbolic metric $b=y_n^{-2}\delta$.
Then 
\begin{align*}
    \mathring \nabla_i (y_n^{-\frac12}\phi)=\frac12 y_n^{-\frac12}(e_ie_n+\delta_{in})\phi-\frac12 {{y_n^{-\frac{1}{2}}}} \delta_{in}\phi=-\frac12 y_n^{-\frac12}\mathbf{i}e_i\phi,
\end{align*}
and $\psi^\infty=y_n^{-\frac12}\phi$ is an imaginary Killing spinor.
However, for our purposes the construction via the Poincar\'e disc model above will be more useful.

\medskip

Conversely, given hyperbolic space and a null spinor $\psi^\infty$, the spinor $\phi=y_n^{\frac12}\psi^\infty$ is a parallel spinor with respect to the metric $y_n^2g$.
Also, note that $dy_n=Y$.

\medskip

This implies that in the null case there are essentially only two charges instead of four, see the discussion in Section \ref{SS:charges} for the general case. 
More precisely, recall from Theorem \ref{pmt maerten}
\begin{equation*}
    \mathcal{H}(N(\psi^\infty),X(\psi^\infty)) = \lim_{R \to \infty} \frac{1}{2(n-1)\omega_{n-1}} \int_{r=R} \left( \mathbb{U}^i(N) + \mathbb{V}^i(X) \right) dS_i\ge 0,
\end{equation*}
where
\begin{equation*}
    \mathbb{U}^i(N) = 2\sqrt{\det g} \left( Ng^{i[k} g^{j]l} \dot{D}_j g_{kl} + D^{[i}Ng^{j]k} (g_{jk} - b_{jk}) \right),
\end{equation*}
and
\begin{equation*}
    \mathbb{V}^i(X) = 2\sqrt{\det g} \left( k^i{}_j - K \delta^i{}_j \right) X^j.
\end{equation*}
For null spinors we have $N(\psi^\infty)=y_n^{-1}N(\phi)$ and $X(\psi^\infty)=y_n^{-1}X(\phi)$.
Therefore, following \cite[page 11]{ChruscielMaertenTod} (also see Section \ref{SS:charges}), we obtain
\begin{align}\label{eq UHS E P def}
    \mathcal{H}(N(\psi^\infty),X(\psi^\infty))=N(\phi)\mathcal H(y_n^{-1},0)+X^i(\phi)\mathcal H(0,e_i y_n^{-1})=:N(\phi)\widehat{\mathcal E}+\langle X(\phi),\widehat{\mathcal P}\rangle.
\end{align}
We call these two charges UHS-energy and UHS-linear momentum which differ from the ones defined in Section \ref{SS: Siklos} via the Poincar\'e ball model. Here UHS stands for upper half-space model.
The latter definition has the advantage that it also works well in the timelike setting.

\subsection{Symmetries between parallel spinors in $\mathbb R^n$ and imaginary Killing spinors in $\mathbb H^n$} 

Given a spinor $\phi$, recall the definitions $N(\phi)=|\phi|^2$, $X(\phi)=\langle e_ie_0\phi,\phi\rangle e_i$, $Y(\phi)=\langle \mathbf{i}e_i\phi,\phi\rangle$, and $\omega_{ij}(\phi)=\operatorname{Im}\langle e_ie_je_0\phi,\phi\rangle$.

\begin{lemma}\label{lemma psi u N X Y omega}
   Let $u$ be a constant spinor in $B_1\subseteq\mathbb R^n$ and let $\psi^\infty$ be the corresponding imaginary Killing spinor, cf. Proposition \ref{Prop construction killing spinor}.
   Then
    \begin{align*}
        \begin{split}
            N(\psi^\infty)=&\frac2{1-{\varrho}^2}((1+{\varrho}^2)N(u)-2{\varrho} Y_{\varrho}(u)),\\
            Y_i(\psi^\infty)=&\frac{2}{1-{\varrho}^2}((1-{\varrho}^2)Y_i(u)-2{\varrho}N(u)\delta_{i{\varrho}}+2{\varrho}^2Y_{{\varrho}}(u)\delta_{i{\varrho}}),\\
            X_i(\psi^\infty)=&\frac{2}{1-{\varrho}^2}((1+{\varrho}^2)X_i(u)-2\varrho\omega_{i{\varrho}}(u)-2{\varrho}^2X_{\varrho}(u)\delta_{i{\varrho}}),\\
            \omega_{ij}(\psi^\infty)=&\frac{2}{1-{\varrho}^2}((1-{\varrho}^2)\omega_{ij}(u)+2{\varrho}\delta_{i{\varrho}}X_j(u)-2\varrho\delta_{j{\varrho}}X_i(u)+2{\varrho}^2\delta_{{\varrho}j}\omega_{i{\varrho}}(u)-2{\varrho}^2\delta_{{\varrho}i}\omega_{j{\varrho}}(u)).
        \end{split}
    \end{align*}
\end{lemma}

Here we use orthonormal coordinates with respect to both $\delta$ and $b$.
E.g. we have $\omega_{ij}(u)=\omega(u)(e_i,e_j)$ and $\omega_{ij}(\psi^\infty)=\omega(\psi^\infty)(\xi^{-1}e_i,\xi^{-1}e_j)$ where $\delta(e_i,e_j)=\delta_{ij}$.

\begin{proof}
  Recall that $\psi=  \xi^{-\frac12}(1-\mathbf{i}{\varrho}e_{\varrho})u$, $\xi=\frac{1-{\varrho}^2}2$, ${\varrho}e_{\varrho}=x_ie_i$, and $g=\xi^{-2}\delta$.
  We compute for $e_i\perp e_j$
  \begin{align*}
      \begin{split}
          &-\langle \mathbf{i}e_ie_je_0(1-\mathbf{i}{\varrho}e_{\varrho})u,(1-\mathbf{i}{\varrho}e_{\varrho})u\rangle\\
          =&\omega_{ij}(u)-{\varrho}\langle e_ie_je_0e_{\varrho}u,u\rangle+{\varrho}\langle e_ie_je_0u,e_{\varrho}u\rangle-{\varrho}^2\langle \mathbf{i}e_ie_je_0e_{\varrho}u,e_{\varrho}u\rangle\\
          =&\omega_{ij}(u)+\textbf{I}+\textbf{II}+\textbf{III}.
      \end{split}
  \end{align*}
  To compute the first term, note that $\langle e_ie_je_0e_{\varrho}u,u\rangle$ is imaginary unless $e_i=e_{\varrho}$ or $e_j=e_{\varrho}$. Therefore,
  \begin{align*}
      \begin{split}
          \textbf{I}=&{\varrho}\delta_{i{\varrho}}X_j(u)-{\varrho}\delta_{j{\varrho}}X_i(u).
      \end{split}
  \end{align*}
Similarly, we obtain $\textbf{II}={\varrho}\delta_{i{\varrho}}X_j(u)-{\varrho}\delta_{j{\varrho}}X_i(u)=\textbf{I}$.
Next, we calculate
\begin{align*}
    \begin{split}
        \langle e_ie_je_0e_{\varrho}u,e_{\varrho}u\rangle=&
        -\langle e_{\varrho}e_ie_je_0e_{\varrho}u,u\rangle\\
        =&-\langle (e_ie_{\varrho}+2\delta_{i{\varrho}})e_je_{\varrho}e_0u,u\rangle\\
        =&\langle (e_ie_{\varrho}+2\delta_{i{\varrho}})(e_{\varrho}e_j+2\delta_{j{\varrho}})e_0u,u\rangle.
    \end{split}
\end{align*}
Consequently,
\begin{align*}
    \textbf{III}=-{\varrho}^2\omega_{ij}(u)+2{\varrho}^2\delta_{j{\varrho}}\omega_{i{\varrho}}(u)-2{\varrho}^2\delta_{i{\varrho}}\omega_{j{\varrho}}(u).
\end{align*}
    Similarly, we obtain
   \begin{align*}
       \begin{split}
           \langle (1-\mathbf{i}{\varrho}e_{\varrho})u,(1-\mathbf{i}{\varrho}e_{\varrho})u\rangle=&(1+{\varrho}^2)|u|^2-2{\varrho}\langle \mathbf{i}e_{\varrho}u,u\rangle,\\
           \langle \mathbf{i}e_i(1-\mathbf{i}{\varrho}e_{\varrho})u,(1-\mathbf{i}{\varrho}e_{\varrho})u\rangle=&(1-{\varrho}^2)\langle \mathbf{i}e_iu,u\rangle-2{\varrho}|u|^2\delta_{i{\varrho}}+2{\varrho}^2\delta_{i{\varrho}}\langle \mathbf{i}e_\varrho u,u\rangle , \\
           \langle e_ie_0(1-\mathbf{i}{\varrho}e_{\varrho})u,(1-\mathbf{i}{\varrho}e_{\varrho})u\rangle=&(1+{\varrho}^2)\langle e_ie_0u,u\rangle-2{\varrho}\operatorname{Im}\langle e_i e_{\varrho} e_0 u,u\rangle-2{\varrho}^2\langle e_{\varrho} e_0u,u\rangle \delta_{i{\varrho}}.
       \end{split}
   \end{align*}
   Hence, the result follows.
\end{proof}

\begin{corollary}\label{cor Nw=XY psi -> u}
We have $$N(\psi^\infty)\omega(\psi^\infty)=X(\psi^\infty)\wedge Y(\psi^\infty)$$ if and only if $$N(u)\omega(u)=X(u)\wedge Y(u).$$
More generally, this conclusion already holds if 
$$N(\psi^\infty)\omega(\psi^\infty)-X(\psi^\infty)\wedge Y(\psi^\infty)=o(1).$$
\end{corollary}

\begin{proof}
To show the first claim, we compute for $e_i=e_{\varrho}$ and $e_j\neq e_{\varrho}$
    \begin{equation*}
    \begin{split}
        &\frac{(1-{\varrho}^2)^2}{4}\left(N(\psi^\infty)\omega(\psi^\infty)-X(\psi^\infty)\wedge Y(\psi^\infty)\right)_{{\varrho}j}
        \\=& \left((1+{\varrho}^2)N(u)-2{\varrho}Y_{\varrho}(u)\right)\left((1+{\varrho}^2)\omega_{{\varrho}j}(u)+2{\varrho}X_j(u)\right)
        \\&-(1-{\varrho}^2)X_{\varrho}(u)(1-{\varrho}^2)Y_j(u)+ \left((1+{\varrho}^2)X_j(u)-2{\varrho}\omega_{j{\varrho}}(u)\right)\left((1+{\varrho}^2)Y_{\varrho}(u)-2{\varrho}N(u)\right)
        \\=&
        (1-{\varrho}^2)^2\left(N(u)\omega_{{\varrho}j}(u)-X_{\varrho}(u)Y_j(u)+X_j(u)Y_{\varrho}(u)\right).
    \end{split}
    \end{equation*}
    Similarly, we obtain for $e_i,e_j\neq e_{\varrho}$ 
    \begin{equation*}
    \begin{split}
        &\frac{(1-{\varrho}^2)^2}{4}\left(N(\psi^\infty)\omega(\psi^\infty)-X(\psi^\infty)\wedge Y(\psi^\infty)\right)_{ij}
        \\=& \left((1+{\varrho}^2)N(u)-2\varrho Y_\varrho(u)\right)(1-{\varrho}^2)\omega_{ij}(u)
        \\&-\left((1+{\varrho}^2)X_i(u)-2\varrho\omega_{i{\varrho}}(u)\right)(1-{\varrho}^2)Y_j(u)+\left((1+{\varrho}^2)X_j(u)-2\varrho\omega_{j{\varrho}}(u)\right)(1-{\varrho}^2)Y_i(u)  
        \\=&(1+{\varrho}^2)(1-{\varrho}^2)\left(N(u)\omega(u)
        -X(u)\wedge Y(u)\right)_{ij}
        \\&-2\varrho(1-{\varrho}^2)\left(Y_{\varrho}(u)\omega_{ij}(u)-\omega_{i{\varrho}}(u)Y_j(u)+\omega_{j{\varrho}}(u)Y_i(u)\right)
        \\=& \left(1+{\varrho}^2-2\varrho N^{-1}(u)Y_{\varrho}(u)\right)(1-{\varrho}^2)\left(N(u)\omega(u)
        -X(u)\wedge Y(u)\right)_{ij}
        \\&-2\varrho(1-{\varrho}^2)\left(Y_{\varrho}(u)N^{-1}(u)(X_i(u)Y_j(u)-Y_i(u)X_j(u))-\omega_{i{\varrho}}(u)Y_j(u)+\omega_{j{\varrho}}(u)Y_i(u)\right)
    \end{split}
    \end{equation*}
    The last line vanishes if $N(u)\omega_{{\varrho}j}(u)=X_{\varrho}(u)Y_j(u)-X_j(u)Y_{\varrho}(u)$ for $e_j\neq e_{\varrho}$.
    Thus, the first claim follows.

    \medskip

 To prove the second claim, suppose that  $N(\psi^\infty)\omega(\psi^\infty)-X(\psi^\infty)\wedge Y(\psi^\infty)=o(1)$.
 Therefore, \begin{align*}
     N(u)\omega_{j{\varrho}}(u)-X_{\varrho}(u)Y_j(u)+X_j(u)Y_{\varrho}(u)=o(1)
 \end{align*} which implies that this term must vanish.
Consequently, we also have
\begin{align*}
    (1-{\varrho}^2)^{-1}(1+{\varrho}^2-2\varrho N^{-1}(u)Y_{\varrho}(u))\left(N(u)\omega(u)  -X(u)\wedge Y(u)\right)_{ij}=o(1)
\end{align*}
 which shows $N(u)\omega(u)=X(u)\wedge Y(u)$ in case $\sup |N^{-1}(u)Y_{\varrho}(u)|<1$.
 If $\sup |N^{-1}(u)Y_{\varrho}(u)|= 1$, then we have $N(u)=|Y(u)|$ everywhere and $N(u)\omega(u)=X(u)\wedge Y(u)$ by Proposition \ref{Prop N=X implies Nw=XY}.
\end{proof}

\begin{corollary}\label{cor N=Y(u) implies N=Y(psi)}
    We have $|Y(u)|=N(u)$ if and only if $|Y(\psi^\infty)|=N(\psi^\infty)$.
    Moreover, if $|Y(u)|=N(u)$, then also $|X(u)|=N(u)$ if and only if $|X(\psi^\infty)|=N(\psi^\infty)$
\end{corollary}

\begin{proof}
    We normalize such that $N(u)=1$ and compute
    \begin{align*}
        \begin{split}
           & \frac{(1-{\varrho}^2)^2}{4}(N(\psi^\infty)^2-|Y(\psi^\infty)|^2)\\
           =&(1+{\varrho}^2)^2+4{\varrho}^2Y_{\varrho}(u)^2-4{\varrho}Y_{\varrho}(u)(1+{\varrho}^2)\\
            &-(1-{\varrho}^2)^2|Y(u)|^2-4{\varrho}^2-4{\varrho}^4Y_{\varrho}(u)^2+4{\varrho}Y_{\varrho}(u)(1-{\varrho}^2)-4{\varrho}^2(1-{\varrho}^2)Y_{\varrho}(u)^2+8{\varrho}^3Y_{\varrho}(u)\\
            =&(1-{\varrho}^2)^2(1-|Y(u)|^2).
        \end{split}
    \end{align*}
    This shows the first claim.

    \medskip

To prove the second claim, recall from Proposition \ref{Prop N=X implies Nw=XY} that $N(u)\omega(u)=X(u)\wedge Y(u)$ and $N(\psi^\infty)\omega(\psi^\infty)=X(\psi^\infty)\wedge Y(\psi^\infty)$.
Moreover, both $X(u)$ and $Y(u)$ as well as $X(\psi^\infty)$ and $Y(\psi^\infty)$ are perpendicular.
Normalize again to $N(u)=1=|Y(u)|$ and compute
    \begin{align*}
        \begin{split}
            & \frac{(1-{\varrho}^2)^2}{4}(N(\psi^\infty)^2-|X(\psi^\infty)|^2)\\
           =&(1+{\varrho}^2)^2+4{\varrho}^2Y_{\varrho}(u)^2-4{\varrho}Y_{\varrho}(u)(1+{\varrho}^2)\\
           &-(1+{\varrho}^2)^2|X(u)|^2-4{\varrho}^2(|X(u)|^2Y_{\varrho}(u)^2+X_{\varrho}(u)^2)-4{\varrho}^4X_{\varrho}(u)^2\\
           &+4{\varrho}(1+{\varrho}^2)|X(u)|^2Y_{\varrho}(u)+4{\varrho}^2(1+{\varrho}^2)X_{\varrho}(u)^2\\
           =&(1-|X(u)|^2)((1+{\varrho}^2)-2{\varrho}Y_{\varrho}(u))^2.   
        \end{split}
    \end{align*}  
    This finishes the proof.
\end{proof}

\begin{corollary} \label{cor N^2-X^2-Y^2-w^2 psi -> u}
    We have
    \begin{align*}
    \begin{split}
        |Y(\psi^\infty)|^2-N^2(\psi^\infty)=&4(|Y(u)|^2-N^2(u)),\\
       |X(\psi^\infty)|^2-\frac12|\omega(\psi^\infty)|^2=&4(|X(u)|^2-\frac12|\omega(u)|^2).
    \end{split}
    \end{align*}
    In particular, both $|Y(\psi^\infty)|^2-N^2(\psi^\infty)$ and $|X(\psi^\infty)|^2-\frac12|\omega(\psi^\infty)|^2$ are constant.
    Moreover,
    \begin{align*}
        |Y(\psi^\infty)|^2+|X(\psi^\infty)|^2-N^2(\psi^\infty)-\frac12|\omega(\psi^\infty)|^2=4\left(|Y(u)|^2+|X(u)|^2-N^2(u)-\frac12|\omega(u)|^2\right).
    \end{align*}
\end{corollary}

\begin{proof}
    We compute
    \begin{align*}
        \begin{split}
        &\sum_{i=1}^n((1-{\varrho}^2)Y_i(u)-2{\varrho}N(u)\delta_{i{\varrho}}+2{\varrho}^2Y_{\varrho}(u)\delta_{i{\varrho}})^2\\
            &-((1+{\varrho}^2)N(u)-2{\varrho} Y_{\varrho}(u))^2\\
                        =&(1-{\varrho}^2)^2(|Y(u)|^2-N^2(u)),
        \end{split}
    \end{align*}
and
    \begin{align*}
        \begin{split}
                        &\sum_{i=1}^n((1+{\varrho}^2)X_i(u)-2{\varrho}\omega_{i{\varrho}}(u)-2{\varrho}^2X_{\varrho}(u)\delta_{i{\varrho}})^2\\
            &-\frac12\sum_{i,j=1}^n((1-{\varrho}^2)\omega_{ij}(u)+2{\varrho}\delta_{i{\varrho}}X_j(u)-2{\varrho}\delta_{j{\varrho}}X_i(u)+2{\varrho}^2\delta_{j{\varrho}}\omega_{i{\varrho}}(u)-2{\varrho}^2\delta_{i{\varrho}}\omega_{j{\varrho}}(u))^2\\
            =&(1-{\varrho}^2)^2(|X(u)|^2-\frac12|\omega(u)|^2).
        \end{split}
    \end{align*}
Hence, the result follows.
\end{proof}

\subsection{Positivity of mass and charges at infinity}\label{SS:charges}

Recall from Theorem \ref{pmt maerten} that for asymptotically AdS initial data sets $(M,g,k)$ satisfying the dominant energy condition
      \begin{align*}
            \mathcal H(N(\psi^\infty),X(\psi^\infty))\ge0
        \end{align*}
        where 
        \begin{align*}
  \mathcal  H(N,X) = \lim_{R \to \infty} \frac{1}{2(n-1)\omega_{n-1}} \int_{r=R} \left( \mathbb{U}^i(N) + \mathbb{V}^i(X) \right) \nu_i dA.
\end{align*}
Moreover, recall from Lemma \ref{lemma psi u N X Y omega} that
            \begin{align*}
        \begin{split}
            N(\psi^\infty)=&\frac2{1-{\varrho}^2}((1+{\varrho}^2)N(u)-2{\varrho} Y_{\varrho}(u)),\\
            X_i(\psi^\infty)=&(1+{\varrho}^2)X_i(u)-2{\varrho}\omega_{i{\varrho}}(u)-2{\varrho}^2X_{\varrho}(u)\delta_{i{\varrho}}.
        \end{split}
    \end{align*}
Hence, we obtain as in \cite{ChruscielMaertenTod}
\begin{align*}
    \mathcal  H(N(\psi^\infty),X(\psi^\infty)) =&
   N(u) \mathcal H(\frac2{1-{\varrho}^2}((1+{\varrho}^2),0)+Y^i(u)\mathcal H(-2x_i ,0)\\
    &+X_i(u)\mathcal H(0,(1+{\varrho}^2)e_i-2x^ix_ke_k)+\omega_{ij}(u)\mathcal H(0,-e_ix_j+x_ie_j)
\end{align*}
and we call the charges
\begin{align}\label{eq: def charges}
\begin{split}
    \mathcal E=&\mathcal H(\frac2{1-{\varrho}^2}((1+{\varrho}^2),0),\\
   \mathcal P_i=&\mathcal H(-2x_i ,0),\\
   \mathcal C_i=&\mathcal H(0,(1+{\varrho}^2)\delta_i^k\partial_k-2x^ix_k\partial_k),\\
   \mathcal A_{ij}=&\mathcal H(0,-\partial_ix_j+x_i\partial_j)
   \end{split}
\end{align}
the energy $\mathcal E$, center of mass $\mathcal C$, the linear momentum $\mathcal P$ and the angular momentum $\mathcal A$.
Thus, the positive mass theorem in the AdS setting becomes:

\begin{corollary}\label{cor pmt maerten}
   Given an asymptotically AdS spin initial data set $(M,g,k)$ satisfying the dominant energy condition $\mu\ge|J|$, we have
   \begin{align*}
      \mathcal EN(u)+\langle   \mathcal P,Y(u)\rangle+\langle   \mathcal C,X(u)\rangle+\langle \omega(u),  \mathcal  A\rangle\ge0
\end{align*}
for any constant spinor $u$ at infinity.
\end{corollary}

Consequently, one obtains a stronger result compared to the asymptotically flat or hyperboloidal case where there is no control on the center of mass or angular momentum.
The algebraic implications of this inequality are analyzed in detail in \cite{ChruscielMaertenTod}.
Since Corollary \ref{cor pmt maerten} is not just showing that the mass is positive, we will not refer to this result as positive mass theorem, but rather call it a BPS bound which is more common in the Physics literature. These Bogomol'nyi–Prasad–Sommerfield bounds play an important role in supersymmetry. 

\begin{remark}
    Mathematically, the reason there are only two charges in the asymptotically flat setting comes from $N,X$ being constant at infinity, so they can be taken outside of the integral at infinity.
    Moreover, there are explicit examples \cite{HuangSchoenWang2011} which demonstrate that the center of mass and angular momentum cannot be controlled.
    Similarly, in the asymptotic hyperboloidal setting, $X(\psi^\infty)$ cannot be chosen freely since $X(\psi^\infty)=-\nabla N(\psi^\infty)$ at infinity.
\end{remark}

\begin{remark}
    Note that $\mathcal C=0$ and $\mathcal A=0$ when $k=0$ (which is equivalent to the case $\Lambda=0$ and $k=g$). In particular, Theorem \ref{Thm Intro rigidity} shows that $(M,g,k)$ is contained in the AdS spacetime in this setting. 
    Hence, Theorem \ref{Thm Intro rigidity} also recovers the results of \cite{HirschJangZhang24}.
\end{remark}

\begin{remark}
    We point out that $\mathcal E$ has one, $\mathcal P,\mathcal C$ have each $n$, and $\mathcal A$ has $\frac{n(n-1)}2$ degrees of freedom.
   This adds up to $1+2n+\frac{n(n-1)}2=\frac{(n+1)(n+2)}2$ which is the dimension of $SO(n,2)$, the symmetry group of AdS space.
\end{remark}

In the remainder of this section, we analyze how the properties of the charges $\mathcal C,\mathcal P,\mathcal A$ characterize the existence of a mass-minimizing null spinor.

    \begin{lemma}
        Suppose $\mathcal H$ gets minimized by a spinor satisfying $N(u)=|Y(u)|$.
        Then we also have $N(u)=|X(u)|$ after possibly replacing $u$ by another minimizing spinor.
        The same statement holds with $Y(u)$ and $X(u)$ interchanged.
    \end{lemma}

    \begin{proof} 
Recall from Proposition \ref{Prop N=X implies Nw=XY} that $N\omega=X\wedge Y$.
Thus, the mass term becomes
   \begin{align*}
  \mathcal H(N(\psi^\infty),X(\psi^\infty))=    \mathcal EN(u)+\langle   \mathcal C,X(u)\rangle+\langle   \mathcal P+N^{-1}(u)  \mathcal A(X(u),\cdot),Y(u)\rangle.
\end{align*}
Keeping $N(u),X(u)$ fixed, the right-hand side becomes minimized for maximal $|Y(u)|$, i.e. $|Y(u)|=N(u)$.
    \end{proof}

According to \cite[Theorem 3.1]{ChruscielMaertenTod}, $\mathcal{E}\ge |\mathcal{P}|$ and $\mathcal{E}=0$ implies the vanishing of all charges. When $\mathcal{E}= |\mathcal{P}|$, we have the following theorem.
\begin{theorem} \label{E=P}
    If $\mathcal{E}=|\mathcal{P}|$, then there exists a type I null mass-minimizing spinor.
\end{theorem}
\begin{proof}
    Assume $\mathcal{E}=|\mathcal{P}|>0$. Let $u$ be a parallel spinor satisfying $N(u)=1$,
    $Y(u)=-|\mathcal{P}|^{-1}\mathcal{P}$, and with $X(u)$ to be determined later. Then 
    $\omega(u)=X(u)\wedge Y(u)$. Hence
    \[
        \mathcal{H}(\psi^\infty)
        =\left\langle \mathcal{C}+\mathcal{A}(\,\cdot\,,Y(u)),\,X(u)\right\rangle.
    \]
    We may choose a unit vector $X(u)\in\mathbb R^n$ perpendicular to $Y(u)$ such that 
    $\mathcal{H}(\psi^\infty)\le 0$. Thus $u$ is a type I null mass-minimizing spinor.
\end{proof}

    \begin{theorem}\label{thm: examples null}
      In the following cases there exists a type I null mass-minimizing spinor:
        \begin{enumerate}
        \item  The angular momentum $\mathcal A$ vanishes.
            \item $  \mathcal A$ is a simple two-form and the center of mass $  \mathcal C$ vanishes.
            \item $  \mathcal A$ is a simple two-form and the linear momentum $  \mathcal P$ vanishes.
        \end{enumerate}
    \end{theorem}

           Here a two-form $A$ is called simple if $  \mathcal A=a e_1e_2$ for some constant $a\in\mathbb R$ and a choice of orthonormal basis $e_1,\dots,e_n$.
    In particular, in dimension $3$ every two-form is simple.

    \begin{proof}
Define 
\begin{align*}
    \mathcal Q=\begin{pmatrix}
        \mathbf{i}\mathcal P&\mathcal C-\mathbf{i}\mathcal A\\-\mathcal C-\mathbf{i}\mathcal A&-\mathbf{i}\mathcal P
    \end{pmatrix}.
\end{align*}
Then $\mathcal Q$ is a Hermitian matrix acting on the spinor space $\overline {\mathcal S}(\mathbb R^n)$.
Moreover,
\begin{align*}
    \langle \mathcal C, X(u)\rangle+\langle \mathcal P,Y(u)\rangle+\langle \mathcal A,\omega(u)\rangle=\langle \mathcal Q(u),u\rangle.
\end{align*}
Thus, to prove the above theorem, we need to find the smallest eigenvalue of $\mathcal Q$ and show that it is realized by a type I null spinor.

\medskip

For the first claim, we compute
\begin{align*}
    \mathcal Q^2=\begin{pmatrix}
        (|\mathcal C|^2+|\mathcal P|^2)\Id &-\mathbf{i}\mathcal C\wedge \mathcal P\\
     -\mathbf{i}\mathcal C\wedge \mathcal P &(|\mathcal C|^2+|\mathcal P|^2)\Id
    \end{pmatrix}
\end{align*}
where we use the convention $\mathcal P\wedge \mathcal C=\mathcal P\mathcal C-\mathcal C\mathcal P$.
Next, observe that
\begin{align*}
   \mathcal F\mathcal Q^2\mathcal F^{-1}= \begin{pmatrix}
          (|\mathcal C|^2+|\mathcal P|^2)\Id-\mathbf{i}\mathcal C\wedge \mathcal P&0\\0& (|\mathcal C|^2+|\mathcal P|^2)\Id+\mathbf{i}\mathcal C\wedge \mathcal P
    \end{pmatrix}
\end{align*}
where
\begin{align*}
  \mathcal F=  \begin{pmatrix}
    \Id&\Id\\-\Id&\Id
\end{pmatrix}.
\end{align*}
Diagonalizing $\mathbf{i}\mathcal C\wedge \mathcal P$, we obtain that the smallest eigenvalue of $\mathcal Q^{2}$ is given by 
\begin{align*}
    \lambda_{\operatorname{min}}=-|\mathcal C|^2-|\mathcal P|^2-|\mathcal C\wedge \mathcal P|.
\end{align*}
Next, let $\Pi$ be the plane containing $\text{span}\{\mathcal P,\mathcal C\}$.
Without loss of generality $\mathcal C\ne0$.
Let $e_1=\mathcal C|\mathcal C|^{-1}$ and $e_2\perp e_1$ with $|e_1|=1$.
Choose $u$ such that $X(u)=\cos\alpha e_1+\sin\alpha e_2$ and $Y(u)=\cos\alpha e_2-\sin\alpha e_1 $ where $\alpha $ satisfies
$|\mathcal C|\sin\alpha+|\mathcal P|\cos(\theta-\alpha)=0$ and $\theta$ is given by $\mathcal P=\cos\theta |\mathcal P|e_1+\sin\theta |\mathcal P|e_2$.
Then
\begin{align*}
    \langle \mathcal C,X(u)\rangle+\langle \mathcal P,Y(u)\rangle=&
    |\mathcal C|\cos\alpha+|\mathcal P|\sin(\theta-\alpha)
\end{align*}
and
\begin{align*}
     &( \langle \mathcal C,X\rangle+\langle \mathcal P,Y\rangle)^2\\
    =&|\mathcal C|^2+ |\mathcal P|^2 +2|\mathcal C||\mathcal P|\sin\theta    \\
    &-|\mathcal C|^2\sin^2\alpha-|\mathcal P|^2\cos^2(\theta-\alpha)-2|\mathcal C||\mathcal P|(\sin\theta-\cos\alpha\sin(\theta-\alpha))\\
    =&|\mathcal C|^2+|\mathcal P|^2+|\mathcal P\wedge \mathcal C|.
\end{align*}
where we used that
\begin{align*}
    \sin\theta-\cos\alpha\sin(\theta-\alpha)=\sin\alpha\cos(\theta-\alpha).
\end{align*}
If $ \langle \mathcal C,X(u)\rangle+\langle \mathcal P,Y(u)\rangle >0$, then we can choose $\check{u}$ such that $X(\check{u})=-X(u)$ and $Y(\check{u})=-Y(u)$. 
This proves the first claim.

\medskip

For the second claim, we have
\begin{align*}
    \mathcal Q^2=\begin{pmatrix}
        (|\mathcal C|^2+\frac12|\mathcal A|^2)\Id-\mathbf{i}(\mathcal C\mathcal A-\mathcal A\mathcal C)&0\\
        0&(|\mathcal C|^2+\frac12|\mathcal A|^2)\Id+\mathbf{i}(\mathcal C\mathcal A-\mathcal A\mathcal C)
    \end{pmatrix}
\end{align*}
Note that $\mathcal C\mathcal A-\mathcal A\mathcal C=-2\mathcal A(\mathcal C)$.
Therefore, the smallest eigenvalue of $\mathcal Q^2$ is given by
\begin{align*}
    \lambda_{\operatorname{min}}=-|\mathcal C|^2-\frac12|\mathcal A|^2-2|\mathcal A(\mathcal C)|.
\end{align*}
Next, let $\Pi$ be the plane spanned by $\mathcal A$ where without loss of generality $\mathcal A\ne0$.
Denote with $\mathcal C^\intercal$ and $\mathcal C^\perp$ the components of $\mathcal C$ which are tangential and perpendicular to $\Pi$ respectively.
Let $A,B$ be perpendicular unit-vectors spanning $\Pi$ such that $A=\mathcal C^\intercal|\mathcal C^\intercal|^{-1}$ in case $\mathcal C^\intercal\ne0$.
Moreover, let $C$ be a unit-vector perpendicular to $\Pi$ satisfying $C=\mathcal C^\perp|\mathcal C^\perp|^{-1}$ in case $\mathcal C^\perp\ne0$.
Choose the angle $\beta$ such that $\sin\beta(|\mathcal C^\intercal|+\frac12|\mathcal A|)=\cos\beta|\mathcal C^\perp|$.
Let $u$ be a spinor such that $X(u)=\cos\beta A+\sin\beta C$ and $Y(u)=B$.
Then
\begin{align*}
    \langle \mathcal C,X(u)\rangle+\langle \mathcal A,\omega(u)\rangle=&(|\mathcal C^\intercal|+\frac12|\mathcal A|)\cos\beta+|\mathcal C^\perp|\sin\beta
\end{align*}
and a short computation yields 
\begin{align*}
    ((|\mathcal C^\intercal|+\frac12|\mathcal A|)\cos\beta+|\mathcal C^\perp|\sin\beta)^2=&|\mathcal C|^2+\frac12|\mathcal A|^2+|\mathcal A||\mathcal C^\intercal|\\
    =&|\mathcal C|^2+\frac12|\mathcal A|^2+2|\mathcal A(\mathcal C)|.
\end{align*}
If $\langle \mathcal C,X(u)\rangle+\langle \mathcal A,\omega(u)\rangle >0$, then we can choose $\check{u}$ such that $X(\check{u})=-X(u)$ and $Y(\check{u})=Y(u)$ which implies $\omega(\check{u})=-\omega(u)$. 
This proves the second claim.

\medskip

For the third claim, we again compute
\begin{align*}
    \mathcal Q^2=\begin{pmatrix}
        (|\mathcal P|^2+\frac12|\mathcal A|^2)\Id&-\mathcal A\mathcal P+\mathcal P\mathcal A\\
        \mathcal A\mathcal P-\mathcal P\mathcal A&(|\mathcal P|^2+\frac12|\mathcal A|^2)\Id
    \end{pmatrix}
\end{align*}
Using  $\mathcal P\mathcal A-\mathcal A\mathcal P=-2\mathcal A(\mathcal P)$, we proceed as above to find that the smallest eigenvalue of $\mathcal Q^2$ satisfies
\begin{align*}
    \lambda_{\operatorname{min}}=-|\mathcal P|^2-\frac12|\mathcal A|^2-2|\mathcal A(\mathcal P)|.
\end{align*}
The corresponding type I null spinor is constructed as in the proof of the second claim.
    \end{proof}

In the reverse direction we have the following result:

\begin{proposition}\label{Prop: example non-null}
    Suppose that $\mathcal P=0$ and $\mathcal C=0$.
    Moreover, assume that $\mathcal A$ satisfies $\mathcal A=\sum_{a=1}^{k}\lambda_ae_{2a-1}\wedge e_{2a} $ for non-zero constants $\lambda_a$ with $2\le k\le \lfloor\frac{n}{2}\rfloor$. Then there exists no null mass-minimizing spinor.
\end{proposition}

\begin{proof}
Let $\epsilon_a=-\operatorname{sgn} \lambda_a$.
According to \cite[Lemma 3.5]{ChruscielMaertenTod}, there exists a spinor $u$ such that 
\begin{equation*}
    e_0u=\epsilon_0u,\quad\quad
    \mathbf i e_0e_{2a-1}e_{2a}u=\epsilon_au
\end{equation*}
for all $a\ge1$, where $\epsilon_0=1$ or $-1$. 
A short computation yields
\begin{equation*}
Y(u)=0,\quad\quad X(u)=0,\quad\quad \omega(u)=-N(u)\sum_a\epsilon_a e_{2a-1}\wedge e_{2a}.
\end{equation*}
Evaluating $\langle \omega(u),\mathcal A\rangle$, the result follows.
\end{proof}

\subsection{Boosted initial data sets}

Recall from Corollary \ref{cor pmt maerten} that $\mathcal{E}=\mathcal{H}(2t,0)$ and $\mathcal{P}_i=\mathcal{H}(2x_i,0)$ in hyperboloidal coordinates. 
By rotating coordinates, we can always ensure that $\mathcal{P}_1=|\mathcal{P}|$ and $\mathcal{P}_\alpha=0$ for $\alpha=2,\cdots,n$. 
Suppose that $\mathcal{E}>|\mathcal{P}|$. 
Let $\Psi$ be a linear transformation in $SO(1,n)$ mapping $\mathbb{R}^{1,n}$ to $\mathbb{R}^{1,n}$ such that $\Psi(\mathcal{E},\mathcal{P})=(\sqrt{\mathcal{E}^2-|\mathcal{P}|^2},0,\cdots,0)$.  
More precisely,
\begin{equation*}
    \Psi=\begin{bmatrix}
        \cosh \theta & \sinh \theta &
        \\ \sinh \theta& \cosh \theta &
        \\ & & I_{n-1}
    \end{bmatrix},
    \quad \text{where}\; \cosh \theta=\frac{\mathcal{E}}{\sqrt{\mathcal{E}^2-|\mathcal{P}|^2}},\; \sinh \theta=\frac{-|\mathcal{P}|}{\sqrt{\mathcal{E}^2-|\mathcal{P}|^2}}.
\end{equation*}
Since $t$ and $x_i$ are functions on $\mathbb{H}^n$, and $\Psi$ is a boost map on $\mathbb{H}^n$, we have 
\begin{align*}
    \widetilde{t}:=&t\circ \Psi=t\frac{\mathcal{E}}{\sqrt{\mathcal{E}^2-|\mathcal{P}|^2}}-x_1\frac{|\mathcal{P}|}{\sqrt{\mathcal{E}^2-|\mathcal{P}|^2}}, 
    \\
    \widetilde{x}_1:=&x_1\circ \Psi=x_1\frac{\mathcal{E}}{\sqrt{\mathcal{E}^2-|\mathcal{P}|^2}}-t\frac{|\mathcal{P}|}{\sqrt{\mathcal{E}^2-|\mathcal{P}|^2}},
    \\ \widetilde{x}_\alpha:=&x_\alpha\circ \Psi=x_\alpha,\;\;\;
    \alpha=2,\dots,n.
\end{align*}
Consequently,
\begin{equation*}
    \mathcal{H}(2\widetilde{t},0)=\sqrt{\mathcal{E}^2-|\mathcal{P}|^2},\quad \mathcal{H}(2\widetilde{x}_i,0)=0,\; i=1,\dots,n.
\end{equation*}
Denote by $\Phi$ is the diffeomorphism from the designated end $M_{end}$ to $\mathbb{H}^n$ and recall that the mass functional $\mathcal{H}$ is defined in terms of $\Phi$.  
Moreover, let $\widetilde{\mathcal{H}}$ be the mass functional corresponding to $\Psi\circ \Phi$.

\begin{lemma}
 Let $\mathcal H,\widetilde{\mathcal H}$ be as above. Then we have $\widetilde{\mathcal{H}}(N,X)=\mathcal{H}(N,X)$.
\end{lemma}

\begin{proof}
    Since $\Psi$ is an isometry on $\mathbb{H}^n$, the integrands of $\widetilde{\mathcal{H}}(N,X)$ and $\mathcal{H}(N,X)$ are the same. 
    Moreover,  according to \cite[Equation 2.11]{CJL},
    \begin{equation*}
        \operatorname{div}\left(N\mathbb{U} + \mathbb{V}(X)\right)=2N(\mu+n(n-1))+2\langle J,X\rangle+O(r^{1-2q}),
    \end{equation*}
    then combined with $N\mu,\langle J,X\rangle\in L^1(M^n)$, the integrals are independent of the chosen exhaustion.
\end{proof}

\begin{proposition}\label{Prop boost}
Suppose $\mathcal{E}>|\mathcal{P}|$. Then   there exists a boost map $\Psi$ such that the corresponding energy $\widetilde{\mathcal{E}}=\sqrt{\mathcal{E}^2-|\mathcal{P}|^2}$ and the linear momentum $\widetilde{\mathcal{P}}=0$.
\end{proposition}
\begin{proof}
This follows from
    \begin{equation*}
        \widetilde{\mathcal{E}}=\widetilde{\mathcal{H}}(2\widetilde{t},0)=\mathcal{H}(2\widetilde{t},0)=\sqrt{\mathcal{E}^2-|\mathcal{P}|^2},\quad 
        \widetilde{\mathcal{P}}_i=\widetilde{\mathcal{H}}(2\widetilde{x}_i,0)=\mathcal{H}(2\widetilde{x}_i,0)=0.
    \end{equation*}
\end{proof}

\begin{proof}[Proof of Theorem \ref{Thm Intro determining causal type}]
    This follows from combining Theorem \ref{E=P}, \ref{thm: examples null}, Proposition \ref{Prop: example non-null}, and Proposition \ref{Prop boost}.
\end{proof}


\section{Siklos waves}\label{SS: Siklos}

\begin{definition}
We call a Lorentzian manifold $(\overline M^{n+1},\overline g)$ a \emph{Siklos wave} if it admits global \emph{AdS-Brinkmann coordinates}, i.e. $\overline M^{n+1}=\mathbb R^n\times \mathbb R_{>0}$ and
\begin{align*}
    \overline g=\frac{1}{u^{2}_n}\Big( 2du_1dt + Ldu_1^{2}+\delta_{\alpha\beta}du_{\alpha}du_{\beta}\Big),
\end{align*}
where $u_n>0$ and $\alpha,\beta=2,\dots,n$.
Here $L$ is the \emph{wave profile function}, which is independent of $t$, and on each $\{u_1=\text{constant}\}$ slice satisfies the \emph{Siklos equation}
\begin{align*}
    \Big(\Delta_{\mathbb H^{n-1}}-2u_n\partial_{n}\Big)L\le0.
\end{align*}
Equivalently, $\mathcal{L}=u_n^{-1}(L-1)$ satisfies
\begin{align*}
     \Big(\Delta_{\mathbb H^{n-1}}-(n-1)\Big)\mathcal{L}\le 0
\end{align*}
\end{definition}

\begin{remark}
We point out that $\tilde g=u_n^2\overline g$ is a pp-wave in Brinkmann coordinates.
Hence, Siklos waves are conformal to pp-waves.
However, we note that the wave profile function of a pp-wave satisfies a different partial differential inequality
\begin{align*}
    \Delta_{\mathbb R^{n-1}}L\le0.
\end{align*}
\end{remark}

Siklos waves in this general form were first discovered by Siklos in \cite{Siklos85}. For an overview, we refer to \cite{Calvaruso2022Siklos} and for their physical interpretation to \cite{Podolsky1998SiklosAdSWaves}.

\begin{proposition}\label{Prop: stress energy siklos wave}
Let $(\overline M^{n+1},\overline g)$ be a Siklos wave with wave profile function $L$.
Then the Ricci tensor is given by
\begin{align*}
    \overline \Ric=-n\overline{g}+\left(u_n\partial_{n}L-\frac{1}{2}\Delta_{\mathbb{H}^{n-1}}L\right)u_n^{-2}du_1^2
\end{align*}
Moreover, the scalar curvature satisfies
\begin{align*}
    \overline{\R}=-n(n+1)
\end{align*}
In particular, the Einstein tensor is given by
\begin{align*}
    \overline{\Ric}-\frac{1}{2}\overline{\R}\overline{g}-\frac{n(n-1)}{2} \overline{g}=8\pi T=\left(u_n\partial_{u_n}L-\frac{1}{2}\Delta_{\mathbb{H}^{n-1}}L\right)u_n^{-2}du_1^2.
\end{align*}
Let $\hat{\mathbf{n}}=u_nL^{-\frac{1}{2}}(L\partial_t-\partial_{u_1})$ be the timelike future-pointing unit normal vector to $\{t=0\}-$slice. Then 
\begin{align*}
    \mu=&8\pi T(\hat{\mathbf{n}}, \hat{\mathbf{n}})=L^{-1}\left(u_n\partial_{u_n}L-\frac{1}{2}\Delta_{\mathbb{H}^{n-1}}L\right)\\
    J=&8\pi T(\hat{\mathbf{n}}, \cdot)=-L^{-\frac{1}{2}}\left(u_n\partial_{u_n}L-\frac{1}{2}\Delta_{\mathbb{H}^{n-1}}L\right) u_n^{-1}du_1
\end{align*}
\end{proposition}

\begin{proof}
   In dimension $n=3$ this is contained in \cite[Equation 2.3]{Calvaruso2022Siklos}.
   The higher dimensional case is computed analogously. 
\end{proof}

\begin{remark}
We can also consider Siklos waves with varying cosmological constant.
Consider the metric
\begin{align*}
    \overline g_\Lambda=\frac{\ell^2}{u^{2}_n}\Big( 2du_1dt + Ldu_1^{2}+\delta_{\alpha\beta}du_{\alpha}du_{\beta}\Big)
\end{align*}
where $\ell=-\frac{n(n-1)}{2\Lambda}$. Changing coordinates $\widetilde u_n=u_n-\ell$, $\widetilde u_A=u_A$, we obtain
\begin{align*}
    \overline g_\Lambda=\frac{\ell^2}{(\widetilde u_n+\ell)^2}\Big( 2d\widetilde u_1dt + Ld\widetilde u_1^{2}+\delta_{\alpha\beta}d\widetilde u_{\alpha}d\widetilde u_{\beta}\Big).
\end{align*}
Multiplying the inequality$ (\Delta_{\mathbb H^{n-1}}-2u_n\partial_{u_n})L\le0$ by $u_n^{-2}$ yields
\begin{align*}
   0\ge \Big(\Delta_{\mathbb R^{n-1}}    -\frac{n-1}{u_n}\partial_{u_n}\Big)L= \Big(\Delta_{\mathbb R^{n-1}}    -\frac{n-1}{\widetilde u_n+\ell}\partial_{u_n}\Big)L
\end{align*}
Thus for $\ell \to\infty$ (equivalently, $\Lambda\to0$), we obtain the well-known pp-wave metric together with the PDE $$ \Delta_{\mathbb R^{n-1}}
    L\le0.$$
\end{remark}

\begin{example}
    Let $L=L(u_1)$ be a function depending only on $u_1$. Then
\begin{align*}
    \overline g=&\frac1{u_n^2}\Big(2du_1dt+Ldu_1^2+\delta_{\alpha\beta}du_{\alpha}du_{\beta}\Big)\\
    =&\frac1{u_n^2}\Big(du_1(2dt+Ldu_1)+\delta_{\alpha\beta}du_{\alpha}du_{\beta}\Big)\\
    =&\frac1{u_n^2}\Big(du_1dv+\delta_{\alpha\beta}du_{\alpha}du_{\beta}\Big),
\end{align*}
where $v=2t+\int_0^{u_1}L(s)ds$.
Hence, the corresponding Siklos wave is the AdS spacetime.
\end{example}

In spacetime dimension $n+1=4$, several special cases of Siklos waves have already been discovered earlier:

\begin{example}
    The \emph{Kaigorodov spacetime} \cite{Kaigorodov1962} is given by the wave profile function $L=au_3^3$, $a>0$, and models pure gravitational waves.
    We have $  (\Delta_{\mathbb H^{2}}-2u_3\partial_{u_3})L=0$ and by Proposition \ref{Prop: stress energy siklos wave}, the spacetime is vacuum.
\end{example}

\begin{example}
The \emph{Defrise spacetime} \cite{Defrise69} is given by the wave profile function $L=au_3$, $a>0$.
It represents null dust or radiation.
\end{example}

\begin{example}
\emph{Ozsváth's homogeneous solution} to the Einstein-Maxwell equations \cite{Ozsvath65} is given by the wave profile function $L=a^2u_3^2$, $a>0$.
Similarly, consider the wave profile function $L=-a^2u_3^4$.
Then 
\begin{align*}
    \overline g=\frac{1}{u^{2}_3}\Big( 2du_1dt - a^2u_3^4du_1^{2}+du_2^2+du_3^2\Big)
\end{align*}
with inverse metric
\begin{align*}
    \overline g^{-1}=
    \begin{pmatrix}
        a^2u_3^6&u_3^2&0&0\\
        u_3^2&0&0&0\\
        0&0&u_3^2&0\\
        0&0&0&u_3^2
    \end{pmatrix}
\end{align*}
in these coordinates.
We point out that $\overline g$ is still Lorentzian although $L$ is chosen to be negative.
In this case the stress-energy tensor is realized by a null Maxwell field.
More precisely, let $F=adu_1\wedge du_3$ and note that
\begin{align*}
    \ast F=a du_1\wedge du_2
\end{align*}
Then the Maxwell equations $dF=0$ and $d\ast F=0$ hold.
Furthermore,
\begin{align*}
 |F|^2=   F_{\alpha\beta}F_{\gamma\delta}\overline g^{\alpha\gamma}\overline g^{\beta\delta}=0,
\end{align*}
and similarly
\begin{align*}
    |\ast F|^2=(\ast F)_{\alpha\beta}(\ast F)_{\gamma\delta}\overline g^{\alpha\gamma}\overline g^{\beta\delta}=0.
\end{align*}
Moreover,
\begin{align*}
  &2\overline g^{\gamma\delta}F_{\alpha\gamma}F_{\beta\delta}-\frac12\overline g_{\alpha\beta}|F|^2
    =2a^2\overline g^{\gamma\delta}\delta_{\alpha 1}\delta_{\gamma 3}\delta_{\beta 1}\delta_{\delta 3}
    =2a^2u_3^{2}\delta_{1\alpha}\delta_{1\beta},
\end{align*}
while on the other side
\begin{align*}
 8\pi T=&   \left(u_3\partial_{u_3}L-\frac{1}{2}\Delta_{\mathbb{H}^{2}}L\right)u_3^{-2}du_1^2
 =2a^2u_3^2du_1^2.
\end{align*}
Hence, the Einstein-Maxwell equations are satisfied and the Siklos wave models an electromagnetic wave within an AdS background.

\end{example}

We will show in Section \ref{S:asymptotic analysis} that there exist non-trivial Siklos waves $(\overline M^{n+1},\overline g)$ containing asymptotically AdS initial data sets $(M^n,g,k)$ for $n\ge5$.

\begin{remark}
Physically, Siklos waves model gravitational and, in some cases, electromagnetic radiation propagating in AdS backgrounds.
In the AdS/CFT context, such bulk waves (e.g. the Kaigorodov spacetime) correspond to energy fluxes or stress-tensor excitations in the dual conformal field theory on the boundary providing a useful exact example for studying holographic non-equilibrium dynamics \cite{CveticLuPope1999, GibbonsGielen2008}.
\end{remark}

\begin{proposition}\label{prop siklos admit spinors}
    Siklos wave spacetimes admit null imaginary Killing spinors.
\end{proposition}

\begin{proof}
Let $(\overline M^{n+1},\overline g)$ be a Siklos wave.
We construct imaginary Killing spinors on $(\overline M^{n+1},\overline g)$ from parallel spinors contained in the pp-wave spacetime $(\widehat M^{n+1},\widehat g)$ conformal to $(\overline M^{n+1},\overline g)$.

 \medskip
 
We denote by $\widehat{c}$ and $\widehat{\nabla}$ the Clifford multiplication and connection with respect to the corresponding pp-wave.
Then we have $\widehat{c}(V)\varphi=y_nV\varphi$ for any spinor $\varphi$ and vector field $V$. 
Let $\overline{\nabla}$ be the connection on $(\overline M^{n+1},\overline g)$.
Let $\phi$ be a null parallel spinor on the pp-wave satisfying $\widehat{c}(\widehat{\nabla} u_n)\phi=-\mathbf{i}\phi$.
Note that such a spinor $\phi$ exists since $\widehat c(\widehat \nabla u_n) \mathbf i:\mathcal S(\widehat M^{n+1})\to\mathcal S(\widehat M^{n+1})$ has eigenvalues $\pm1$ and divides $\mathcal S$ into eigenspaces $\mathcal S^\pm$.
Observe that both eigenspaces have the same dimension since $\widehat c(\widehat \nabla u_n)$ interchanges $\mathcal S^\pm$ with $\mathcal S^\mp$.
In particular, $\mathcal S^-$ is non-empty.
Similarly, for a pp-wave it is well-known \cite{HZ24} that a spinor $\phi$ is parallel if and only if it solves $\widehat c(\partial u_1)\phi=-\widehat c(\partial_t)\phi$.
Note that this corresponds to $\phi$ being null and the Dirac current $\llangle e_\alpha\phi,\phi\rrangle e_\alpha$ pointing in the direction of the pp-wave where $\llangle \cdot,\cdot\rrangle$ is the inner product in $\mathcal S$.
As above, $\widehat c(\partial u_1)\widehat c(\partial_t)$ divides $\mathcal S$ into eigenspaces of the same dimension.
Furthermore, $\widehat c(\partial u_1)\widehat c(\partial_t)$ commutes with $\widehat c(\widehat \nabla u_n) \mathbf i$.
Consequently, $\mathcal S$ is divided into four spaces having the same dimension.
In particular, we can find a spinor $\phi$ which is parallel (i.e. $\widehat c(\partial u_1)\phi=-\widehat c(\partial_t)\phi$) and satisfies $\widehat{c}(\widehat{\nabla} u_n)\phi=-\mathbf{i}\phi$.

\medskip

Applying the conformal formula for spinors, cf. \cite{BHMMM}
 \begin{equation*}
     \overline{\nabla}_V \phi= \widehat{\nabla}_V \phi+\frac{1}{2}\left(\widehat{\nabla}_V\log u_n+\widehat{c}(V)\widehat{c}(\widehat{\nabla} \log u_n)\right)\phi,
 \end{equation*}
 we obtain for $V=\partial t, \partial u_1\cdots, \partial u_{n-1}$
 \begin{equation*}
     \overline{\nabla}_{V}\phi=\frac{1}{2}\left(\widehat{\nabla}_V\log u_n+\widehat{c}(V)\widehat{c}(\widehat{\nabla} \log u_n)\right)\phi=-\frac{1}{2}\mathbf iV\phi
     \end{equation*}
     and for $V=\partial u_n=\widehat{\nabla}u_n$
     \begin{equation*}
     \overline{\nabla}_{\partial u_n}\phi= \frac{1}{2}\left(\widehat{\nabla}_{\partial u_n}\log u_n+\widehat{c}(\partial u_n)\widehat{c}(\widehat{\nabla} \log u_n)\right)\phi=0.
 \end{equation*}
Define $\psi=u_n^{-\frac{1}{2}}\phi$.
Then $\psi$ is an imaginary Killing spinor on $(\overline M^{n+1},\overline g)$.
\end{proof}

\begin{remark}
    The spinor space of a $(n+1)$-dimensional Lorentzian manifold satisfies $\dim (\mathcal S)=2^{\lfloor\frac{n+1}2\rfloor}$.
    In Minkowski space we have $\dim (\mathcal S)$ linearly independent parallel spinors and we have $\frac12\dim (\mathcal S)$ linearly independent parallel spinors in a pp-wave.
    For Siklos waves the proof above constructs $\frac14\dim (\mathcal S)$ spinors solving $\nabla_i\psi=-\frac12\mathbf{i}e_i\psi$.
    Moreover, by a similar argument, there are also $\frac14\dim (\mathcal S)$ many spinors solving $\nabla_i\psi=+\frac12\mathbf{i}e_i\psi$.
\end{remark}

 \begin{theorem}\label{Thm Siklos wave charges}
     Let $(M^n,g,k)$ be the $t=0$-slice of the Siklos wave $(\overline{M}^{n+1},\overline{g})$ given by
     \begin{equation*}
          \overline g=\frac{1}{u^{2}_n}\Big( 2du_1dt+ Ldu_1^{2}  +\delta_{\alpha \beta}du_{\alpha}du_{\beta}
          \Big)\quad \text{and}\quad L=1+O_2(r^{-q}).
     \end{equation*}
     Then we have for any lapse-shift pair $(N,X)$
     \begin{equation*}
       \mathcal{H}(N,X)= \frac{1}{(n-1)\omega_{n-1}}\int_{\mathbb{H}^n}(N\mu+\langle X,J\rangle)LdV_b.
     \end{equation*}
 \end{theorem}

The proof is lengthy and will be deferred to Appendix \ref{A:Siklos}.

\begin{remark}
   Recall from \eqref{eq: def charges} the formulas
    \begin{align*}
    \mathcal E=&\mathcal H(\frac2{1-r^2}((1+r^2),0),\\
   \mathcal P_i=&\mathcal H(-2x_i ,0),\\
   \mathcal C_i=&\mathcal H(0,(1+r^2)\delta_i^k\partial_k-2x^ix_k\partial_k),\\
   \mathcal A_{ij}=&\mathcal H(0,-\partial_ix_j+x_i\partial_j).
\end{align*}
Thus, we obtain explicit expressions for the charges $\mathcal E,\mathcal P,\mathcal C,\mathcal A$ of Siklos waves.
\end{remark}

\begin{remark}
    Similar formulas hold in the asymptotically flat setting:
    Given an asymptotically flat $t=0$-slice of a pp-wave spacetime $(\overline M^{n+1},\overline g)$ given by
      \begin{equation*}
          \overline g= 2dudt  + Fdu^{2} +\delta_{ij}dx^{i}dx^{j}
          \end{equation*}
          we have
          \begin{align*}
            \mathcal E=|\mathcal P|= \frac{1}{(n-1)\omega_{n-1}}\int_{\mathbb R^n}\mu F dV.
          \end{align*}
          This follows easily from integration by parts and the formula for $\mathcal E$ established in \cite[p. 751]{Huang-Lee:2024}.
\end{remark}

\begin{remark}
Interestingly, we do not have
\begin{align*}
    \mathcal E=|\mathcal P|
\end{align*}
for Siklos waves, and Siklos waves have non-trivial center of mass and angular momentum.
In particular, the propagation direction rotates.
This contrasts with pp-waves, which always have $\mathcal E=|\mathcal P|$.
On the other hand, for Siklos waves we do have
    \begin{align*}
\widehat{\mathcal E}=    |\widehat{\mathcal P}|
\end{align*}
for the UHS-energy and the UHS-linear momentum defined in \eqref{eq UHS E P def}.
To verify this, we can choose $\psi^\infty$ as in Proposition \ref{infinty} such that $N(\psi^\infty)=y_n^{-1}$ and $X(\psi^\infty)=y_n^{-2}\mathring{\nabla} y_1$.  
     Note that $N(\psi^\infty)\mu+\langle J, X(\psi^\infty)\rangle=0$. Thus, $\mathcal{H}(N(\psi^\infty),X(\psi^\infty))=0$, i.e., $\hat{\mathcal{E}}=|\hat{\mathcal{P}}|$.
\end{remark}

AdS space can be realized as sphere in the pseudo-Riemannian manifold $\mathbb R^{n,2}$.
A similar result holds for Siklos waves:

\begin{theorem}
    Siklos waves $(\overline M^{n+1},\overline g)$ embed into pseudo-Riemannian manifolds $(\check M^{n+2},\check g)$ with two time dimensions which admit parallel spinors.
    Moreover, the second fundamental form of this embedding is given by $\overline g$.
\end{theorem}

\begin{proof}
   Consider on $\check M=\overline M^{n+1}\times\mathbb R$ the Lorentzian cone metric
   \begin{align*}
       \check g=\tau^2\overline{g}-d\tau^2
   \end{align*}
   and extend $\overline \psi$ trivially to $(\check M,\check g)$ to obtain a spinor $\check \psi$ on $\check M$.
   Then $\check \psi$ is parallel, cf. \cite[Theorem 3.1]{Leitner2003_JMP}.
\end{proof}


\section{Causal character of mass-minimizing spinors}\label{S:causality}

In this section we show that there always exists a mass-minimizing spinor which is either null or timelike and does not change causal character.
Physically, this implies that there are no ergospheres in this setting.
Interestingly, even for initial data sets in the AdS spacetime, there exist mass-minimizing spinors with mixed causal character, cf. \cite[Example 5.4]{Leitner2003_JMP}.
Thus, it is necessary to replace in general the original mass-minimizing spinor by a possibly different spinor.
We restate Theorem \ref{Thm Intro cannot change causal type} for convenience:

\begin{theorem}\label{thm causal character}
    Let $(M,g,k)$ be an asymptotically AdS initial data set satisfying the dominant energy condition and suppose that there exists a spinor $ \psi^\infty\in \overline {\mathcal S}(\mathbb H^n)$ such that the mass functional vanishes, i.e.
        \begin{align*}
           \mathcal H(N(\psi^\infty),X(\psi^\infty)).
    \end{align*}
    Then there exists another spinor $\phi^\infty\in \overline {\mathcal S}(\mathbb H^n)$ such that the spinor $\phi$ which asymptotes to $\phi^\infty$ and solves the PDE $\slashed D\phi=\frac12\tr_g(k)e_0\phi+\frac n2\mathbf{i}\phi$ is either null or timelike.
    Moreover, $\mathcal H(\phi^\infty)=0$.
\end{theorem}

In the asymptotically flat setting, we have $\nabla_iN(\psi)=-k_{ij}X_j(\psi)$ and $\nabla_iX_j(\psi)=-k_{ij}N(\psi)$ which immediately implies $\nabla_i(N^2-|X|^2)=0$. 
Consequently, spinors cannot change causal type in this setting.

\medskip

Since the arguments below are quite sophisticated, we begin with giving an overview.
At the heart of the proof there are two monotonicity formulas for 
\begin{align*}
    N^2-|X|^2-|Y|^2+\frac12|\omega|^2\qquad\text{and}\qquad N\omega-X\wedge Y.
\end{align*}
They also imply that the condition $N^2=|X|^2=|Y|^2$ is preserved.

\begin{enumerate}
\item Suppose that $|X( \psi)|=N(\psi)$ at $p\in M$ and $n\ge4$. For $n=3$, a separate argument is needed.
    \item Since $|X( \psi)|=N( \psi)$, we also have $N( \psi)\omega( \psi)=X( \psi)\wedge Y( \psi)$ and $N( \psi)^2+\frac12|\omega( \psi)|^2=|X( \psi)|^2+|Y( \psi)|^2$.
    \item The properties $N( \psi)w( \psi)=X( \psi)\wedge Y( \psi)$ and $N^2( \psi)+\frac12|\omega( \psi)|^2=|X( \psi)|^2+|Y( \psi)|^2$ still hold for $\psi^\infty$ via ODE methods and decay estimates for $n\ge4$.
    \item We also have $N(u)\omega(u)=X(u)\wedge Y(u)$ and $N^2(u)+\frac12|\omega(u)|^2=|X(u)|^2+|Y(u)|^2$ for the corresponding parallel spinor $u$ in Euclidean space.
    \item Let $\widetilde u=\mathbf{i}e_0X(u)Y^\perp(u)|X(u)|^{-1}|Y^\perp(u)|^{-1}u=\mathbf{i}e_0\omega(u) |\omega(u)|^{-1}u$ (in case $\omega(u)\ne0$). Then $N(u)=N(\widetilde u)$, $X(u)=X(\widetilde u)$, $Y(u)=Y(\tilde u)$, $\omega(u)=\omega(\widetilde u)$.
    \item This implies $N(\psi^\infty)=N(\widetilde{\psi}^\infty)$, $X(\psi^\infty)=X(\widetilde{\psi}^\infty)$, $Y(\psi^\infty)=Y(\widetilde{\psi}^\infty)$, $\omega(\psi^\infty)=\omega(\widetilde{\psi}^\infty)$ where $\widetilde \psi^\infty$ is the imaginary Killing spinor in hyperbolic space corresponding to $\widetilde u$.
    \item Since $\widetilde u, \widetilde \psi^\infty$ are still mass-minimizing, there exists a spinor $\widetilde \psi$ asymptotic to $\widetilde \psi^\infty$ solving $\nabla_i\widetilde \psi=-\frac12k_{ij}e_je_0\widetilde \psi-\frac12\mathbf{i}e_i\widetilde \psi$.
\item  Let $v=\frac1{\sqrt 2}(u+\widetilde u)$, $\phi^\infty=\frac1{\sqrt2}(\psi^\infty+\widetilde\psi^\infty)$ and $\phi=\frac1{\sqrt2}(\psi+\widetilde \psi)$. Then $v=0$ if and only if $N(u)=|X(u)|=|Y(u)|$.
\item If $N(u)=|X(u)|=|Y(u)|$, then $N(\psi^\infty)=|X(\psi^\infty)|=|Y(\psi^\infty)|$ and $N(\psi)=|X(\psi)|=|Y(\psi)|$, in which case $\psi$ is strictly null. Thus, suppose that $N(u)=|X(u)|=|Y(u)|$ does not hold, i.e. $v\ne0$.
\item We have $N(v)=|X(v)|=|Y(v)|$ which shows that $N(\phi^\infty)=|X(\phi^\infty)|=|Y(\phi^\infty)|$ as well, and $N(\phi)=|X(\phi)|=|Y(\phi)|$.
\item By linearity $\nabla_i\phi=-\frac12k_{ij}e_je_0\phi-\frac12\mathbf{i}e_i\phi$. 
\item A computation yields $JN(\phi)=-\mu X(\phi)$. By the mass formula, the spinor $\phi^\infty$ is also mass-minimizing.
\end{enumerate}

To transfer information between $\psi$ and $\psi^\infty$ (as well as $\phi$ and $\phi^\infty)$, we also need to prove decay estimates which are stronger than the general ones obtained in Theorem \ref{pmt maerten}.


\subsection{Improved decay estimates}\label{S:decay estimates}

Throughout this subsection let $\psi=(\psi_1,\psi_2)$ be a spinor in the spacetime spinor bundle $\overline{\mathcal S}$ satisfying 
\begin{align*}
     \nabla_i\psi=-\frac12k_{ij}e_je_0\psi-\frac12\mathbf{i}e_i\psi
\end{align*}
which is asymptotic to a pair of imaginary Killing spinors $\psi^\infty=(\psi_1^\infty,\psi_2^\infty)$ where $\psi_i^\infty$, $i=1,2$.

\begin{lemma} \label{decay estimate}
    Suppose $\psi$ satisfies $\widetilde{\nabla}\psi=0$. Then $$\sigma=O_2(r^{-q}|\psi^\infty|_b)$$
    where $\sigma=\psi-\psi^\infty$.
\end{lemma}
\begin{proof}
Let $(\theta_1,\theta_2,\cdots, \theta_{n-1})$ be a coordinate system on $S^{n-1}$, and $(r,\theta_1,\theta_2,\dots, \theta_{n-1})$ be a hyperboloidal coordinate system on the asymptotic region of $M^n$ such that $|\psi^\infty|^2_b=t+c_0r\cos \theta_1$ for some constant $c_0\in[-1,1]$. 
Hence, $r^{-\frac{1}{2}}|\psi^\infty|_b$ is decreasing with respect to $r$.
Since $\widetilde{\nabla}_j \psi^\infty=O(r^{-q}|\psi^\infty|_b)$, we have $\widetilde{\nabla}_j \sigma=O(r^{-q}|\psi^\infty|_b)$. Consequently,
    \begin{equation}\label{ODE sigma}
        |\nabla |\sigma||\le |\sigma|+O(r^{-q}|\psi^\infty|_b)        
    \end{equation}
  outside the ball $B_R(0)$ where $R\gg1$ is chosen such that $|k|\le 1$ outside $B_R(0)$.  

  \medskip

    We first show that $|\sigma|\to 0$ when $r\to \infty$. If not, then there exists $\varepsilon>0$ such that $|\sigma|(p_i)\ge \varepsilon$ when $r(p_i)\to \infty$. Thus, for $r(p_i)$ sufficiently large,  Equation \eqref{ODE sigma} implies that $|\sigma|>\frac{\varepsilon}{2}$ for a fixed-size ball $B_{r_0}(p_i)$. However, this contradicts $\sigma\in L^2(M)$.  

    \medskip

    Next, we prove that $\sigma=O(r^{-q}|\psi^\infty|_b)$.   
    Let $\gamma(s): [0,\infty)\to M^n$ be a curve parametrized by its arclength such that $\gamma(0)=p$, $\gamma'(s)=\frac{\partial_r}{|\partial_r|}|_{\gamma(s)}$ and $\lim_{s\to \infty}r(\gamma(s))=\infty$. Then we have 
    \begin{equation*}
        \frac{d}{ds} |\sigma|(\gamma(s))\ge -|\sigma|(\gamma(s))-Cr^{-q}|\psi^\infty|_b.
    \end{equation*}
    Note that $r(\gamma(s))=O(e^{s})$, and
    \begin{equation*}
        \frac{d}{ds}[e^s|\sigma|(\gamma(s))]\ge -C_1 e^{(\frac{3}{2}-q)s}(r^{-\frac{1}{2}}|\psi^\infty|_b).
    \end{equation*}
    Since $|\sigma|\to 0$ as $r\to \infty$ and $[r^{-\frac{1}{2}}|\psi^\infty|_b]_{\gamma(s)}$ is decreasing with respect to $s$, we obtain 
    \begin{equation*}
        |\sigma|(\gamma(s))\le e^{-s}\int_s^\infty C_1[r^{-\frac{1}{2}}|\psi^\infty|_b]_{\gamma(s)} e^{(-q+\frac{3}{2})\tau}d\tau 
        =O(r^{-q}|\psi^\infty|_b).
    \end{equation*}
    Finally, the $C^2$-estimates follow by applying the standard elliptic estimates to the PDE $\widetilde{\slashed{D}}\sigma=O_{1,a}(r^{-q}|\psi^\infty|_b)$.
    \end{proof}

\begin{lemma} \label{NXY decay}
In case $\psi$ is null and type I, we have the following decay estimates:
    \begin{equation*}
        N-y_n^{-1}=O_2(r^{-q}y_n^{-1}), \qquad
        X-y_n^{-2}\nabla y_1=O_2(r^{-q}y_n^{-1}), \qquad
        Y-y_n^{-2}\nabla y_n=O_2(r^{-q}y_n^{-1}). 
    \end{equation*}
\end{lemma}
\begin{proof}
   According to Lemma \ref{decay estimate}, we have $\psi-\psi^\infty=O(r^{-q}y_n^{-\frac{1}{2}})$.
   Hence, a standard elliptic equation argument implies that $\psi-\psi^\infty=O_2(r^{-q}y_n^{-\frac{1}{2}})$. 
   Combining this with Proposition \ref{infinty}, it follows that the decay estimates hold. 
\end{proof}


\subsection{Extracting differential forms}

The next theorem shows how to square the spinor and extract the corresponding differential forms.
Physically, this corresponds to constructing the bosonic partners of the fermion $\psi$.

\begin{theorem}\label{Thm: N,X,Y,omega}
    Let $\psi\in\overline{\mathcal S}(M)$, $\psi\ne0$, be a spinor satisfying 
    \begin{align*}
    \nabla_i\psi=-\frac12k_{ij}e_je_0\psi-\frac12\mathbf{i}e_i\psi.
\end{align*}
Then $N=N(\psi)$, $X=X(\psi)$, $Y=Y(\psi)$ and $\omega=\omega(\psi)$ satisfy
\begin{equation} \label{NXYw}
\begin{split}
    \nabla_i N=&-k_{ij}X_j-Y_i, \\
    \nabla_i X_j=&-k_{ij}N+\omega_{ij},\\
    \nabla_i Y_j=&-g_{ij}N+k_{il}\omega_{jl},\\
    \nabla_i\omega_{jk}=&g_{ij}X_k-g_{ik}X_j-k_{ij}Y_k+k_{ik}Y_j.
    \end{split}
\end{equation}
\end{theorem}

\begin{proof}[Proof of Theorem \ref{Thm: N,X,Y,omega}]
Recall from Section \ref{S:preliminaries} that $N,X,Y,\omega$ are real-valued.
We compute
\begin{align*}
    \begin{split}
        \nabla_iN=&\nabla_i\langle \psi,\psi\rangle\\
        =&-\frac12\langle (k_{ij}e_je_0+\mathbf{i}e_i)\psi,\psi\rangle-\frac12\langle \psi,(k_{ij}e_je_0+\mathbf{i}e_i)\psi\rangle\\
        =&-k_{ij}X_j-Y_i.
    \end{split}
\end{align*}
Moreover,
\begin{align*}
    \begin{split}
        \nabla_iX_j=&\nabla_i\langle e_je_0\psi,\psi\rangle\\
        =&-\frac12\langle e_je_0(k_{il}e_le_0+\mathbf{i}e_i)\psi,\psi\rangle-\frac12\langle e_je_0\psi,(k_{il}e_le_0+\mathbf{i}e_i)\psi\rangle
        \\=&-k_{ij}N-\frac12i\langle e_je_0 e_i\psi,\psi\rangle +\frac12i \langle e_je_0\psi, e_i\psi\rangle\\
        =&-k_{ij}N+\omega_{ij},
    \end{split}
\end{align*}
where $\omega_{ij}=\operatorname{Im}\langle e_ie_je_0\psi,\psi\rangle$, and 
\begin{align*}
    \begin{split}
        \nabla_iY_j=&\nabla_i\langle \mathbf{i}e_j\psi,\psi\rangle
        \\ =&-\frac12\langle \mathbf{i}e_j(k_{il}e_le_0+\mathbf{i}e_i)\psi,\psi\rangle-\frac12\langle \mathbf{i}e_j\psi,(k_{il}e_le_0+\mathbf{i}e_i)\psi\rangle
        \\
        =&-g_{ij}N-\frac12k_{il}\langle \mathbf{i}e_je_le_0\psi,\psi\rangle-\frac12k_{il}\langle \mathbf{i}e_j\psi,e_le_0\psi\rangle\\
        =&-g_{ij}N+k_{il}\omega_{jl}.
    \end{split}
\end{align*}
Furthermore, for pairwise disjoint $i,j,k$,
\begin{align*}
    \begin{split}
        \nabla_i\omega_{jk}=&
        \nabla_i\operatorname{Im}\langle e_je_ke_0\psi,\psi\rangle\\
        =&-\frac12\operatorname{Im}\langle e_je_ke_0(k_{il}e_le_0+\mathbf{i}e_i)\psi,\psi\rangle
        -\frac12\operatorname{Im}\langle e_je_ke_0\psi,(k_{il}e_le_0+\mathbf{i}e_i)\psi\rangle\\
           =&-\frac12k_{il}\operatorname{Im}\langle e_je_ke_0e_le_0\psi,\psi\rangle
        -\frac12k_{il}\operatorname{Im}\langle e_je_ke_0\psi, e_le_0\psi\rangle\\
                   =&\frac12k_{il}\operatorname{Im}\langle e_je_ke_l\psi,\psi\rangle
        +\frac12k_{il}\operatorname{Im}\langle e_je_k\psi, e_l\psi\rangle\\
        =&\sum_{l\ne j,k}\left(\frac12k_{il}\operatorname{Im}\langle e_je_ke_l\psi,\psi\rangle
        +\frac12k_{il}\operatorname{Im}\langle e_je_k\psi, e_l\psi\rangle\right)\\
        &+\frac12k_{ij}\operatorname{Im}\langle e_je_ke_j\psi,\psi\rangle
        +\frac12k_{ij}\operatorname{Im}\langle e_je_k\psi, e_j\psi\rangle       \\
      &+  \frac12k_{ik}\operatorname{Im}\langle e_je_ke_k\psi,\psi\rangle
        +\frac12k_{ik}\operatorname{Im}\langle e_je_k\psi, e_k\psi\rangle\\
        =&0-k_{ij}Y_k+k_{ik}Y_j,
    \end{split}
\end{align*}
and for disjoint $i,j$ we compute (not summing over $i$)
\begin{align*}
    \begin{split}
        \nabla_i\omega_{ij}=&
        \nabla_i \operatorname{Im}\langle e_ie_je_0\psi,\psi\rangle\\
        =&-\frac12\operatorname{Im}\langle e_ie_je_0(k_{il}e_le_0+\mathbf{i}e_i)\psi,\psi\rangle
        -\frac12\operatorname{Im}\langle e_ie_je_0\psi, (k_{il}e_le_0+\mathbf{i}e_i)\psi\rangle\\
        =&\langle e_je_0\psi,\psi\rangle+\frac12k_{il}\operatorname{Im}\langle e_ie_je_l\psi,\psi\rangle
        +\frac12k_{il}\operatorname{Im}\langle e_ie_j\psi, e_l\psi\rangle\\
        =&X_j+\sum_{l\ne i,j}\left(\frac12k_{il}\operatorname{Im}\langle e_ie_je_l\psi,\psi\rangle
        +\frac12k_{il}\operatorname{Im}\langle e_ie_j\psi, e_l\psi\rangle\right)\\
        &+\frac12k_{ii}\operatorname{Im}\langle e_ie_je_i\psi,\psi\rangle
        +\frac12k_{ii}\operatorname{Im}\langle e_ie_j\psi, e_i\psi\rangle\\
        &+\frac12k_{ij}\operatorname{Im}\langle e_ie_je_j\psi,\psi\rangle
        +\frac12k_{ij}\operatorname{Im}\langle e_ie_j\psi, e_j\psi\rangle\\
        =&X_j+0-k_{ii} Y_j+k_{ij}Y_i.
    \end{split}
\end{align*}
This finishes the proof.
\end{proof}

\begin{corollary}
     Given $\psi$ as in Theorem \ref{Thm: N,X,Y,omega}, $N$ is nowhere vanishing.
\end{corollary}

\begin{proof}
    Since $|X|\le N$ and $|Y|\le N$, Theorem \ref{Thm: N,X,Y,omega} yields
    \begin{align*}
        |\nabla_i N|\le (C_k+1) N
    \end{align*}
    where $C_k$ is a constant such that $|k|\le C_k$.
    Hence, the result follows from the fact that $N$ is nonvanishing near infinity in view of Lemma \ref{decay estimate}.
\end{proof}

We also have the following Fierz-type identity:

\begin{corollary}\label{Cor: constancy}
     Given $\psi$ as in Theorem \ref{Thm: N,X,Y,omega}, the quantity
    \begin{align*}
        N^2-|X|^2-|Y|^2+\frac12|\omega|^2
    \end{align*}
    is constant.
\end{corollary}

\begin{proof}
By Theorem \ref{Thm: N,X,Y,omega} we have
    \begin{align*}
\begin{split}
    &  \frac{1}{2}\nabla_i\left(|X|^2+|Y|^2-\frac12|\omega|^2\right)-\frac{1}{2}\nabla_iN^2
    \\=&-k_{ij}X_jN+\omega_{ij}X_j-g_{ij}Y_j N+k_{il} \omega_{jl}Y_j
    \\&-\frac{1}{2}(g_{ij}X_k-g_{ik}X_j-k_{ij}Y_k+k_{ik}Y_j) \omega_{jk}+k_{ij}X_j N+Y_i N=0.
\end{split}
\end{align*}
Hence $N^2-|X|^2-|Y|^2+\frac12|\omega|^2$ is constant.
\end{proof}

\subsection{Proof of Theorem \ref{Thm Intro cannot change causal type}}

\begin{proposition}\label{Prop: zeta = 0}
 Given $\psi$ as in Theorem \ref{Thm: N,X,Y,omega}, the following holds.
Suppose that either 
\begin{align*}
    N\omega =XY-YX
\end{align*}
at a point $p\in M$, or that
\begin{align*}
    N(\psi^\infty)=|Y(\psi^\infty)|
\end{align*}
Then
    \begin{align*}
        N\omega = XY-YX
    \end{align*}
    everywhere in $M$.
\end{proposition}

\begin{proof}
Let  $\zeta_{jk}=N\omega_{jk}-(X\wedge Y)_{jk}$.
We compute using Theorem \ref{Thm: N,X,Y,omega}
\begin{align*}
    \begin{split}
    &\nabla_i\zeta_{jk}\\
       = &\nabla_i(N\omega_{jk}-X_jY_k+X_kY_j)\\
        =&N(g_{ij}X_k-g_{ik}X_j-k_{ij}Y_k+k_{ik}Y_j)-\omega_{jk}(k_{il}X_l+Y_i)\\
        &-Y_k(-k_{ij}N+\omega_{ij})+Y_j(-k_{ik}N+\omega_{ik})+X_k(-g_{ij}N+k_{il}\omega_{jl})-X_j(-g_{ik}N+k_{il}\omega_{kl})\\
             =&-\omega_{jk}(k_{il}X_l+Y_i)-Y_k\omega_{ij}+Y_j\omega_{ik}+X_kk_{il}\omega_{jl}-X_jk_{il}\omega_{kl}\\
             =&-(Y\wedge \omega)_{ijk}-k_{il}(X\wedge\omega)_{jkl}.
    \end{split}
\end{align*}
Decomposing $N\omega_{jk}=\alpha (X\wedge Y)_{jk}+\beta_{jk}$ where $\beta$ is a two-form perpendicular to $X\wedge Y$, we obtain
\begin{align*}
    N\nabla_i \zeta_{jk}=-(Y\wedge \beta)_{ijk}-k_{il}(X\wedge \beta)_{jkl}.
\end{align*}
Hence $|\zeta|=\sqrt{(1-\alpha)^2|X\wedge Y|^2+|\beta|^2}\ge|\beta|$ and
\begin{align*}
  N|  \nabla_i|\zeta|| \le 3(|Y|+|k||X|)|\beta| \le 3(|Y|+|k||X|)|\zeta|\le (3+3C_kN)|\zeta|
\end{align*}
where $C_k$ is a constant such that $|k|\le C_k$.
Suppose $\zeta\neq0$ at some point, then $\zeta\neq0 $ on all of $M^n$. 
Hence, we have already shown the first statement.

\medskip

Since $N(\psi^\infty)=|Y(\psi^\infty)|$, we have $N(\psi^\infty)\omega(\psi^\infty)=X(\psi^\infty)Y(\psi^\infty)-Y(\psi^\infty)X(\psi^\infty)$ by Proposition \ref{Prop N=X implies Nw=XY}.
Moreover, we may apply the decay estimates in Lemma \ref{NXY decay}.
Let $\gamma(s)$ be the integral curve of $\nabla y_n$ parametrized by arclength. Then $y_n^{-1}=O(e^{-s})$ and $r^{-1}=O(e^{-s})$. 
Without loss of generality we may assume that $|k|(\gamma(s))\le 0.1$.
This yields the ODE
\begin{equation*}
    \frac{d}{ds} |\zeta|(\gamma(s))\le 3.3|\zeta|(\gamma(s)).
\end{equation*}
Using the decay estimate $|\zeta|(\gamma(s))=O(e^{(-q-2)s})$ which follows from Lemma \ref{NXY decay} we obtain a contradiction. 
\end{proof}

\begin{remark}
    We have already seen by Proposition \ref{Prop N=X implies Nw=XY}, Corollary \ref{Cor: constancy} and Proposition \ref{Prop: zeta = 0} that if $N=|X|$ at a point $p$, we have
    \begin{align*}
        N^2=|X|^2+|Y|^2-\frac12|\omega|^2
    \end{align*}
    and 
    \begin{align*}
        N\omega=XY-YX
    \end{align*}
 everywhere in $M$.
A similar identity holds at a spinorial level, and one can show (with the help of Theorem \ref{Thm Intro hidden symmetry}) that if $N=|X|$ somewhere, we must have
    \begin{align*}
    (N+e_0X)(N-\mathbf{i}Y)\psi=\langle X,Y\rangle \mathbf{i}e_0\psi
\end{align*}
everywhere in $M$.
\end{remark}

\begin{proposition}\label{Prop: |X| = N=Y}
Given $\psi$ as in Theorem \ref{Thm: N,X,Y,omega}, suppose that 
    \begin{align*}
        N=|X|=|Y|
    \end{align*}
    at a point $p\in M$, or that
        \begin{align*}
        N(\psi^\infty)=|X|(\psi^\infty)=|Y|(\psi^\infty).
    \end{align*}
    Then     \begin{align*}
        N=|X|=|Y|
    \end{align*}
    and $X\perp Y$ everywhere in $M$.
Moreover, the vector fields $X,Y$ are everywhere non-vanishing.
\end{proposition}

\begin{proof}
Recall that by Proposition \ref{Prop: zeta = 0} and Corollary \ref{Cor: constancy} we have
\begin{equation}\label{key1}
N\omega_{ij}=X_iY_j-X_jY_i \quad and \quad N^2-|X|^2-|Y|^2+\frac{1}{2}|\omega|^2=0.
\end{equation}
Consequently,
\begin{equation} \label{easy0}
\begin{split}
    0=&N^4-N^2(|X|^2+|Y|^2)+\frac{1}{2}|X_iY_j-X_jY_i|^2
    \\=&N^4-N^2(|X|^2+|Y|^2)+|X|^2|Y|^2-\langle X,Y\rangle^2
    \\=&(N^2-|X|^2)(N^2-|Y|^2)-\langle X,Y\rangle^2.
\end{split}
\end{equation}
Applying equation \eqref{key1} again, we obtain 
\begin{equation} \label{easy1}
\begin{split}
    &\frac{1}{2}\nabla_i(N^2-|X|^2)
    \\=& -Nk_{ij}X_j-NY_i+k_{ij}NX_j-\omega_{ij}X_j
    \\=&-N Y_i-N^{-1}(X_i\langle X , Y\rangle-|X|^2 Y_i)
    \\=&-(N^2-|X|^2)N^{-1}Y_i-\langle X,Y\rangle N^{-1} X_i,
\end{split}
\end{equation}
and 
\begin{equation} \label{easy2}
    \begin{split}
        &\frac{1}{2}\nabla_i(N^2-|Y|^2)
        \\=& 
        -Nk_{ij}X_j-N Y_i+NY_i-k_{il}\omega_{jl}Y_j
        \\=& -Nk_{ij} X_j-N^{-1}k_{il}(\langle X,Y\rangle Y_l-X_l |Y|^2)
        \\=& -(N^2-|Y|^2)k_{ij} N^{-1}X_j-k_{il}\langle X,Y\rangle N^{-1}Y_l.
    \end{split}
\end{equation}
Let 
\begin{align*}
\Omega_1=\{p\in M: N^2-|X|^2\ge 100(N^2- |Y|^2)\;\text{at}\; p\},
\end{align*}
and $\Omega_2=M^n\setminus \Omega_1$. 
Equation \eqref{easy0} implies that $0.1(N^2-|X|^2)\ge  |\langle X, Y\rangle|$ holds on $\Omega_1$.
Therefore, using Equation \eqref{easy1} and the facts $N\ge |X|$ as well as $N\ge|Y|$, we obtain 
\begin{equation} \label{F1}
    \left|\nabla (N^2-|X|^2)\right|\le  2.2(N^2-|X|^2)
\end{equation}
in $\Omega_1$.
Similarly, assuming $|k|\le C_k$ in $M$, we have, $10(N^2-|Y|^2)\ge |\langle X, Y\rangle|$ in $\Omega_2$, and
 \begin{equation} \label{F2}
     |\nabla(N^2-|Y|^2)|\le 22C_k (N^2-|Y|^2).
 \end{equation}
Next, we denote with $F$ the function $F=\max\{N^2-|X|^2, 100(N^2-|Y|^2)\}$.
Then $F=N^2-|X|^2$ on $\Omega_1$, and $F=100(N^2-|Y|^2)$ on $\Omega_2$.
Moreover, $F$ is Lipschitz. 
Combining equation \eqref{F1} and \eqref{F2}, we obtain
 \begin{equation*} \label{F22}
     |\nabla F|\le \max\{2.2, \;22C_k\}F\quad a.e..
 \end{equation*}
 Hence, $F\equiv 0$ or $F> 0$ everywhere on $M^n$ which shows the first statement.

\medskip

 Next, we use a similar argument in Proposition \ref{Prop: zeta = 0}.
 Suppose $F\neq 0$, and recall again the notation $\gamma$  from Proposition \ref{Prop: zeta = 0}.
Without loss of generality we may assume that $|k|(\gamma(s))\le 0.1$.
Then 
\begin{equation*}
    \left|\frac{d}{ds}F(\gamma(s))\right| 
    \le 2.2 F(\gamma(s)) \quad a.e..
\end{equation*}
However, note that $y_n^{-1}=O(e^{-s})$ and $r^{-1}=O(e^{-s})$, we have $F(\gamma(s))=O(e^{(-q-2)s})$ by Lemma \ref{NXY decay}.
Hence $N=|X|=|Y|$.
Finally, equation \eqref{easy0} yields $X\perp Y$.
\end{proof}

The following lemma is the only instance where we use the assumption $n\ge4$.

\begin{lemma}\label{lemma n>=4}
    Let $n\ge4$.
    Given $\psi$ as in Theorem \ref{Thm: N,X,Y,omega}, suppose that $N(\psi)\omega(\psi)=X(\psi)Y(\psi)-Y(\psi)X(\psi)$ at a point $p\in M$.
    Then 
    \begin{align*}
        N(u)\omega(u)=X(u)Y(u)-Y(u)X(u).
    \end{align*}
  Similarly, suppose that $ N^2(\psi)=|X(\psi)|^2+|Y(\psi)|^2-\frac12|\omega (\psi)|^2$ at a point $p\in M$. Then
    \begin{align*}
        N^2(u)=|X(u)|^2+|Y(u)|^2-\frac12|\omega (u)|^2.
    \end{align*}
\end{lemma}

\begin{proof}
  Suppose that $N(\psi)\omega(\psi)=X(\psi)Y(\psi)-Y(\psi)X(\psi)$ at a point $p\in M$.
    Recall from Theorem \ref{pmt maerten} that $\psi-\psi^\infty=\mathcal O_{2}(r^{\frac{1}{2}-q})$ where $q>\frac{n}{2}\ge 2$.
    Therefore, 
    \begin{align*}
        N(\psi)-N(\psi^\infty)=&\mathcal O_2(r^{1-q}),\quad\quad X(\psi)-X(\psi^\infty)=\mathcal O_2(r^{1-q}),\\
        Y(\psi)-Y(\psi^\infty)=&\mathcal O_2(r^{1-q}),\quad\quad\omega(\psi)-\omega(\psi^\infty)=\mathcal O_2(r^{1-q}). 
    \end{align*}
Since $N(\psi)\omega(\psi)=X(\psi)Y(\psi)-Y(\psi)X(\psi)$ everywhere on $M$ by Proposition \ref{Prop: zeta = 0}, we have
\begin{align*}
    N(\psi^\infty)\omega(\psi^\infty)-X(\psi^\infty)Y(\psi^\infty)+Y(\psi^\infty)X(\psi^\infty)=\mathcal O(r^{2-q}).
\end{align*}
Since $2-q<0$, Corollary \ref{cor Nw=XY psi -> u} implies $N(u)\omega(u)=X(u)Y(u)-Y(u)X(u)$.

\medskip

Similarly, we obtain in the second case
\begin{align*}
    N(\psi^\infty)^2-|X(\psi^\infty)|^2-|Y(\psi^\infty)|^2+\frac{1}{2}|\omega(\psi^\infty)|^2=\mathcal O(r^{2-q}).
\end{align*}
    Therefore, the result follows from Corollary \ref{cor N^2-X^2-Y^2-w^2 psi -> u}.
\end{proof}

Let $e_1=X(u)|X(u)|^ {-1}$ and $e_2=|Y^\perp(u)|^{-1}Y^\perp(u)$ where $Y^\perp= Y(u)-\frac{\langle X(u),Y(u)\rangle}{|X(u)|^2}X(u)$. 
In case $\{X(u), Y(u)\}$ is linearly dependent, we instead choose two orthonormal vectors $e_1$ and $e_2$ such that $\{X(u), Y(u)\}\subset \text{Span}\{e_1,e_2\}$.
Let $e_3,e_4$ be two other unit normal vectors perpendicular to $e_1,e_2$ and define the constant spinors 
\begin{align*}
    \widetilde u=\mathbf{i}e_0e_1e_2u,\quad\quad \widehat u=e_3e_4u.
\end{align*}
Moreover, let $\widetilde\psi^\infty$ and $\widehat\psi^\infty$ be the corresponding imaginary Killing spinors in hyperbolic space.

\begin{proposition}\label{tilde u mass minimizing}
    Suppose that $N(u)\omega(u)=X(u)Y(u)-Y(u)X(u)$.
    Then
    \begin{align*}
        N(\widetilde u)=&N(\hat u)=N(u),\quad\quad X(\widetilde u)=X(\widehat u)=X(u),\\
        Y(\widetilde u)=&Y(\hat u)=Y(u),\quad\quad \omega(\widetilde u)=\omega(\widehat u)=\omega(u).
    \end{align*}
    In particular both $\tilde u$ and $\hat u$ are also mass-minimizing.
\end{proposition}

\begin{proof}
Clearly, we have $ N(\tilde u)=N(\hat u)=N(u)$.
Next, we compute
\begin{align*}
    X_i(\widetilde u)=&\langle e_ie_1e_2u,e_0e_1e_2u\rangle=X_i,\\
    Y_i(\widetilde u)=&-\langle \mathbf{i}e_ie_1e_2u,e_1e_2u\rangle=Y_i,\\
    \omega_{ij}(\widetilde u)=&-\langle \mathbf{i}e_ie_je_1e_2u,e_0e_1e_2u\rangle=\omega_{ij}
\end{align*}
where we used that $X_i(u)=0,\,Y_i(u)=0,\,\omega_{ij}(u)=0$ for $i,j>2$.
The identities for $\hat u$ are shown analogously.
\end{proof}

Next, we define $v=\frac1{\sqrt2}(u+\tilde u)$ and $\phi^\infty=\frac1{\sqrt2}(\psi^\infty+\tilde\psi^\infty)$.

\begin{proposition}\label{Prop properties v}
 Let $u\in\overline{\mathcal S}(B_1^n(0))$ and let $\widetilde u, v$ be as above.
 Then we have $N(v)=0$ if and only if $N(u)=|X(u)|=|Y(u)|$.
Moreover, if we additionally assume $N(u)\omega(u)=X(u)Y(u)-Y(u)X(u)$ and $  N^2(u)=|X(u)|^2+|Y(u)|^2-\frac12|\omega (u)|^2$, we have
\begin{align*}
    N(v)=|X(v)|=|Y(v)|.
\end{align*}
\end{proposition}

\begin{proof}
We obtain
    \begin{align*}
    N(v)=N(u)-\omega_{12}(u)=N(u)-N(u)^{-1}(|X(u)||Y^\perp(u)|).
\end{align*}
This implies that $N(v)=0$ if and only if $N(u)=|X(u)|=|Y(u)|$.

\medskip

Next, we compute 
\begin{align*}
    X_i(v)=&
   \frac12 \langle e_ie_0 (u+\mathbf{i}e_0e_1e_2u),u+\mathbf{i}e_0e_1e_2u\rangle\\
    =&X_i(u)-\operatorname{Re}\langle e_iu, \mathbf{i}e_1e_2u\rangle
    \\=& X_i(u) -Y_2(u)\delta_{1i}+Y_1(u)\delta_{2i}.
\end{align*}
Therefore,
\begin{align*}
    |X(v)|^2=|X|^2+|Y|^2-2X_1Y_2=N(v)^2
\end{align*}
where we used $N(u)\omega(u)=X(u)Y(u)-Y(u)X(u)$ and $  N^2(u)=|X(u)|^2+|Y(u)|^2-\frac12|\omega (u)|^2$.
Furthermore,
\begin{align*}
    Y_i(v)=&\frac12\langle \mathbf{i}e_i (u+\mathbf{i}e_0e_1e_2u),u+\mathbf{i}e_0e_1e_2u\rangle\\
    =&Y_i+\operatorname{Re}\langle e_iu, e_0e_1e_2u\rangle
    \\=&Y_i-X_1\delta_{i2}.
\end{align*}
Similarly, this implies $|Y(v)|^2=N(v)^2$.
\end{proof}

\begin{lemma}\label{lemma mu J X N}
  Suppose that $\phi\in\overline{\mathcal S}(M^n)$ solves $\nabla_i\phi=-\frac12k_{ij}e_je_0\phi-\frac12\mathbf{i}e_i\phi$.
  Then
  \begin{align*}
      \mu X=-NJ.
  \end{align*}
In particular, assuming additionally $\mu\ge|J|$, yields
  \begin{align*}
      \mu=|J|.
  \end{align*}
\end{lemma}

\begin{proof}
Recall the identity $\mathcal R(e_j,e_i)e_j\phi=-\frac12R_{ij}e_j\phi$ where $\mathcal R$ is the curvature operator given by $\mathcal R(e_j,e_i)\phi=(\nabla_j\nabla_i-\nabla_i\nabla_j)\phi$.
    Hence, we may compute
    \begin{align*}
    \begin{split}
        NR_{ij}=&\operatorname{Re}\langle R_{il}e_l\phi, e_j\phi\rangle\\
        =&-2 \operatorname{Re}\langle e_l(\nabla_i\nabla_l-\nabla_l\nabla_i)\phi, e_j\phi\rangle\\
        =& \operatorname{Re}\langle e_l\nabla_i(k_{lm}e_m e_0 \phi+\mathbf{i}e_l\phi)-e_l\nabla_l(k_{im}e_me_0\phi+\mathbf{i}e_i\phi), e_j\phi\rangle \\
        =&\operatorname{Re}\langle e_l\nabla_ik_{lm}e_m e_0 \phi, e_j\phi\rangle-\operatorname{Re}\langle e_l\nabla_lk_{im}e_me_0\phi, e_j\phi\rangle \\
        &-\frac12 \operatorname{Re}\langle  e_l k_{lm}e_me_0k_{in}e_ne_0\phi, e_j\phi\rangle +\frac12 \operatorname{Re}\langle e_lk_{im}e_me_0k_{ln}e_ne_0\phi,e_j\phi\rangle\\
                &-\frac12 \operatorname{Re}\langle  e_l k_{lm}e_me_0\mathbf{i}e_i\phi, e_j\phi\rangle +\frac12 \operatorname{Re}\langle e_lk_{im}e_me_0\mathbf{i}e_l\phi,e_j\phi\rangle\\
                &-\frac12\operatorname{Re}\langle  i e_le_l k_{im}e_me_0\phi,e_j\phi\rangle +\frac12\operatorname{Re}\langle   e_le_le_i\phi,e_j\phi\rangle     \\
                &+\frac12\operatorname{Re}\langle i e_le_ik_{lm}e_me_0\phi,e_j\phi\rangle-\frac12\operatorname{Re}\langle  e_le_ie_l\phi,e_j\phi\rangle\\
        =&:\text{\textbf{I}+\textbf{II}+\textbf{III}+\textbf{IV}+\textbf{V}+\textbf{VI}+\textbf{VII}+\textbf{VIII}+\textbf{IX}+\textbf{X}}.
        \end{split}
        \end{align*}
We compute
        \begin{align*}
            \begin{split}
                \text{\textbf{I}}=&\left(\sum_{l=m}+\sum_{l=j\ne m}+\sum_{m=j\ne l}  \right)\nabla_ik_{lm}\operatorname{Re}\langle e_le_m e_0 \phi, e_j\phi\rangle\\
                =&\nabla_i\tr_g(k)X_j+\nabla_ik_{jm}X_m-\nabla_ik_{jj}X_j-\nabla_ik_{mj}X_m+\nabla_ik_{jj}X_j\\
                =&\nabla_i\tr_g(k)X_j,
            \end{split}
        \end{align*}
        and
               \begin{align*}
            \begin{split}
                \text{\textbf{II}}=&-\left(\sum_{l=m}+\sum_{l=j\ne m}+\sum_{m=j\ne l}  \right)\nabla_lk_{im}\operatorname{Re}\langle e_le_me_0\phi, e_j\phi\rangle\\
                =&-\nabla_lk_{il}X_j-\nabla_jk_{im}X_m+\nabla_jk_{ij}X_j+\nabla_lk_{ij}X_l-\nabla_jk_{ij}X_j\\
                 =&-\nabla_lk_{il}X_j-\nabla_jk_{im}X_m+\nabla_lk_{ij}X_l.
            \end{split}
            \end{align*}
            Moreover,
                \begin{align*}
            \begin{split}
                \text{\textbf{III}}=&-\frac12 \operatorname{Re}\langle  e_l k_{lm}e_me_0k_{in}e_ne_0\phi, e_j\phi\rangle \\
                =&-\frac12\tr_g(k)k_{in} \operatorname{Re}\langle e_n\phi, e_j\phi\rangle \\
                =&-\frac12\tr_g(k)k_{ij}N.
                \end{split}
                \end{align*}
                where we exploited the symmetry of $k_{lm}$ and antisymmetry of $e_le_m$ to obtain $k_{lm}e_le_m=-\tr_g(k)$. 
 Similarly, we calculate
     \begin{align*}
            \begin{split}
                \text{\textbf{IV}}=&\frac12 \operatorname{Re}\langle e_lk_{im}e_me_0k_{ln}e_ne_0\phi,e_j\phi\rangle\\
                =&\frac12k_{im} \operatorname{Re}\langle (e_me_l+2\delta_{ml})k_{ln}e_n\phi,e_j\phi\rangle\\
                =&-\frac12Kk_{im} \operatorname{Re}\langle e_m\phi,e_j\phi\rangle
                +k_{il} k_{ln} \operatorname{Re}\langle e_n\phi,e_j\phi\rangle\\
                =&-\frac12\tr_g(k)k_{ij}N+k_{il}k_{lj}N,
                \end{split}
                \end{align*}
Consequently,
    \begin{align*}
     \text{\textbf{I}+\textbf{II}+\textbf{III}+\textbf{IV}} =X_j\nabla_i\tr_g(k)-X_j\nabla_lk_{il}+\nabla_Xk_{ij}-\nabla_ik_{lj}X_l+Nk_{il}k_{lj}-N\tr_g(k)k_{ij}.
    \end{align*}
Next, we compute
  \begin{align*}
            \text{\textbf{V}}=\frac12\tr_g(k)\operatorname{Re}\langle  e_0\mathbf{i}e_i\phi, e_j\phi\rangle=-\frac12\tr_g(k)\operatorname{Re}\langle  \mathbf{i}e_ie_0\phi, e_j\phi\rangle,
        \end{align*}
and
        \begin{align*}
            \text{\textbf{VI}}=&  \frac12 \operatorname{Re}\langle e_lk_{im}e_me_0\mathbf{i}e_l\phi,e_j\phi\rangle\\
           =& -\frac12 \operatorname{Re}\langle e_lk_{im}e_me_le_0\mathbf{i}\phi,e_j\phi\rangle\\
           =&-\frac n2\operatorname{Re}\langle k_{im}e_me_0 \mathbf{i}\phi,e_j\phi\rangle+\operatorname{Re}\langle e_m k_{im}e_0 \mathbf{i}\phi,e_j\phi\rangle,
        \end{align*}
as well as
\begin{align*}
  \text{\textbf{VII}}=&  \frac n2\operatorname{Re}\langle    k_{im}e_me_0i\phi,e_j\phi\rangle.
\end{align*}
Finally, we have
     \begin{align*}
           \text{\textbf{VIII}}=-\frac n2g_{ij}N,
     \end{align*}   
     and
             \begin{align*}
       \text{\textbf{IX}}  =&  \frac12\operatorname{Re}\langle \mathbf{i} e_le_ik_{lm}e_me_0\phi,e_j\phi\rangle\\
       =&\frac12\tr_g(k)\operatorname{Re}\langle \mathbf{i}e_ie_0\phi,e_j\phi\rangle-\operatorname{Re}\langle \mathbf{i} k_{im}e_me_0\phi,e_j\phi\rangle,
        \end{align*}
as well as 
     \begin{align*}
         \text{\textbf{X}}=-\frac{n-2}{2}g_{ij}N.
     \end{align*}
     Therefore, we obtain
\begin{align*}
    \text{\textbf{V}+\textbf{VI}+\textbf{VII}+\textbf{VIII}+\textbf{IX}+\textbf{X}}=&-(n-1)Ng_{ij}.
\end{align*}
Hence, we may conclude
\begin{align*}
    NR_{ij}=&X_j\nabla_i\tr_g(k)-X_j\nabla_lk_{il}+\nabla_Xk_{ij}-\nabla_ik_{lj}X_l+Nk_{il}k_{lj}-N\tr_g(k)k_{ij}-(n-1)Ng_{ij}.
\end{align*}
Taking the trace implies
\begin{align*}
    2N\mu=&-2\langle J,X\rangle.
\end{align*}
Since $N>0$, $N\ge|X|$ and $\mu\ge|J|$, the result follows.
\end{proof}

\begin{proof}[Proof of Theorem \ref{thm causal character} for $n\ge4$]
    Suppose that there exists a point $p\in M$ such that $N=|X|$ at $p$.
    Using Proposition \ref{Prop N=X implies Nw=XY} and Lemma \ref{lemma n>=4}, we obtain
    \begin{align*}
        N(u)\omega(u)=X(u)Y(u)-Y(u)X(u),\quad\quad
        N^2(u)=|X(u)|^2+|Y(u)|^2-\frac12|\omega (u)|^2.
    \end{align*}
    If $N(u)=|Y(u)|=|X(u)|$, we can use Corollary \ref{cor N=Y(u) implies N=Y(psi)} and Proposition \ref{Prop: |X| = N=Y} to obtain that $N=|X|$ everywhere.
    Thus, let us assume that $N(u)=|Y(u)|=|X(u)|$ does not hold.
    Then we can use Proposition \ref{Prop properties v} to construct a constant spinor $v=\frac1{\sqrt2}(u+\widetilde u)$ with $N(v)=|X(v)|=|Y(v)|\ne0$ and with imaginary Killing spinor $\phi^\infty=\frac1{\sqrt2}(\psi^\infty+\widetilde \psi^\infty)$.
    Recall that $u,\widetilde u $ give rise to the imaginary Killing spinors $\psi^\infty $ and $\widetilde\psi^\infty$.
    Moreover, there are spinors $\psi,\widetilde\psi$ solving $\nabla_i\psi=-\frac12k_{ij}e_je_0\psi-\frac12\mathbf{i}e_i\psi$ and $\nabla_i\widetilde \psi=-\frac12k_{ij}e_je_0\widetilde\psi-\frac12\mathbf{i}e_i\widetilde\psi$ since both $u$ and $v$ are mass-minimizing, cf. Proposition \ref{tilde u mass minimizing}.
    By linearity, the spinor $\phi=\frac1{\sqrt2}(\psi+\widetilde \psi)$ also solves $\nabla_i\phi=-\frac12k_{ij}e_je_0\phi-\frac12\mathbf{i}e_i\phi$ and asymptotes to $\phi^\infty$.
    Since $N(\phi^\infty)=|X(\phi^\infty)|=|Y(\phi^\infty)|$, we can apply Proposition \ref{Prop: |X| = N=Y} to find that $N(\phi)=|X(\phi)|=|Y(\phi)|$ everywhere.
    Consequently, $\phi$ is a null spinor and does not change causal character.

    \medskip

    It remains to show that $v$ is mass-minimizing.
    This follows from the integral formula, Theorem \ref{pmt maerten}, together with the identities $\nabla_i\phi=-\frac12k_{ij}e_je_0\phi-\frac12\mathbf{i}e_i\phi$ and $(\mu+\frac12n(n-1))N=-\langle J,X\rangle$, cf. Lemma \ref{lemma mu J X N}.
\end{proof}

\subsection{Spinors of mixed causal type in hyperbolic space}

The following has been observed in \cite{Leitner2003_JMP} by Leitner.
Since no proof was provided, we include one here.

\begin{proposition}\label{Prop: existence spinor of mixed type}
    There exists a spinor $u\in\overline{\mathcal S}(\mathbb H^n)$ of mixed causal type.
\end{proposition}

\begin{proof}
Let $E_1,\dots, E_n$ be a parallel orthonormal frame on $\mathbb{R}^n$. Let $u_1$ be a parallel spinor in $\mathcal{S}(\mathbb{R}^n)$ satisfying $\mathbf{i}e_1u_1=u_1$ and $\|u_1\|^2=\frac{1}{2}$. Thus, $\langle \mathbf{i}e_j u_1,u_1\rangle=0$ for $j>1$. Let $u=(u_1,\mathbf{i} u_1)$. Then 
\begin{align*}
N(u)=&\|u_1\|^2+\|\mathbf{i}u_1\|^2=1
\\
X_j(u)=&\langle E_j(\mathbf{i}u_1,u_1),  (u_1,\mathbf{i}u_1)\rangle=2\langle \mathbf{i}e_j u_1,u_1\rangle=\delta_{1j},  
\\    Y_j(u)=&\langle \mathbf{i}E_j(u_1,iu_1),(u_1,iu_1)\rangle=0,\quad \forall j=1,\dots, n.
\end{align*}
Since $N(u)=|X(u)|$, we have $N(u)\omega(u)=X(u)\wedge Y(u)$. Hence, $\omega(u)=0$.
Next, applying Lemma \ref{lemma psi u N X Y omega}, denote $e_\alpha$ as the direction perpendicular to $e_{\varrho}$, we have 
\begin{equation*}
    N(\psi^\infty)=\frac{2(1+\varrho^2)}{1-\varrho^2},\quad
 Y_{\varrho}(\psi^\infty)=-\frac{4\varrho}{1-\varrho^2},\quad Y_\alpha(\psi^\infty)=0,
\end{equation*}
\begin{equation*}
    X_{\varrho}(\psi^\infty)=2X_\varrho(u), \quad X_\alpha(\psi^\infty)=\frac{2(1+\varrho^2)}{1-\varrho^2}X_\alpha(u). 
\end{equation*}
Therefore, $N(\psi^\infty)=|X(\psi^\infty)|$ if and only if $e_{\varrho}\perp E_1$.
\end{proof}


\section{Hidden symmetries of timelike imaginary Killing spinors}

\subsection{Hidden symmetries of imaginary Killing spinors}

Throughout this section let $(M^n,g,k)$ be an asymptotically AdS initial data set.
We first need the technical preliminary lemma which shows that in the timelike case we can automatically construct a second spinor solving $\nabla_i\psi=-\frac12k_{ij}e_je_0\psi-\frac12\mathbf{i}e_i\psi$.
We restate Theorem \ref{Thm Intro hidden symmetry} for convenience:

\begin{theorem}\label{lemma construct new timelike spinor}
    Suppose that $\psi\in \overline{\mathcal S}(M) $ solves $\nabla_i\psi=-\frac12k_{ij}e_je_0\psi-\frac12\mathbf{i}e_i\psi$.
    Then $\varphi$ defined by 
    \begin{align*}
        \varphi=e_0(N(\psi)-\mathbf{i}Y(\psi)+e_0X(\psi)-\frac12\mathbf{i}e_0\omega(\psi))\psi
    \end{align*}
    also solves
    \begin{align*}
        \nabla_i\varphi=-\frac12k_{ij}e_je_0\varphi-\frac12 \mathbf{i}e_i\varphi.
    \end{align*}
\end{theorem}

\begin{proof}
    We obtain
    \begin{align*}
        \begin{split}
            &\nabla_i(N-\mathbf{i}Y+e_0X-\frac12\mathbf{i}e_0\omega)\psi\\
            =&-k_{ij}X_j\psi-Y_i\psi+\mathbf{i}Ne_i\psi-\mathbf{i}k_{il}\omega_{jl}e_j\psi   + Nk_{ij}e_j e_0\psi-\omega_{ij}e_je_0\psi \\
&-\frac12\mathbf{i}e_0(g_{ij}X_k-g_{ik}X_j-k_{ij}Y_k+k_{ik}Y_j)e_je_k\psi\\
            &-\frac12(N-\mathbf{i}Y+e_0X-\frac12\mathbf{i}e_0\omega)(k_{ij}e_je_0+\mathbf{i}e_i)\psi\\
            =&\frac12(k_{ij}e_je_0+\mathbf{i}e_i)(N-\mathbf{i}Y+e_0X-\frac12\mathbf{i}e_0\omega)\psi
        \end{split}
    \end{align*}
    where we wrote $N,X,Y,\omega$ instead of $N(\psi),X(\psi),Y(\psi),\omega(\psi)$ to simplify the notation.
    To verify the last equation holds, we check term by term.
    We compute the identities
\begin{align*}
\begin{split}
    Nk_{ij}e_je_0\psi-\frac{1}{2}Nk_{ij}e_je_0\psi=&\frac{1}{2}N k_{ij}e_j e_0\psi,\\
    -k_{ij}X_j\psi-\frac{1}{2}e_0 X k_{ij}e_j e_0\psi=&\frac{1}{2}k_{ij}e_jX\psi,\\
    iNe_i\psi-\frac{1}{2}N\mathbf{i}e_i\psi=&\frac{1}{2}\mathbf{i}e_i N\psi,\\
    -Y_i\psi-\frac{1}{2}Ye_i\psi=&\frac{1}{2}e_i Y\psi,
    \end{split}
\end{align*}
and
\begin{align*}
\begin{split}
  &-\frac{1}{2}\mathbf{i}e_0 (-k_{ij} Y_k+k_{ik}Y_j)e_j e_k\psi +\frac{1}{2} \mathbf{i} Y k_{ij}e_j e_0\psi
  \\=& -\frac{1}{2}\mathbf{i} k_{ij}e_j e_0 Y\psi-\frac{1}{2}\mathbf{i}Y k_{ik}e_k e_0\psi+\frac{1}{2} \mathbf{i} Y k_{ij}e_j e_0 \psi 
  \\=&-\frac{1}{2}\mathbf{i} k_{ij}e_j e_0 Y\psi.
  \end{split}
\end{align*}
Using $\omega e_j\psi=\omega_{kl}e_k(-2\delta_{lj}-e_je_l)\psi=e_j\omega-4\omega_{kj}e_k \psi$, we find
\begin{equation*}
    -\mathbf{i}k_{il}\omega_{jl}e_j\psi+\frac{1}{4}\mathbf{i}e_0\omega k_{ij}e_je_0\psi= -\mathbf{i}k_{ij}\omega_{lj}e_l\psi-\frac{1}{4}i k_{ij}\omega e_j\psi=-\frac{1}{4}\mathbf{i}k_{ij}e_j\omega\psi.
\end{equation*}
Finally, we have
\begin{equation*}
    -\frac{1}{2}\mathbf{i}e_0(g_{ij}X_k-g_{ik}X_j)e_je_k\psi-\frac{1}{2}e_0X\mathbf{i}e_i\psi=-\frac{1}{2}\mathbf{i}e_0(e_iX-Xe_i+Xe_i)\psi=\frac{1}{2}\mathbf{i}e_ie_0 X\psi
\end{equation*}
and
\begin{equation*}
    -\omega_{ij}e_j e_0\psi-\frac{1}{4}e_0\omega e_i\psi=e_0(\omega_{ij}e_j-\frac{1}{4}\omega e_i)\psi=-\frac{1}{4}e_0e_i \omega\psi=\frac{1}{4}e_ie_0 \omega\psi
\end{equation*}
which finishes the proof.
\end{proof}

\subsection{The Killing development}

\begin{theorem}
    Suppose that $(M^n,g,k)$ admits a spinor $\psi$ solving $$\nabla_i\psi=-\frac12k_{ij}e_je_0\psi-\frac12 \mathbf{i}e_i\psi.$$
    Then there exists a Lorentzian manifold $(\overline M^{n+1},\overline g)$ such that $(M^n,g)$ isometrically embeds into $(\overline M^{n+1},\overline g)$ with second fundamental form $k$.
    Moreover, $(\overline M^{n+1},\overline g)$ admits an imaginary Killing spinor $\overline \psi$, i.e. $\overline \psi$ solves $$\overline \nabla_\alpha\overline \psi=-\frac12\mathbf{i}e_\alpha \overline{\psi},\quad\alpha=0,1,\dots,n.$$
    In particular, when $N>|X|$, $(\overline M^{n+1},\overline g)$ is a vacuum spacetime. 
\end{theorem}

Constructing the so-called \emph{Killing development} $(\overline M^{n+1},\overline g)$ is standard and the main challenge is to extend $\psi$ to an imaginary Killing spinor. 
To do so, Theorem \ref{lemma construct new timelike spinor} is crucial.

\begin{proof}
Let $\overline M^{n+1}=M\times[-a,a]$ and let $\tau$ be the coordinate function on $[-a,a]$.
We extend $\psi$ to $\overline M^{n+1}$ by solving 
\begin{align*}
    \partial_\tau\overline\psi(\tau)=&-\frac12(\mathbf{i}e_0N+e_0Y+\mathbf{i}X+\frac12\omega)\overline{\psi}
\end{align*}
with $\overline \psi(0)=\psi$, and where $N=N(\overline \psi(\tau)),X=X(\overline \psi(\tau)),Y=Y(\overline\psi(\tau)),\omega=\omega(\overline\psi(\tau))$.
By Lemma \ref{lemma construct new timelike spinor} and the linearity of the underlying equation, the extended spinor $\overline\psi$ still solves $\nabla_i\overline\psi=-\frac12k_{ij}e_je_0\overline\psi-\frac12\mathbf{i}e_i\overline\psi$ for $e_i\in TM$.
Moreover, by choosing $a$ sufficiently small, we can ensure $\overline\psi$ stays bounded.
Below, we will show that we can actually set $a=\infty$.

\medskip

We define on $\overline M^{n+1}$ the Lorentzian metric
\begin{align}\label{eq: Killing development}
      \overline{g}=-N^2d\tau^2+g_{ij}(dx^i+X^i d\tau)(dx^j+X^j d\tau).
\end{align}
Since $\nabla_iX_j+\nabla_jX_i=-2k_{ij}N$, the second fundamental form of $(M,g)\subseteq(\overline M^{n+1},\overline g)$ is given by $k$, cf. \cite[Lemma B1]{Huang-Lee:2024}.
Next, we show that $\overline\psi$ is an imaginary Killing spinor in $(\overline M^{n+1},\overline g)$.
Also, note that $e_0=N^{-1}(\partial_\tau-X)$.

\medskip

For this purpose we need to prove that 
\begin{equation*}
     \overline{\nabla}_{i}\overline\psi=-\frac12\mathbf{i}e_i\overline\psi 
     \quad\quad
     \text{and}\quad 
\quad     \overline{\nabla}_{0}\overline\psi=-\frac12\mathbf{i}e_0\overline\psi,
\end{equation*}
where $e_0$ is the unit normal of $M$ in $ \overline M^{n+1}$.

\medskip

To show the first identity, we directly compute
\begin{align*}
    \overline \nabla_i\overline\psi=&\nabla_i\overline\psi-\frac14\langle \overline \nabla_i e_j,e_k\rangle e_je_k\overline\psi\\
    &-\frac14\langle \overline \nabla_i e_j,e_0\rangle e_je_0\overline\psi-\frac14\langle \overline \nabla_i e_0,e_j\rangle e_0e_j\overline\psi-\frac14\langle \overline \nabla_i e_0,e_0\rangle e_0e_0\overline\psi\\
    =&-\frac12k_{ij}e_je_0\overline\psi-\frac12\mathbf{i}e_i\overline\psi-0+\frac14k_{ij}e_je_0\overline\psi+\frac14k_{ij}e_je_0\overline\psi-0\\
    =&-\frac12\mathbf{i}e_i\overline\psi.
\end{align*}
Before we prove the second identity, we first calculate with the help of the gradient equation $\nabla_i X_j=-k_{ij} f+\omega_{ij}$ and the Koszul formula
\begin{align*}
    \langle\overline{\nabla}_0 e_i, e_j\rangle=&-\frac{1}{2}\left(\langle[e_0,e_j],e_i\rangle+\langle[e_i,e_0],e_j\rangle+\langle[e_i,e_j],e_0\rangle\right)
    \\=&-\frac{1}{2}N^{-1}(\langle [e_j,X],e_i\rangle+\langle [X,e_i],e_j\rangle)
    \\=& -\frac{1}{2}N^{-1}(\nabla_j X_i-\nabla_i X_j)\\
    =&N^{-1}\omega_{ij}
\end{align*}
where we also used that 
\begin{align*}
    [e_0,e_j]=N^{-1}[Ne_0,e_j]+N^{-1}\nabla_jNe_0=N^{-1}[e_j,X]+N^{-1}\nabla_jNe_0.
\end{align*}
Moreover, we have 
\begin{align*}
    \langle \overline{\nabla}_0 e_i, e_0\rangle=&\langle[e_0,e_i], e_0\rangle
    \\=&\langle N^{-1}[e_i,X]+N^{-1}\nabla_iNe_0,e_0\rangle
    \\=&-N^{-1}\nabla_iN.
\end{align*}
Consequently,
\begin{align*}
    \overline \nabla_0\overline\psi=&N^{-1}(\partial_\tau-\nabla_X)\overline\psi-\frac14\langle \overline \nabla_0 e_j,e_k\rangle e_je_k\overline\psi\\
    &-\frac14\langle \overline \nabla_0 e_j,e_0\rangle e_je_0\overline\psi-\frac14\langle \overline \nabla_0 e_0,e_j\rangle e_0e_j\overline\psi-\frac14\langle \overline \nabla_0 e_0,e_0\rangle e_0e_0\overline\psi\\
    =& N^{-1}(\partial_\tau-\nabla_X)(\overline\psi)-\frac{1}{2}N^{-1}\nabla_i N e_0e_i\overline\psi-\frac{1}{4}N^{-1}\omega_{ij}e_je_i\overline\psi.
\end{align*}
Using that
\begin{align*}
     \nabla_X\overline\psi=-\frac12iX\overline\psi-\frac{1}{2}k_{ij}X_ie_je_0\overline\psi=-\frac12iX\overline\psi+\frac{1}{2}(\nabla_j N)e_j e_0\overline\psi+\frac12Y_je_0\overline\psi
\end{align*}
and the evolution equation for $\partial_\tau\overline\psi$, we obtain
\begin{align*}
     2N \overline{\nabla}_{e_0} \overline\psi=&-(\mathbf{i}e_0N+e_0Y+\mathbf{i}X+\frac12\omega)\overline\psi\\
     &-(-\mathbf{i}X+(\nabla_j N)e_j e_0+Y_je_0)\overline\psi-\nabla_i N e_0e_i\overline\psi-\frac{1}{2}\omega_{ij}e_je_i\overline\psi\\
     =&-\mathbf{i}Ne_0\psi
\end{align*}
as desired.
Therefore, $\overline\psi$ is indeed an imaginary Killing spinor.

\medskip

For an imaginary Killing spinor, we have $|\overline\nabla \;\overline\psi|\le \frac12|\overline\psi|$.
Therefore, $\overline \psi$ cannot blow up in finite time $\tau$ and we can choose $a=\infty$, i.e. $\overline M^{n+1}=M\times\mathbb R$.

\medskip

Since $({\overline M}^{n+1},\overline g)$ admits a timelike imaginary Killing spinor, we automatically obtain that $(\overline M^{n+1},\overline g)$ is an Einstein manifold, i.e. $(\overline M^{n+1},\overline g)$ is vacuum.
\end{proof}

We remark that the above argument also works in the null case.
In this setting we have
  \begin{align*}
        \varphi=e_0(N(\psi)-\mathbf{i}Y(\psi)+e_0X(\psi)-\frac12\mathbf{i}e_0\omega(\psi))\psi=0
    \end{align*}
and $\psi$ is extended trivially to $\overline M^{n+1}$, i.e. $\partial_\tau\overline\psi=0$.

\medskip

It would be of interest to further classify the asymptotically AdS Lorentzian manifolds admitting timelike imaginary Killing spinors. 
This will be discussed further in Section \ref{S:ultraspinning}.


\section{Spinorial 2-slicings} \label{S:main}

In this section we will prove Theorem \ref{Thm Intro minimal slicing}.
We will reuse a lot of the tools established in Section \ref{S:causality}.

\subsection{Construction of the codimension $2$ foliation}

In the asymptotically flat setting one obtains in the case of equality of the positive mass theorem a spinor $\psi$ satisfying the equation 
\begin{align}\label{eq:AF spinor}
    \nabla_i\psi=-\frac12k_{ij}e_je_0\psi.
\end{align}
Defining the vector field $X=\langle e_ie_0\psi,\psi\rangle$, it is then easy to see that $dX=0$ and that $X$ gives rise to a foliation $\Sigma_t$ for $M$.
Moreover, equation \eqref{eq:AF spinor} implies that $N^{-\frac12}\psi$ is parallel on $\Sigma_t$ which yields in combination with the abundance of spinors satisfying \eqref{eq:AF spinor} flatness of $\Sigma_t$.

\medskip

When $(M^n,g,k)$ is asymptotic to AdS at infinity the situation becomes much more complicated.
In particular, we need to use the spinor to slice down two dimensions and construct a codimension $2$ foliation.

\medskip

Throughout this section we assume that $N(\psi^\infty)=|X(\psi^\infty)|=|Y(\psi^\infty)|$.
Recall that this implies $N(\psi)=|X(\psi)|=|Y(\psi)|$ according to Proposition \ref{Prop: |X| = N=Y}.

\begin{corollary}\label{Cor: existence u1}
There exists a global foliation $\{\Sigma^1_\alpha\}$ of $M^n$ where $X$ is perpendicular to each $\Sigma^1_\alpha$.
\end{corollary}

\begin{proof}
Since $\zeta=0$, we have $N\omega = X\wedge Y$.
Moreover $dX=2\omega$, which implies that $dX\wedge X=0$.
Hence, the result follows from Frobenius' theorem.
More precisely, let $D$ be the rank $n-1$ distribution on $M$ that annihilates the 1-form $X$, then it also annihilates $dX$ because $dX=2N^{-1}X\wedge Y$. Therefore, $D$ is involutive, thus, it is integrable, i.e., locally, $D$ is the tangent bundle of a hypersurface. 
\end{proof}

To denote coordinates on a leaf $\Sigma^1_\alpha$ of the above foliation we use Greek indices.

\begin{proposition}
On each $\Sigma^1_\alpha$ there exists a foliation $\Sigma^2_s$ such that $Y$ is normal to each $\Sigma^2_s$.
\end{proposition}

\begin{proof}
Observe that $N\omega =XY-YX$ implies
\begin{align*}
       \nabla_\alpha Y_\beta=&-g_{\alpha\beta}N-k_{\alpha \nu}Y_\beta.
\end{align*}
Moreover, we have $\nabla^{\Sigma_1}_\alpha Y=\nabla_\alpha Y$ since $Y$ is tangential to $\Sigma_1$.
Hence,
\begin{align*}
    \nabla_\alpha Y_\beta-\nabla_\beta Y_\alpha= k_{\alpha \nu} Y_\beta-k_{\beta\nu} Y_\alpha.
\end{align*}
Hence $d_{\Sigma^1} Y\wedge Y=0$.
The desired result follows along the lines of the proof of Corollary \ref{Cor: existence u1}.
\end{proof}

To denote coordinates on $\Sigma^2_s$ we use capital indices.

\begin{lemma}\label{Lemma different spinors}
    Let $\psi,\overline\psi$ be two spinors asymptotic to $\psi_\infty,\overline\psi_\infty$ solving the PDEs $\nabla_i\psi=-\frac12k_{ij}e_je_0\psi-\frac12\mathbf{i}e_i\psi$ and $\nabla_i\overline\psi=-\frac12k_{ij}e_je_0\overline \psi-\frac12\mathbf{i}e_i\overline\psi$ such that $N_\infty, X_\infty, Y_\infty, \omega_\infty$ equal $\overline N_\infty,\overline X_\infty,\overline Y_\infty,\overline \omega_\infty$. Then
    \begin{align*}
        N=\overline N,\quad X=\overline X, \quad Y=\overline Y,\quad \omega=\overline \omega.
    \end{align*}
    In particular, the foliations $\Sigma^1_\alpha$ and $\Sigma^2_s$ agree and are independent of the choice of asymptotic spinor.
\end{lemma}

\begin{proof}
    Define the function
    \begin{align*}
        \Xi=|N-\overline N|^2+|X-\overline X|^2+|Y-\overline Y|^2 +|\omega-\overline\omega|^2.
    \end{align*}
    Then we have for any point where $|k|\le 0.1$
    \begin{align*}
    \begin{split}
       \frac12 |\nabla_i \Xi|\le& |N-\overline N|(0.1|\overline X-X|+|\overline Y-Y|)+|X-\overline X|(0.1|N-\overline N|+|\omega -\overline \omega|)\\
       &+|\overline Y-Y|(|N-\overline N|+0.1|\omega -\overline\omega|)+|\omega-\overline \omega|(|X-\overline X|+0.1|Y-\overline Y|)\\
       \le& 1.1 \Xi.
       \end{split}
    \end{align*}
    where we used Theorem \ref{Thm: N,X,Y,omega} and Cauchy-Schwarz inequality repeatedly.
    Integrating this along the integral curve of $\nabla y_n$ shows $\Xi=0$ since $\Xi=O(e^{(-q-2)s})$ by Lemma \ref{NXY decay}.
\end{proof}

\subsection{Geometry and topology of the foliation}

We define $\mathbf x=X|X|^{-1}$ and $\mathbf y=Y|Y|^{-1}$.
Then $\mathbf x,\mathbf y$ are well-defined since $X,Y\ne0$.

\begin{proposition}
    The second fundamental form $h$ of each leaf $\Sigma^1_\alpha$ satisfies
    \begin{align*}
        h^{\Sigma_1}_{\alpha\beta}=-k_{\alpha\beta}
    \end{align*}
\end{proposition}

\begin{proof}
    We compute
       \begin{align*}
        h^{\Sigma^1}_{\alpha\beta}=|X|^{-1}\nabla_\alpha X_\beta=-k_{\alpha\beta}
    \end{align*}
    where we used that $\omega_{\alpha\beta}=0$, $X_\beta=0$ and $|X|=N$ from Proposition \ref{Prop: |X| = N=Y} and Proposition \ref{Prop: zeta = 0}.
\end{proof}

\begin{lemma}\label{Lemma psi1 identity on Sigma1}
   On each $\Sigma^1_\alpha$ we have
       \begin{align*}
    \begin{split}
        \nabla^{\Sigma^1}_{\alpha}\phi         =&\frac12  e_\alpha\mathbf{y} \phi +\frac12\mathbf{y}_\alpha \phi.
        \end{split}
    \end{align*}
    for $\phi=N^{-\frac12}\psi_1$ in odd dimensions and $\phi = N^{-\frac12} (1+\sigma)\psi_1$ in even dimensions where  $\sigma=\mathbf{i}^\frac{n}{2}e_1\cdots e_n$.
\end{lemma}

\begin{proof}
If the dimension $n$ is odd, then the spinor bundle on $M^n$ can be identified with the spinor bundle on the level sets $\Sigma_1$.
Recall that $\psi=(\psi_1,\psi_2)$ and $\nabla_i\psi=-\frac12k_{ij}e_je_0\psi-\frac12\mathbf{i}e_i\psi$. 
Since $N=|X|$ we obtain similar to the proof of Theorem 4.1 in \cite{HZ24} that $\psi_1=\mathbf {x} \psi_2$, $\psi_2=-\mathbf {x} \psi_1$. 
Furthermore, $N=|Y|$ implies that $\mathbf{y}\psi_1=-\mathbf{i}\psi_1$.
Combining both statements, we obtain $\nabla_i\psi_1=\frac12k_{ij}e_j\mathbf {x} \psi_1+\frac12 e_i\mathbf{y}\psi_1$.
Next, we compute
    \begin{align*}
        \nabla^{\Sigma^1}_\alpha \psi_1=
        \nabla_\alpha \psi_1+\frac{1}{2}h_{\alpha\beta} e_\beta \mathbf {x}  \psi _1
        =-\frac12k_{\alpha \mathbf {x} }\psi_1+\frac12 e_\alpha\mathbf{y}\psi_1.
    \end{align*}
    Hence,
    \begin{align*}
    \begin{split}
        \nabla^{\Sigma_1}_{\alpha}(N^{-\frac12}\psi_1) = &\frac12 N^{-\frac32}(k_{\alpha j}X_j+Y_\alpha)\psi_1-\frac12N^{-\frac12}(k_{\alpha \mathbf {x} }\psi_1- e_\alpha\mathbf{y}\psi)\\
        =&\frac12  e_\alpha\mathbf{y} (N^{-\frac12} \psi_1) +\frac12N^{-1}Y_\alpha (N^{-\frac12}\psi_1)
        \end{split}
    \end{align*}
    which finishes the proof in this case. 

    \medskip

 When $n$ is even, the spinor bundle over the level sets $\Sigma_1$ can be identified with an eigenspace of the involution $\sigma: \mathcal{S}(M^n)\to \mathcal{S}(M^n)$ introduced earlier (cf. \cite[Page 903]{harmonicspinor}). This operator satisfies $\sigma^2=1$, commutes with the even part $\operatorname{Cl}^0(\mathbb{R}^n)$ of the Clifford algebra, and anticommutes with its odd part.
These properties induce a decomposition of the total spinor bundle into the $\pm 1$-eigenspaces of $\sigma:
\mathcal{S}(M^n)=\mathcal{S}^+(M^n)\oplus \mathcal{S}^-(M^n)$.
The map $e_\alpha \mapsto e_\alpha e_n$ induces an isomorphism $\operatorname{Cl}(\mathbb{R}^{n-1}) \to \operatorname{Cl}^0(\mathbb{R}^n)$, which consequently gives $\mathcal{S}^+(M^n)$ the structure of a $\operatorname{Cl}(\mathbb{R}^{n-1})$-module. Furthermore, we have the identifications
\[
\mathcal{S}({\Sigma_1})\cong \mathcal{S}^+(M^n)|_{\Sigma_1}\cong \mathcal{S}^-(M^n)|_{\Sigma_1},
\]
while the odd part $\operatorname{Cl}^1(\mathbb{R}^n)$ interchanges the two subbundles, mapping $\mathcal{S}^+(M^n)$ to $\mathcal{S}^-(M^n)$.
By convention, we work with the $+1$-eigenspace $\mathcal{S}^+(M^n)$. For any $\tilde{\psi} \in \mathcal{S}(M^n)$, the projection $(1+\sigma)\tilde{\psi}$ lies in this subbundle, since $\sigma(1+\sigma)\tilde{\psi} = (1+\sigma)\tilde{\psi}$.
 Using $\nabla \sigma=\sigma \nabla $ and $e_\alpha\mathbf{y}\sigma=\sigma e_\alpha\mathbf{y}$, we have
\begin{equation*}
    \nabla_i(1+\sigma)\psi_1=\frac{1}{2}k_{ij}e_j \mathbf {x} (1+\sigma)\psi_1+\frac12 e_i\mathbf{y}(1+\sigma)\psi_1.
\end{equation*}
Hence, we may carry through the same computation as above.
\end{proof}

\begin{proposition} \label{sigma2}
    The surfaces $\Sigma^2_s$ are flat and umbilic within $\Sigma^1_\alpha$.
\end{proposition}

\begin{proof}
The umbilicity follows from
\begin{equation} \label{h2AB}
   h^{\Sigma^2}_{AB}= \nabla^{\Sigma^1}_A (Y_B |Y|^{-1}_{g_{\Sigma^1}})= -g^{\Sigma^1}_{AB}
\end{equation}
where we used that $|Y|_{g_{\Sigma^1}}=|Y|=N$, Theorem \ref{Thm: N,X,Y,omega}, and the fact that $k_{Al}\omega_{lB}=0$ which follows from Proposition \ref{Prop: zeta = 0}.
Next, we construct a single parallel spinor.
To see this, we compute
    \begin{align*}
        \nabla_A^{\Sigma^2}\phi= \nabla^{\Sigma^1}_A\phi +\frac12 h_{AB}^{\Sigma_2} e_B \mathbf y \phi=   0
    \end{align*}
     where we used Lemma \ref{Lemma psi1 identity on Sigma1}.
     Hence, we have a single parallel spinor on $\Sigma^2$ which implies Ricci flatness.
     To obtain Riemann flatness we need more parallel spinors.
     This follows along the lines of \cite{HZ24}: 
     
     \medskip
     Let $u \in \overline{\mathcal{S}}(\mathbb{R}^n)$ be the unit-length parallel spinor corresponding to $\psi^\infty$, and define $\epsilon_1=X(u)$ and $\epsilon_2=Y(u)$. Consider any other unit-length parallel spinor $\widehat{u}\in \overline{\mathcal{S}}(\mathbb{R}^n)$ that satisfies the conditions 
     \begin{equation} \label{hat u}
     \epsilon_1\epsilon_0 \widehat{u}=\widehat{u} \quad \text{and} \quad
     \mathbf{i}\epsilon_2 \widehat{u}=\widehat{u}.
     \end{equation}
     Such a spinor inherits the key invariants of $u$: $N(\widehat{u})=N(u)$, $X(\widehat{u})=X(u)$, $Y(\widehat{u})=Y(u)$, and $\omega(\widehat{u})=\omega(u)$. 
     Thus, $\widehat{u}$ gives rise a spinor $\widehat{\psi}$ on $M$ possessing the same $N$, $X$, $Y$, and $\omega$ as the original $\psi$. In light of Lemma \ref{Lemma different spinors}, this implies $\widehat{\psi}$ generates the same foliations $\Sigma^1_t$ and $\Sigma^2_s$ as $\psi$.
Moreover, we obtain $2^{\lfloor\frac{n}{2}\rfloor-1}$ linearly independent spinors solving equation \eqref{hat u} which coincides with $\dim \mathcal{S}(\mathbb{R}^{n-2})$.
Hence, $\Sigma^2_s$ are flat.
\end{proof}

\begin{theorem}\label{Sigma1}
 The vector field  $\mathbf{y}$ is integrable on $\Sigma^1$.  Moreover, the surfaces $\Sigma^1$ are hyperbolic.
\end{theorem}

\begin{proof}
Since $\nabla_A \mathbf{y}_B=-g^{\Sigma^2}_{AB}$ by equation \eqref{h2AB} and $\nabla_\mathbf{y} \mathbf{y}_A=0$ by equation \eqref{NXYw}, we have $\nabla^{\Sigma^1}_\alpha \mathbf{y}_\beta=\nabla^{\Sigma^1}_\beta \mathbf{y}_\alpha$. Thus, $\mathbf{y}$ is integrable on $\Sigma^1$. Hence, the foliation $\{\Sigma_s^2\}$ is generated by the unit normal flow from a leaf $\Sigma^2_0$.

\medskip

According to Proposition \ref{sigma2}, we can compute all components of the Riemann curvature tensor of $\Sigma^1$. By the Gauss-Codazzi equations, 
we have
\[R^{\Sigma^1}_{ABCD}=g^{\Sigma^2}_{AC}g^{\Sigma^2}_{BD}-g^{\Sigma^2}_{AD}g^{\Sigma^2}_{BC}, \qquad R^{\Sigma^1}_{\mathbf{y}ABC}=0.\]
Moreover,
\begin{equation*}
     \mathbf y ( h^{\Sigma^2}_{AB})=-R^{\Sigma^1}_{ \mathbf y  AB \mathbf y  }+
    \langle\nabla^{\Sigma^1}_A  \mathbf y ,\nabla^{\Sigma^1}_B  \mathbf y \rangle_{\Sigma^1}.
\end{equation*}
Consequently, using $ \mathbf y ( h^{\Sigma^2}_{AB})=-\mathbf y  (g^{\Sigma^2}_{AB})=-2h^{\Sigma^2}_{AB}=2g^{\Sigma^2}_{AB}$, we obtain
\begin{equation*}
   R^{\Sigma^1}_{ \mathbf y  AB \mathbf y  }=
    \langle\nabla^{\Sigma^1}_A  \mathbf y ,\nabla^{\Sigma^1}_B  \mathbf y \rangle_{\Sigma^1}- \mathbf y  (h^{\Sigma^2}_{AB})
    =g^{\Sigma^2}_{AB}-2g^{\Sigma^2}_{AB}=-g^{\Sigma^2}_{AB}.
\end{equation*}
Therefore, $\Sigma^1$ has constant sectional curvature $-1$ and we obtain that each $\Sigma^1_\alpha$ is hyperbolic space or a quotient thereof. 
\end{proof}

\begin{theorem} \label{topology}
    Each $\Sigma^1_\alpha$ is diffeomorphic to $\mathbb R^{n-1}$.
    In particular, $M$ is diffeomorphic to $\mathbb R^n$
\end{theorem}

\begin{proof}
    This is a consequence of Reeb's stability theorem, see for instance \cite[Appendix A]{HZ24}.
\end{proof}


\section{Construction of global AdS-Brinkmann coordinates}

Throughout the section let $(M^n,g,k)$ be an initial data set with vanishing total charges, equipped with a type I null spinor solving $\nabla_i\psi=-\frac12k_{ij}e_je_0\psi-\frac12\mathbf{i}e_i\psi$.

\begin{theorem}
    There exists a global coordinate chart $\{u_1,\cdots,u_n\}$ on $M^n$ such that 
    \begin{equation*}
         g=(f^2+|\mathbf{w}|^2u_n^{-2})du_1^2+2 \sum_{\alpha=2}^n w_\alpha u_n^{-2}du_\alpha du_1+
u_n^{-2}(du_2^2+\cdots +du_n^2),
    \end{equation*}
    where $\mathbf{w}=(w_2,\cdots,w_n)$ is a vector field on the level sets of $u_1$ and $f=|\nabla u_1|^{-1}$. Moreover, $u_i=y_i+O_2(r^{-q}y_n)$, for $i=1,...,n$. 
\end{theorem}

Here $(y_1, \dots, y_n)$ are the upper half-space coordinates, i.e. $b=y_n^{-2}(dy_1^2+\cdots+dy_n^2)$ for the hyperbolic background metric $b$.
Furthermore, the coordinate functions $y_1,\dots, y_n$ are extended to the entire manifold $M^n$ so that $y_i\in C^2(M^n)$ for $1\le i\le n$, and $y_n>0$.

\medskip

We will later show that there exists a function $w$ such that  $w_\alpha=\partial_\alpha w$.

\medskip

The main difficulty in constructing the coordinates lies in the fact that $X$ is not closed; instead, we only have $dX\wedge X=0$. Although this condition indicates a codimension-one foliation, it does not allow us to obtain the level-set function by solving an elliptic equation as in the asymptotically flat case. To address this difficulty, we construct the coordinates $u_i(i<n)$ by taking the limit of 
$y_i$ along the integral curves of $N^{-1}Y$, as established in Lemma \ref{u-y}.

\medskip

Next, we introduce notation for three different types of decay at infinity: $O_{i,\alpha}$, $\mathcal{O}_{i,\alpha}$  and $\mathfrak{O}_{i,\alpha}$ for $M^n$, $\Sigma_1$ and $\Sigma_2$.
\begin{enumerate}
\item The notation $O_{i,\alpha}$ is already introduced in Definition \ref{Def:weighted}.
    \item 
    Let $\{e_2,...,e_n\}$ and $\{e_2,...,e_{n-1}\}$  be  orthonormal frames on $\Sigma_1$ and $\Sigma_2$, respectively. 
    \item Let $\alpha\in (0,1)$. The notion $\xi\in \mathcal{O}_{s,\alpha}(e^{-q\tau }y_n)$ indicates the existence of a constant $C>0$ such that for all multi-indices $I\subset \{e_2,...,e_n\}$:
\[
\sum_{|I|\le s}\left|\nabla^{\Sigma_1}_I \xi\right|+
\sum_{|I|=s}\sup_{\mathbf{x}_1,\mathbf{x}_2\in \Sigma_1
   } \frac{|\nabla^{\Sigma_1}_I \xi(\mathbf{x}_1)-\nabla^{\Sigma_1}_I\xi(\mathbf{x}_2)|}{d_{\Sigma_1}(\mathbf{x}_1,\mathbf{x}_2)^{\alpha}}\le Ce^{-q\tau}y_n,
\]
where $\tau$ is the distance function to a base point on $\Sigma_1$.
\item Similarly, $\xi\in \mathfrak{O}_{s,\alpha}(\widetilde{\tau}^{-q})$ indicates that there exists a constant $C$ such that for $I\subset \{e_2,...,e_{n-1}\}$, 
\[
\sum_{|I|\le s}\left|\widetilde{\tau}^{|I|+q}\nabla^{\Sigma_2}_I \xi\right|+
\sum_{|I|=s}\sup_{\mathbf{x}_1,\mathbf{x}_2\in \Sigma_2} 
    \widetilde{\tau}^{\alpha+s+q}  \frac{|\nabla^{\Sigma_2}_I \xi(\mathbf{x}_2)-\nabla^{\Sigma_2}_I\xi(\mathbf{x}_2)|}{d_{\Sigma_2}(\mathbf{x}_1,\mathbf{x}_2)^\alpha}\le C.
\]
where $\widetilde{\tau}$ is the radial distance function on $\Sigma_2$
\end{enumerate}

 We point out that the decay rate of $\mathfrak{O}_{i,\alpha}$  deteriorates when taking derivatives or difference quotients, while the other two will not. 
 This matches $M^n,\Sigma_1$ being asymptotically hyperbolic, and $\Sigma_2$ being asymptotically flat.

\medskip

To simplify the notation we write in the following $N^\infty=N(\psi^\infty)$, $X^\infty=X(\psi^\infty)$ and $Y^\infty=Y(\psi^\infty)$.
Moreover, we write $N=N(\psi)$, $X=X(\psi)$ and $Y=Y(\psi)$.
Note that $N=|X|=|Y|$ and $X\perp Y$ everywhere on $M^n$.

\begin{lemma} \label{rhorq}
    Let $\rho$ be a non-negative function on $\Sigma_1$ defined by $$\cosh\rho=\frac{y_2^2+\cdots+y_n^2+1}{2y_n}.$$
    Then $\rho $ satisfies $|\nabla^{\Sigma_1} \rho|=1+O(r^{-q})$. 
\end{lemma}

This shows that $\rho$ can roughly be viewed  as the distance function on $\Sigma_1$ which will be made more precise in the lemma below. 

\begin{proof}
    Let $s=\frac{y_2^2+\cdots+y_n^2+1}{2y_n}$.
    We compute
    \begin{equation*}
        \nabla s=\frac{y_2}{y_n}\nabla y_2+\cdots+\frac{y_{n-1}}{y_n}\nabla y_{n-1}+(1-\frac{s}{y_n})\nabla y_n.
    \end{equation*}
Using that $g$ is asymptotic to the hyperbolic metric $b=y_n^{-2}(dy_1^2+\cdots+dy_n^2)$, we obtain
    \begin{equation*}
        |\nabla s|^2= (s^2-1)[1+O(r^{-q})],
        \qquad
        \langle \nabla s, |X|^{-1}X\rangle=
        (s^2-1)^\frac{1}{2}O(r^{-q}).
    \end{equation*}
  Therefore, $|\nabla^{\Sigma_1} s|=(s^2-1)^\frac{1}{2}[1+O(r^{-q})]$, and $|\nabla^{\Sigma_1} \rho|=1+O(r^{-q})$.
\end{proof}

Let us fix an integral curve $\Gamma$ of $X$. According to the previous section which used Reeb's stability theorem, $\Gamma$ intersects each leaf only once.  Let $\tau$ be the intrinsic distance function to $\Gamma$ on $\Sigma_1$, i.e., $\tau=d_{\Sigma_1}(\cdot, \Gamma)$.

\begin{lemma} \label{rho-tau}
    The function $|\rho-\tau|$ is bounded. 
\end{lemma}

\begin{proof}

Throughout the proof, let $C_i$ be sufficiently large positive constants.
As $X$ is asymptotic to $\partial{y_1}$ in the coordinate system $\{y_1, \dots, y_n\}$ by Lemma \ref{NXY decay},  it follows that $\Gamma\subset \{\rho\le C_1\}$.  
By the asymptotics of $\rho$ in Lemma \ref{rhorq}, $\rho-\tau$ is uniformly bounded on $\mathfrak{L}:=\{\rho=C_1\}$. 

\medskip

Let $p\in\Sigma_1$ be a point outside $\{\rho\le C_1\}$.
    Connect $\mathfrak{L}$ and $p$ with an integral curve $\gamma_\rho(s) $ of $\nabla^{\Sigma_1} \rho$ and an integral curve $\gamma_\tau(s)$ of $\nabla^{\Sigma_1} \tau$.
    Note that  $r\ge \sinh \rho$.
    Therefore, we may use the decay estimate from Lemma \ref{rhorq} to obtain
    \begin{align}
    \label{rho1}
         \partial_s \rho(\gamma_\rho(s))=& |\nabla^{\Sigma_1} \rho|^2=1+O(r^{-q})\ge 1-C_2e^{-q\rho},
         \\ \partial_s \rho(\gamma_\tau(s))=& \langle \nabla^{\Sigma_1} \rho,\nabla^{\Sigma_1} \tau\rangle\le|\nabla^{\Sigma_1} \rho|\le 1+C_2e^{-q\rho}.
         \label{tau2}
    \end{align}
The uniform bound $|\rho-\tau|$ on $\mathfrak{L}$, together with \eqref{rho1}, yields an upper bound $|\gamma_\rho|\le\rho(p)+C_3$. Similarly, \eqref{tau2} provides a lower bound $ |\gamma_\tau|\ge\rho(p)-C_3$. Thus,  $|\rho-\tau|$ is bounded.
\end{proof}

Now we are ready to construct the AdS-Brinkmann coordinates to show that $(M^n,g,k)$ is contained in a Siklos wave.
We begin with the coordinate function $u_n$.

\begin{lemma} \label{u_n}
    On each fixed leaf $\Sigma_1$, there exists a unique function $u_n$ such that $$u_n=y_n+\mathcal{O}(e^{-q\tau}y_n)$$ and $N^{-1}Y=\nabla^{\Sigma_1} \log u_n$.  Moreover, the level sets of $u_n$ are horospheres in $\Sigma_1$. 
\end{lemma}

\begin{proof}
 According to Theorem \ref{Sigma1}, 
    $N^{-1}Y$ is integrable on $\Sigma_1$. Moreover, $N^{-1}Y=\nabla^{\Sigma_1} \log y_n+O_1(r^{-q})$.  Thus, applying Lemma \ref{rho-tau} and the inequality $r\ge \sinh \rho$, there exists a function $\widetilde{v}$ such that $$\nabla^{\Sigma_1}\widetilde{v}=\nabla^{\Sigma_1} \log y_n+\mathcal{O}(e^{-q\tau}).$$
 By possibly adding a constant, we can also ensure that $(\tilde{v}-\log y_n)\to 0$ when $y_n\to \infty$. 
   Integrating the error term $\mathcal{O}(e^{-q\tau})$, we obtain that $\tilde{v}-\log y_n\to 0$ in any direction towards infinity.  

\medskip

   Next, note that $|\nabla^{\Sigma_1} \log y_n|=1+\mathcal{O}(e^{-q\tau})$. 
   Since $y_n(\Gamma)$ is uniformly bounded, we may integrate $\nabla^\Sigma \log y_n$ from $\Gamma\cap {\Sigma_1}$ to any given point to find that $  C^{-1} e^{-\tau} \le y_n\le Ce^\tau$ on $\Sigma_1$. Let $u_n:=e^{\widetilde{v}}$.
    Integrating the new error term $\nabla^{\Sigma_1} (\log u_n-\log y_n)=\mathcal{O}(e^{-q\tau})$, it follows that $u_n=y_n+\mathcal{O}(e^{-q\tau}y_n)$. 

\medskip

    Finally, the level-sets of $u_n$ correspond to the level-sets of $\widetilde v$ which are given by the leaves $\Sigma_2$. 
    Moreover, we know from the previous section that $\Sigma_1$ is hyperbolic space foliated by the horospheres $\Sigma_2=\mathbb R^{n-2}$.
\end{proof}

Next, we construct the coordinate function $u_1$. 
Understanding the properties of $u_1$ presents the main technical challenge.

\begin{lemma} \label{y1-u1}
    On each leaf $\Sigma_1$, there exists a constant $u_1$ such that $$y_1=u_1+\mathcal{O}_3(e^{-q\tau}u_n)$$ for the coordinate function $y_1$.
\end{lemma}

\begin{proof}
We have the asymptotics $\nabla^{\Sigma_1} y_1=O_2(r^{-q}y_n)$, and Lemma \ref{rho-tau} implies that $$\nabla^{\Sigma_1} y_1=\mathcal{O}_2(e^{-q\tau}u_n).$$
To prove the claim, it thus suffices to show that $y_1$ asymptotes to a constant on $\Sigma_1$.

\medskip

Since $\Sigma_1\cong \mathbb{H}^{n-1}$, there is a coordinate system $\{v_2,\cdots, v_{n-1},u_n\}$ on $\Sigma_1$ such that $$g_{\Sigma_1}=u_n^{-2}(dv_2^2+\cdots+dv_{n-1}^2+du_n^2).$$ 
Note that for any vector $(c_2,\dots, c_{n-1})\in \mathbb{R}^{n-2}$, the map
\[(v_{2},\dots,v_{n-1},u_n)\to (v_{2}+c_2,\cdots,v_{n-1}+c_{n-1},u_n)\]
is an isometry on $\Sigma_1$. 
Thus, we can assume $v_2=\cdots=v_{n-1}=0$ at the point $\Gamma\cap \Sigma_1$. 
Furthermore, using the asymptotic of $X$, the coordinate functions $y_2,\cdots, y_n$ are bounded on $\Gamma$. Moreover, $u_n$ is uniformly bounded
on $\Gamma$ as well.  

\medskip

Define $$\mathbf{r}:=\sqrt{v_2^2+\cdots+v_{n-1}^2+1}\quad\text{and} \quad\tilde t:=\cosh\tau= \cosh\Big(d_{\Sigma_1}(\cdot,\Gamma)\Big)=\frac{\mathbf{r}^2+u_n^2}{2u_n}.$$ 
Note that in this notation we have $\nabla^{\Sigma_1} y_1=\mathcal{O}(\tilde{t}^{-q}u_n)$.
Let $p_0\in \Sigma_1$.
First we show that $y_1$ is near constant on an integral curve $\gamma $ of $\nabla^{\Sigma_1}y_1$.

\medskip

First suppose that $u_n(p_0)\le \mathbf{r}(p_0)$. Note that  $\tilde{t}$ is decreasing with respect to $u_n$ when $u_n\in(0, \mathbf{r}]$. 
Consequently, for any $p_1$ on the integral curve of $Y$ through $p_0$ with $u_n(p_1)\in (0,  u_n(p_0)]$, we have 
\begin{equation}\label{case 1}
    |y_1(p_0)-y_1(p_1)|\le \int_{u_n(p_1)}^{u_n(p_0)} |\nabla^{\Sigma_1} y_1|u_n^{-1} du_n \le C\int_{0}^{u_n(p_0)} \tilde{t}^{-q}du_n\le C\tilde{t}^{-q}(p_0)u_n(p_0)
\end{equation}
where $C$ is a constant such that $|\nabla^{\Sigma_1}y_1|\le C\tilde{t}^{-q}u_n$. 

\medskip

Next, suppose $u_n(p_0)\ge \mathbf{r}(p_0)$, we have $\tilde{t}\le u_n$.  
Suppose that  $p_2$ is a point on the integral curve of $Y$ through $p_0$ with $u_n(p_2)\in (u_n(p_0),\infty)$.
Then we have 
\begin{align}\label{case 2}
\begin{split}
    |y_1(p_2)-y_1(p_0)|\le&
    \int_{u_n(p_0)}^{u_n(p_2)} |\nabla^{\Sigma_1} y_1|u_n^{-1} du_n
    \\ \le& C\int_{u_n(p_0)}^{u_n(p_2)} \tilde{t}^{-q}du_n
    \\ \le& 2^{q}C\int_{u_n(p_0)}^{u_n(p_2)} u_n^{-q}du_n
    \\ \le& 2^{q}C\cdot \frac{u_n(p_0)^{1-q}}{q-1}
    \\ \le & \frac{2^{q}C}{q-1} \tilde{t}^{-q}(p_0)u_n(p_0),
    \end{split}
\end{align}
where the third inequality uses $\tilde{t}\ge \frac{u_n}{2}$. 
     \medskip
     
Now  for any points $p_0$ and $p_0'$, we can connect $p_0$ with $p_1$ and $p_0'$ with $p_2$ along integral curves of  $Y$ such that $\tilde{t}(p_1)=\tilde{t}(p_2)=\tilde{t}_1$.
We will show that $|y_1(p_1)-y_1(p_2)|\le C\pi\tilde{t}^{1-q}_1$ which will imply that $u_1:=\lim_{p\to\infty}y_1(p)$ is well-defined and independent of $p\in \Sigma_1$. Moreover, combining this with the estimates \eqref{case 1} and \eqref{case 2}, we will obtain $y_1=u_1+\mathcal{O}(\tilde{t}^{-q}u_n)$.

\medskip

Recall that $\tilde{t}=\frac{v_2^2+\cdots+v_{n-1}^2+u_n^2+1}{2u_n}$.
On the sphere $\{\tilde{t}=\tilde{t}_1\}$, we note that  
\[v_2^2+\cdots+v_{n-1}^2+(u_n-\tilde{t}_1)^2=\tilde{t}_1^2-1,\]
i.e., the sphere $\{\tilde{t}=\tilde{t}_1\}$ can be considered a sphere in Euclidean space of radius $\sqrt{\tilde{t}_1^2-1}$. Let $\gamma$ be a curve on the sphere connecting $p_1$ and $p_2$, we have 
\begin{equation*}
    |y_1(p_1)-y_1(p_2)|\le \int_\gamma |\nabla^{\Sigma_1} y_1|\le C\int_\gamma \tilde{t}_1^{-q}u_n \le C|\gamma|_{g_{\mathbb R^{n}}}\tilde{t}_1^{-q}\le C\pi \tilde{t}_1^{1-q},
\end{equation*}
where $\gamma$ is chosen so that the length $|\gamma|_{g_{\mathbb R^n}}$ of $\gamma$ with respect to Euclidean metric is not greater than $\pi \tilde{t}_1$.
Therefore, the result follows.
\end{proof}

In the following let $\pmb{\rho}=(y_1^2+\cdots+y_{n-1}^2+1)^\frac{1}{2}$.
Moreover, let $\Phi(s,p)$ be the flow generated by $|Y|^{-1}Y$. Note that 
\[u_1(p)=\lim_{s\to \infty} y_1(\Phi(s,p))=\lim_{s\to -\infty} y_1(\Phi(s,p)).\]
We have shown $u_1$ is constant on the leaves $\Sigma_1$. 
Next, we prove that $u_1$ is differentiable in the normal direction (to $\Sigma_1$). 
To do so, we will analyze the differentiability of the solution to the ODE system \eqref{ODE} below with respect to the initial condition. Subsequently, the regularity for the remaining coordinates
$u_A$ $(A=2,\dots,n-1)$ can be established in a similar manner to $u_1$. It relies on the observation established in Lemma \ref{uA-yA} that $y_A$
  is asymptotic to the same constant $u_A$
  at both ends of the integral curve of $|Y|^{-1}Y$.
  
\begin{lemma} \label{u-y}
The following decay and regularity estimates hold:
\begin{align*}
    u_1-y_1=&O_{2}(r^{-q}y_n),\\
    u_n-y_n=&O_2(r^{-q}y_n) \qquad \text{for $q>2$},\\
    u_n-y_n=&O_1(r^{-q}y_n) \qquad \text{for any $q>\frac n2$.}
\end{align*}
Moreover, if there exists $\tilde{u}_A$ such that $$\tilde{u}_A=\lim_{s\to \infty} y_A(\Phi(s,p))=\lim_{s\to -\infty} y_A(\Phi(s,p)), \qquad A=2,\cdots,n-1,$$ then we have the decay estimate $$\tilde{u}_A-y_A=O_2(r^{-q}y_n).$$
\end{lemma}

\begin{proof}
We begin with proving the first estimate for $u_1$.
From there it will be easy to show that the others hold as well.

\medskip

    Throughout the proof, we use the  coordinate functions $(y_1,\dots, y_n)$,   $C_i$ and  $\tilde{C}_i$ denote positive constants depending only on $(M^n,g,k)$, and the subscript $a$ ranges from $1$ to $n-1$. 
    Let $p(\mathfrak{t})$ be a curve in $M^n$ with $p(0)=p$.
    Suppose that $|p'(\mathfrak{t})|$, $|p''(\mathfrak{t})|$ and $|p'''(\mathfrak{t})|\in[\frac{1}{2},1]$, i.e.,
    \begin{equation}\label{pt}
    |p'(\mathfrak{t})|_{g_{\mathbb R^n}},\;\;|p''(\mathfrak{t})|_{g_{\mathbb R^n}} ,\;\;|p'''(\mathfrak{t})|_{g_{\mathbb R^n}}\in [C_0^{-1}y_n(p),C_0y_n(p)]. 
    \end{equation}
    
\medskip

   Recall that $\Phi(s,p)$ is the flow generated by $|Y|^{-1}Y$.
    Let $\mathbf{p}(s,\mathfrak{t})= \Phi(s,p(\mathfrak{t}))$ and let $F=\frac{Y}{|Y|}$.
    Applying Lemma \ref{NXY decay}, we find that there exists a constant $C>0$ such that 
    \begin{equation} \label{Const}
        \|F-y_n\partial_{y_n}\|_{C^2(M^n)}\le Ct^{-q}.
    \end{equation}
   This implies
    \begin{equation} \label{ODE}
        \partial_s \mathbf{p}=F(\mathbf{p})=y_n(\mathbf{p}) \partial_{y_n}+O_2(t^{-q}). 
    \end{equation}
    Write $\mathbf{p}(s,0)=(\mathbf{p}_1,\dots,\mathbf{p}_n)$ where we use the coordinates $(y_1,\dots,y_n)$. 
    Thus, there exists a constant $C>0$ such that  \begin{equation}\label{yPhispi}
        |\partial_s \mathbf{p}_a|\le C\mathbf{p}_nt^{-q}(\mathbf{p}), a=1,\cdots, n-1; \quad 
        |\partial_s \mathbf{p}_n-\mathbf{p}_n|\le C\mathbf{p}_nt^{-q}(\mathbf{p}).
    \end{equation}
    
\medskip

    \textbf{Case I.} When $y_n(p)\ge \pmb{\rho}(p)$ and $t(p)$ is sufficiently large, the $y_n(p)\ge t(p)$ is also  sufficiently large. 
    Thus, using $t\ge\frac{y_n}{2}$, the second inequality in  \eqref{yPhispi} implies that
    \begin{equation} \label{pnC}
        |\partial_s\mathbf{p}_n-\mathbf{p}_n|\le 2^{q}C\mathbf{p}^{1-q}_n. 
    \end{equation}
    Assuming $y_n(p)\ge 100C^{1/q}$, we apply upper and lower barrier arguments to inequality \eqref{pnC} to obtain
    \[\frac{1}{2}(e^{s}+ e^{\frac{s}{2}})y_n(p)\le \mathbf{p}_{n}(s)\le (2e^{s}- e^{\frac{s}{2}})y_n(p).\]
  In particular, we have  
  \begin{equation} \label{pns0}
      \frac{1}{2}e^{s}y_n(p)\le \mathbf{p}_{n}(s)\le 2e^{s}y_n(p). 
  \end{equation}
 Combining the estimate $t\ge \frac{y_n}{2}$ with the inequalities in \eqref{yPhispi}, we obtain for $s\ge0$
    \begin{equation}
        |\partial_s\mathbf{p}_{a }|\le 4^{q}C e^{s(1-q)}y^{1-q}_n(p)\quad \text{and} \quad |\partial_s\mathbf{p}_{n}-\mathbf{p}_{n}|\le 4^{q}C e^{s(1-q)}y^{1-q}_n(p).
    \end{equation}
     Therefore, for $0\le s_1<s_2$,
     we have
     \begin{equation}
     \label{pas}
|\mathbf{p}_a(s_1)-\mathbf{p}_a(s_2)|\le C_1e^{s_1(1-q)}y_n^{1-q}(p) ,\quad
            |e^{-s_1}\mathbf{p}_n(s_1)-e^{-s_2}\mathbf{p}_n(s_2)|\le C_1e^{-qs_1}y_n^{1-q}(p) .  
     \end{equation}

     \medskip

    Next, let $\mathbf{q}:=\partial_\mathfrak{t}\mathbf{p}(s,\mathfrak{t})|_{\mathfrak{t}=0}$. Equation \eqref{ODE} implies that $F_j=\delta_{jn}y_n+O_{2}(y_nt^{-q})$ and
\begin{equation}\label{psq}
        \partial_s\mathbf{q}_j=\frac{\partial F_j}{\partial y_i}\mathbf{q}_i=\delta_{jn}\mathbf{q}_n+|\mathbf{q}|_{g_{\mathbb{R}^n}}O_1(t^{-q}).
    \end{equation}

Note that $|\mathbf{q}(s)|_{g_{\mathbb{R}^n}}\neq 0$ by the uniqueness of the ODE solution. 
Hence, the function $|\mathbf{q}(s)|_{g_{\mathbb{R}^n}}$ 
is differentiable. Using  the inequalities $t\ge \frac{1}{2}y_n$ and $\mathbf{p}_n(s,0)\ge \frac{1}{2}e^s y_n(p)$, and combining this with equation \eqref{Const} and equation \eqref{psq}, we obtain
\begin{equation*}
    \partial_s |\mathbf{q}(s)|_{g_{\mathbb{R}^n}}\le |\partial_s\mathbf{q}(s)|_{g_{\mathbb{R}^n}}\le (1+C_2e^{-qs}y_n^{-q}(p))|\mathbf{q}(s)|_{g_{\mathbb{R}^n}}.
\end{equation*}
Thus, $|\mathbf{q}(s)|_{g_{\mathbb{R}^n}}\le C_3|\mathbf{q}(0)|_{g_{\mathbb{R}^n}}e^s$. Moreover, for $ 0\le s_1<s_2$, equation \eqref{pns0} and equation \eqref{psq} imply that
\begin{equation}\label{qs}
    |\mathbf{q}_a(s_1)-\mathbf{q}_a(s_2)|\le C_4e^{s(1-q)}y^{1-q}_n(p), \quad 
    |e^{-s_1}\mathbf{q}_n(s_1)-e^{-s_2} \mathbf{q}_n(s_2)|\le C_4 e^{-qs_1}y^{1-q}_n(p).
\end{equation}
    
\medskip

Next, we consider the second variation by differentiating equation \eqref{psq}  and applying the estimate
$\frac{\partial^2 F_j}{\partial y_i\partial y_l}=O(y_n^{-1}t^{-q})$ from equation \eqref{Const}. 
Let $\mathring{\mathbf{q}}(s)=\partial_\mathfrak{t}^2\mathbf{p}(s,\mathfrak{t})|_{\mathfrak{t}=0}$.
Then we obtain
\begin{equation} \label{ptq2}
\begin{split}
    \partial_s\mathring{\mathbf{q}}_j
    =&\frac{\partial^2 F_j}{\partial y_i\partial y_l}\mathbf{q}_i\mathbf{q}_l+\frac{\partial F_j}{\partial y_i}\mathring{\mathbf{q}}_i
    \\=& \delta_{jn}
    \mathring{\mathbf{q}}_n
    +(|\mathbf{q}|^2_{g_{\mathbb{R}^n}}+y_n|\mathring{\mathbf{q}}|_{g_{\mathbb{R}^n}})O(y_n^{-1}t^{-q}).
\end{split}
\end{equation}
Hence, using  $|\mathbf{q}(s)|_{g_{\mathbb{R}^n}}\le C_3|\mathbf{q}(0)|_{g_{\mathbb{R}^n}}e^s$  and $|\mathbf{q}(0)|_{g_{\mathbb{R}^n}}=|p'(\mathfrak{t})|_{g_{\mathbb{R}^n}}\le C_0y_n(p)$ from equation \eqref{pt}, we have 
\begin{equation} \label{ptqi}
    \partial_s \mathring{\mathbf{q}}_j=\delta_{jn}\mathring{\mathbf{q}}_n+ O(y_n t^{-q})+|\mathring{\mathbf{q}}|_{g_{\mathbb{R}^n}} O(t^{-q}).
\end{equation}
Combining this with the estimates $t\ge \frac{1}{2}y_n$, $\mathbf{p}_n(s)\ge \frac{1}{2}e^s y_n(p)$ and equation \eqref{Const}, we obtain
\begin{equation}
    \begin{split}
        \partial_s|\mathring{\mathbf{q}}|_{g_{\mathbb{R}^n}}\le&
    (1+Ct^{-q})|\mathring{\mathbf{q}}|_{g_{\mathbb{R}^n}} 
    +C_5y_n t^{-q}
    \\ \le& 
     (1+4^qCy_n^{-q}(p)e^{-qs})|\mathring{\mathbf{q}}|_{g_{\mathbb{R}^n}} 
    +4^qC_5y_n^{1-q}(p)e^{(1-q)s}.
    \end{split}
\end{equation}
Thus, $|\mathring{\mathbf{q}}|_{g_{\mathbb{R}^n}}\le C_6y_n(p)e^{s}$. 
Therefore,  for any $0\le s_1<s_2$ equation \eqref{ptqi} yields that
\begin{equation}\label{qas}
    \begin{split}
        &|\mathring{\mathbf{q}}_a(s_1)-\mathring{\mathbf{q}}_a(s_2)|\le C_7e^{(1-q)s_1}y^{1-q}_n(p)\le C_7e^{(1-q)s_1}y_n(p)t^{-q}_n(p),
        \\ &|e^{-s_1}\mathring{\mathbf{q}}_n(s_1)-e^{-s_2} \mathring{\mathbf{q}}_n(s_2)|\le C_7e^{-qs_1}y_n(p)t^{-q}_n(p).
    \end{split} 
\end{equation}
Thus, equations \eqref{pas}, \eqref{qs} and \eqref{qas} imply that $\mathbf{p}_1(s,0)$, 
$\partial_\mathfrak{t}\mathbf{p}_1(s,0)$ and $\partial^2_\mathfrak{t}\mathbf{p}_1(s,0)$ converge uniformly for $s\to \infty$. Hence, $u_1(p)=\lim_{s\to \infty}\mathbf{p}_1(s,0)$ is $C^{2}$. Moreover, $u_1-y_1=O_{2}(t^{-q}y_n)$.

\medskip

\textbf{Case II.} Next, suppose that $y_n(p)\le \pmb{\rho}(p)+1$ and $t(p)$ is large.
In this case,  we only consider the case $s\le 0$. 
Let 
\[\pmb{v}(s)=\frac{\pmb{\rho}^2}{y_n}\Big{|}_{\mathbf{p}(s,0)}.\]
Then by assumption, $\pmb{v}(0)\ge \frac{1}{4}t(p)$ is sufficiently large. 
According to line \eqref{yPhispi}, we have
\begin{equation}
    \partial_s \pmb{v}=-\pmb{v}\mathbf{p}_n\partial_s\mathbf{p}_n+\sum_{a=1}^{n-1}\frac{\mathbf{p}_a}{\mathbf{p}_n}\partial_s \mathbf{p}_a
    \le  -\pmb{v}+nC\pmb{v}t^{-q}
    \le -\pmb{v}+nC\pmb{v}^{1-q}.
\end{equation}
Assuming $\pmb{v}(0)\ge 10nC^{1/q}$, we obtain for $s\le 0$ the lower barrier    
\begin{equation}
\label{v(s)}
    \pmb{v}(s)\ge \frac{1}{2}(e^{-s}+e^{-\frac{s}{2}})\pmb{v}(0).
\end{equation} 
Moreover, when $s\le 0$,   the second inequality in \eqref{yPhispi} implies
\begin{equation}
    |\partial_s \mathbf{p}_n-\mathbf{p}_n|\le C\mathbf{p}_n\pmb{v}^{-q}\le 2^qe^{qs}\pmb{v}^{-q}(0)C\mathbf{p}_n.
\end{equation}
Similarly,  a barrier argument for $s\le 0$ yields 
\begin{equation} \label{Cpn}
     \frac{1}{2}e^sy_n(p)\le \frac{1}{2}(e^s+e^{\frac{3}{2}s})y_n(p)\le \mathbf{p}_n\le (2e^{s}-e^{\frac{3}{2}s})y_n(p)\le 2e^{s}y_n(p).
\end{equation}
Combining this with lines \eqref{yPhispi}, \eqref{v(s)}, \eqref{Cpn}, we find that 
\begin{equation}
    |\partial_s \mathbf{p}_a|\le 2^{q+1}Ce^{(q+1)s} y_n(p) 
    \pmb{v}^{-q}(0)\quad
    \text{and}
    \quad 
    |\partial_s\mathbf{p}_n-\mathbf{p}_n|\le 2^{q+1}Ce^{(q+1)s} y_n(p) 
    \pmb{v}^{-q}(0).
\end{equation}
Hence, using $\pmb{v}(0)\ge \frac{1}{4}t(p)$, we obtain for $s_2< s_1\le0$
\begin{equation}\label{pas12}
    \begin{split}
|\mathbf{p}_a(s_1)-\mathbf{p}_a(s_2)|\le& 2^{q+1}Ce^{(q+1)s_1} y_n(p) 
    \pmb{v}^{-q}(0)\le 2^{3q+1}Ce^{(q+1)s_1} y_n(p) 
    t^{-q}(p),
\\ |e^{-s_1}\mathbf{p}_n(s_1)-e^{-s_2}\mathbf{p}_n(s_2)|\le& 2^{q+1}Ce^{qs_1} y_n(p) 
    \pmb{v}^{-q}(0)\le 2^{3q+1}Ce^{qs_1} y_n(p) 
    t^{-q}(p).
    \end{split}
\end{equation}

\medskip

Next, we combine equations \eqref{v(s)} and \eqref{psq} with the estimate $|\mathbf{q}(s)|_{g_{\mathbb{R}^n}}$ to find 
\begin{equation}
     \partial_s |\mathbf{q}(s)|_{g_{\mathbb{R}^n}}\ge -|\partial_s\mathbf{q}(s)|_{g_{\mathbb{R}^n}}\ge -(1+\tilde{C}_1 e^{qs} \pmb{v}^{-q}(0))|\mathbf{q}(s)|_{g_{\mathbb{R}^n}}.
\end{equation}
Thus, we have $|\mathbf{q}(s)|_{g_{\mathbb{R}^n}}\le \tilde{C}_2 e^{-s} |\mathbf{q}(0)|_{g_{\mathbb{R}^n}}$ when $s\le 0$. 
Next, lines \eqref{psq} and \eqref{v(s)} yield 
\begin{equation}
    \partial_s \mathbf{q}_j=\delta_{nj}\mathbf{q}_n+|\mathbf{q}(0)|_{g_{\mathbb{R}^n}} \pmb{v}^{-q}(0) O(e^{(q-1)s}).
\end{equation}
Applying the estimate $|\mathbf{q}(0)|_{g_{\mathbb{R}^n}}\le C_0y_n(p)$ from \eqref{pt} and the assumption $\pmb{v}(0)\ge \frac{1}{4}t(p)$, we find that $$|\mathbf{q}_a(s)-\mathbf{q}_a(0)|\le \tilde{C}_3y_n(p)t^{-q}(p),$$ and  $$|\mathbf{q}_n(s)-e^s\mathbf{q}_n(0)|\le \tilde{C}_3y_n(p)t^{-q}(p)e^{\frac{s}{2}}$$ which uses $q>\frac{3}{2}$. 
Therefore, we obtain $|\mathbf{q}(s)|_{g_{\mathbb{R}^n}}\le \tilde{C}_4|\mathbf{q}(0)|_{g_{\mathbb{R}^n}}$. Next, using equations \eqref{psq} and \eqref{v(s)}, we have the improved estimate
\begin{equation}
    \partial_s \mathbf{q}_j=\delta_{nj}\mathbf{q}_n+|\mathbf{q}(0)|_{g_{\mathbb{R}^n}} \pmb{v}^{-q}(0) O(e^{qs}).
\end{equation}
 Since $|\mathbf{q}(0)|_{g_{\mathbb{R}^n}}\le C_0y_n(p)$, we obtain for $s_2< s_1\le0$
\begin{equation} \label{qan12}
    \begin{split}
        &|\mathbf{q}_a(s_1)-\mathbf{q}_a(s_2)|\le \tilde{C}_5|\mathbf{q}(0)|_{g_{\mathbb{R}^n}} \pmb{v}^{-q}(0)e^{s_1 q}\le  2^{2q}\tilde{C}_5C_0e^{s_1q} y_n(p) t^{-q}(p);
    \\ &|e^{-s_1}\mathbf{q}_n(s_1)-e^{-s_2}\mathbf{q}_n(s_2)|\le \tilde{C}_6e^{s_1 (q-1)}y_n(p) t^{-q}(p).
    \end{split}
\end{equation}

\medskip

Finally, we estimate $\mathring{\mathbf{q}}$ satisfying equation \eqref{ptq2}. 
Similar to the estimates $|\mathbf{q}(s)|_{g_{\mathbb{R}^n}}\le \tilde{C}_4|\mathbf{q}(0)|_{g_{\mathbb{R}^n}}$, we obtain
$$|\mathring{\mathbf{q}}(s)|_{g_{\mathbb{R}^n}}\le \tilde{C}_7|\mathring{\mathbf{q}}(0)|_{g_{\mathbb{R}^n}}\le C_0\tilde{C}_7y_n(p).$$ 
Therefore, equation \eqref{ptq2} implies
\begin{equation*}
    \partial_s \mathring{\mathbf{q}}_j=\delta_{jn}\mathring{\mathbf{q}}_n+ O(y_n(p)t^{-q}(p)e^{(q-1)s}).
\end{equation*}
Thus, we have for $s_2<s_1\le0$
\begin{equation} \label{ringq}
    \begin{split}
        &|\mathring{\mathbf{q}}_a(s_1)-\mathring{\mathbf{q}}_a(s_2)|\le \tilde{C}_8e^{s_1(q-1)}y_n(p)t^{-q}(p);
    \\ &|e^{-s_1}\mathring{\mathbf{q}}_n(s_1)-e^{-s_2}\mathring{\mathbf{q}}_n(s_2)|\le \tilde{C}_8e^{s_1(q-2)} y_n(p)t^{-q}(p), \quad \text{if}\;\;q>2.
    \end{split}
\end{equation}
Hence, Line \eqref{pas12}, \eqref{qan12} and \eqref{ringq} yield that
$\partial_\mathfrak{t}\mathbf{p}_1(s,0)$ and $\partial^2_\mathfrak{t}\mathbf{p}_1(s,0)$ converge uniformly when $s\to \infty$. Hence, $u_1(p)=\lim_{s\to -\infty}\mathbf{p}_1(s,0)$ is $C^{2}$. Moreover, $u_1-y_1=O_{2}(r^{-q}y_n)$.

\medskip

 Lemma \ref{u_n} implies that 
\begin{equation}
     u_n(p)=\lim_{s\to \infty} e^{-s}y_n(\Phi(s,p))=\lim_{s\to -\infty}e^{-s} y_n(\Phi(s,p)).
\end{equation}
Consequently, combining the results from equations \eqref{pas}, \eqref{qs}, \eqref{qas} \eqref{pas12}, \eqref{qan12} and \eqref{ringq} we obtain the estimates for $u_n$ and $\tilde{u}_A$. 
\end{proof}

Next we define $u_{A}$, ${A}=2,\dots,n-1$ as follow: 
\begin{equation}
    \Delta_{\Sigma_2} u_{A}=0,\quad u_{A}\to y_{A}.
\end{equation}

\begin{lemma} \label{uA-yA}
The functions $\{u_2,\dots,u_n\}$ form an upper half-space coordinate system for $\Sigma_1$, i.e. the metric $g_{\Sigma_1}$ of $\Sigma_1$ can be written as
\begin{align*}
    g_{\Sigma_1}=u_n^{-2}(du_2^2+\dots+du_n^2).
\end{align*}
Moreover, we have the decay estimate $u_A-y_A= O_2(r^{-q}y_n)$.
\end{lemma}

\begin{proof}
  Let $e_1=N^{-1}X$ and $e_n=N^{-1}Y$.
  The decay estimates from Lemma \ref{NXY decay} yields 
    \begin{equation} \label{Delta yA}
        \Delta_{\Sigma_2} y_{A}=
        \Delta y_{A}-\nabla_{e_1 e_1} y_{A}-H_{\Sigma_1}e_1(y_{A})
        -\nabla^{\Sigma_1}_{e_n e_n}y_{A}
        -H_{\Sigma_2}e_n(y_{A})=O_1(y_nr^{-q}).
    \end{equation}
Using the asymptotics of $u_1$ and $u_n$, we obtain that $\{u_n^{-1}y_2,..., u_n^{-1}y_{n-1}\}$ form an 
asymptotically flat coordinate system on $\Sigma_2$. 
Thus, the radial distance function (on $\Sigma_2$) $\widetilde{\tau}$ satisfies $$\widetilde{\tau}^2\le Cu_n^{-2}(y_2^2+\cdots+ y_{n-1}^2)+C.$$ Hence, $Ct\ge y_n(\widetilde{\tau}^2+1)+y_n^{-1}$ and we have  $$y_nt^{-q}=O(\widetilde{\tau}^{-2-a}(u_n+u_n^{-1})^{-q+1+a})$$ for $a\in (0,q-1)$.
  This implies $$u_{A}-y_{{A}}= (u_n+u_n^{-1})^{-q+1+a}\mathfrak{O}_2(\widetilde{\tau}^{-a}).$$ 
  Using $|\nabla^{\Sigma_2} y_{A}|=O(u_n)$, we find that $|\nabla^{\Sigma_2} u_{A}|$ is bounded on $\Sigma_2$. 
    Moreover, since $y_n$ is bounded on $\Sigma_2$, we have 
    \[\langle\nabla^{\Sigma_2} y_{A},\nabla^{\Sigma_2}  y_{B}\rangle=u_n^{-2}\delta_{{AB}}+O(y_n^2r^{-q})=u_n^{-2}\delta_{{AB}}+\mathfrak{O}(\varrho^{-2q}).\]
    Hence, using again that $y_n$ is bounded on a fixed level set $\Sigma_2$, we obtain 
    \[\langle\nabla^{\Sigma_2} u_{A}, \nabla^{\Sigma_2}  u_{B}\rangle=\langle\nabla^{\Sigma_2} y_{A},\nabla^{\Sigma_2}  y_{B}\rangle+\mathfrak{O}( \varrho^{1-2q})=u_n^{-2}\delta_{{AB}}+\mathfrak{O}(\varrho^{1-2q}).\]  
    Consequently, $\{u_A\}$ form a  flat coordinate chart on $\Sigma_2$ and 
    \[g_{\Sigma_2}=u_n^{-2}\sum_{{A}=2}^{n-1}d u_{A}^2
.\]
Moreover, differentiating $\Delta_{\Sigma_2} u_A$ in the $e_n=|\nabla^{\Sigma_1} u_n|^{-1}\nabla^{\Sigma_1} u_n$ direction, we have 
\begin{equation*}
  \Delta_{\Sigma_2}(\nabla_{e_n} u_A)= \nabla_{e_n} (\Delta_{\Sigma_2} u_A)+2h^{\Sigma_2}_{AB}\nabla^{\Sigma_2}_{AB}u_A+(2\nabla^{\Sigma_2}_{C}h^{\Sigma_2}_{CB}-\nabla^{\Sigma_2}_Bh^{\Sigma_2}_{CC}) \nabla^{\Sigma_2}_Bu_A=0
\end{equation*}
where $h^{\Sigma_2}=g_{\Sigma_2}$ denotes the mean curvature of $\Sigma_2\subseteq \Sigma_1$.
Since $u_A$ is approximated by $u_{A,R}$ which satisfies $\Delta_{\Sigma_2}u_{A,R}=0$ with $u_{A,R}=y_A$ on $\{\varrho=R\}$, and since $e_n(y_A)=O(y_nr^{-q})$, we obtain $\nabla_{{\epsilon_n}} u_A=0$. 
Therefore, $\{u_2,\cdots,u_{n-1},u_n\}$ form an upper half-space coordinate system on $\Sigma_1$. 

\medskip

Next, we note that $u_A$ is invariant under the flow generated by $\frac{\nabla^{\Sigma_1} u_n}{|\nabla^{\Sigma_1} u_n|}$.
Hence, they satisfy 
\begin{equation*}
    u_{A}(p)=\lim_{s\to\infty} y_{A}(\Phi(s,p))=\lim_{s\to -\infty} y_{A}(\Phi(s,p)).
\end{equation*}
Combining this with Lemma \ref{u-y}, the decay estimate follows.
 \end{proof}


\subsection{Constructing the graph function}

Using the coordinate system $\{u_1,\dots, u_n\}$ constructed in the previous section, the metric can be expressed as
\begin{equation*}
    g=(f^2+|\mathbf{w}|^2u_n^{-2})du_1^2+2 w_\alpha u_n^{-2}du_\alpha du_1+
u_n^{-2}(du_2^2+\cdots +du_n^2)
\end{equation*}
where $\mathbf{w}=(w_2,\cdots,w_n)$.
Moreover, a short computation shows that its inverse is given by
\begin{equation*}
    g^{-1}=\begin{bmatrix}
        f^{-2}& -f^{-2} \mathbf{w}^T
        \\ -f^{-2} \mathbf{w} &
        u_n^{2}I_{n-1}+  \mathbf{w}\mathbf{w}^T
    \end{bmatrix}.
\end{equation*}

Let $e_\alpha= u_n \partial u_\alpha$, $e_1=f^{-1}(\partial u_1-w_\alpha \partial u_\alpha)$, and use the convention $\alpha,\beta\in \{2,\dots, n\}$, $A,B\in \{2,\dots,n-1\}$. 
Then $\{e_i\}$ forms an orthonormal frame. Furthermore, combining this with Lemma \ref{NXY decay} and Lemma
   \ref{u-y}, we obtain
   \begin{equation} \label{ei}
       \begin{split}
           e_1=&\frac{X}{|X|}=y_n^{-1}\nabla y_1+O_2(r^{-q}),
           \\e_A=&u_n^{-1}\nabla^{\Sigma_1} u_A=y_n^{-1}\nabla y_A+O_1(r^{-q}), \\e_n=&\frac{Y}{|Y|}=y_n^{-1}\nabla y_n+O_2(r^{-q}).
       \end{split}
   \end{equation}

\begin{lemma}\label{fwa}
We have the following decay and regularity estimates:
\[
    f=y_n^{-1}+O_1(r^{-q} y_n^{-1}),\quad
\quad    w_\alpha=O_1(r^{-q}),
\quad  w_{\alpha,n}=O_1(r^{-q}y_n^{-1}).
\]
\end{lemma}

\begin{proof}
Lemma \ref{u-y} implies that
    $f=|d u_1|^{-1}=y_n^{-1}+O_1(r^{-q} y_n^{-1})$, and  Lemma \ref{uA-yA} implies that
    \[w_A=-f^2\langle du_1,du_A \rangle=-f^2\langle du_1,
    d(u_A-y_A+y_A)\rangle=O_1(r^{-q}).\]
   However, we only have $w_n=O(r^{-q})$ by part (3) of Lemma \ref{u-y}.  
According to equation \eqref{ei}, 
 $   \langle [e_1,e_n],e_\alpha\rangle=O_1(r^{-q})$. On the other hand, by equation \eqref{nuealpha}, we have 
 \begin{equation} \label{e1en}
     \langle [e_1,e_n],e_\alpha\rangle=-u_n^{-1}f^{-1}\delta_{n\alpha}w_n+f^{-1}w_{\alpha,n}=O_1(r^{-q}). 
 \end{equation}
 Therefore, $w_{A,n}=O_1(r^{-q}y_n^{-1})$. 
 When $\alpha=n$ and $q<2$, $w_{n,n}$ is only defined as a distribution.
 Thus, equation \eqref{e1en} holds in a distributional sense. 
 In this case, the regularity can be improved via the identity
 \begin{equation*}
     \langle [e_1,e_n],e_n\rangle
     =f^{-1}u_n\partial_n (u_n^{-1}w_n)=O_1(r^{-q}).
 \end{equation*}
Combining this with $u_n=y_n+O_1(y_nr^{-q})$, we conclude that $w_n=O_1(r^{-q})$ and $w_{n,n}=O_1(r^{-q}y^{-1}_n)$.
   
\end{proof}

\begin{lemma} \label{Nf}
  In the above coordinate system $\{u_1,\cdots,u_n\}$ we have  $$N=f^{-1}u_n^{-2}.$$
  Furthermore, the regularity of $f$ can be improved.
  More precisely, we obtain $$f=y_n^{-1}+\mathcal{O}_2(r^{-q}y_n^{-1}).$$
\end{lemma}
\begin{proof}
Recall that $X=Ne_1$ and $Y=Ne_n$. 
Theorem \ref{Thm: N,X,Y,omega} yields
 \begin{align*}
    \nabla_i N=&-k_{ij}X_j-Y_i
    =-k_{i1}N-N\delta_{ni}
    \\ 
    \end{align*}
    and 
    \begin{align}\label{kialpha}
    \nabla_i (N^{-1}X)_j=&-k_{ij}+\delta_{i1}\delta_{jn}-\delta_{j1}\delta_{in}+(k_{i1}+\delta_{ni})\delta_{j1}=-k_{i\alpha}\delta_{j\alpha}+\delta_{i1}\delta_{jn}.
\end{align}
With the help of Lemma \ref{lemma christoffel symbols}, we obtain
\begin{equation*}
    k_{1\alpha}=\delta_{\alpha n}-\langle\nabla_1e_1, e_\alpha\rangle= \delta_{\alpha n}-u_n\partial_\alpha(f^{-1}) f.
\end{equation*}
Moreover,
\begin{equation*}
     u_n\partial_\alpha N=\nabla_\alpha N=-k_{\alpha 1}N-N\delta_{n\alpha}=Nu_n \partial_\alpha(f^{-1})f-2N\delta_{\alpha n}.
 \end{equation*}
Therefore, $ \partial_\alpha \log N= \partial_\alpha \log( f^{-1}u_n^{-2})$. 
Since $N=y_n^{-1}+O_2(r^{-q} y_n^{-1})$ by Lemma \ref{NXY decay} and $f^{-1}u_n^{-2}=y_n^{-1}+O_1(r^{-q}y_n^{-1})$ by Lemma \ref{fwa}, the first claim follows. 
Finally, the regularity estimates for $f$ are implied by the estimates for $u_n$ and $N$ from Lemma \ref{u-y} and Lemma \ref{NXY decay}.
\end{proof}

The next important lemma allows us to extract the graph function.
Compare this to \cite[Lemma 6.5]{HZ24}.

\begin{lemma} \label{w estimate} 
     There exists a function $w$ such that  $w_\alpha=\partial_\alpha w$.
     Moreover, we have the decay and regularity estimates
     \begin{equation} \label{wrq}
         w=O_1(u_n\hat{t}^{-q}),\qquad
         w_\alpha=\mathcal{O}_2( \hat{t}^{-q}),\qquad
         e_1(w)=\mathcal{O}_2(u_n \hat{t}^{-q}).
     \end{equation}
     where $\hat{t}=\frac{u_2^2+\cdots+u_n^2+1}{2u_n}$.
\end{lemma}

\begin{proof}
We will first show that $w$ is closed.
Since $e_n=N^{-1}Y$, equation \eqref{NXYw} implies that
 \begin{equation*}
    \nabla_{e_i} e_n=\nabla _{e_i}(N^{-1}Y)
    =-e_i+k_{il}(\delta_{j1}\delta_{ln}-\delta_{l1}\delta_{jn})e_j+(k_{i1}+\delta_{in})e_n
    =k_{in}e_1+\delta_{in}e_n-e_i.
\end{equation*}
For $i=1$, this becomes in view of Lemma \ref{lemma christoffel symbols}
\begin{equation} \label{wAn}
   0=\langle \nabla_{e_1}e_n, e_A\rangle=\frac{1}{2}f^{-1}(w_{A,n}-w_{n,A}).
\end{equation}
Combining this with $w_{\alpha,n}\in C^1(M^n)$ from Lemma \ref{Nf}, we obtain $w_n\in C^2(\Sigma_1)$.  
This in particular implies $$w_{A,nB}=w_{n,AB}=w_{B,nA}.$$
Let $\pmb{\xi}$ be a smooth compact supported function.
Recall $w_A$ and $w_B\in C^1(M)$ from Lemma \ref{Nf}.
Then we obtain for $A\neq B$ 
\begin{align*}
    &\int_{\mathbb{R}^3} \pmb{\xi}w_{A,nB}du_Adu_Bdu_n= \int_{\mathbb{R}^3} (\partial_{u_n}\partial_{u_B}\pmb{\xi})w_{A}du_Adu_Bdu_n=-\int_{\mathbb{R}^3} (\partial_{u_n}\pmb{\xi})w_{A,B}du_Adu_Bdu_n,
    \\ &\int_{\mathbb{R}^3} \pmb{\xi}w_{B,nA}du_Adu_Bdu_n= \int_{\mathbb{R}^3} (\partial_{u_n}\partial_{u_A}\pmb{\xi})w_{B}du_Adu_Bdu_n=-\int_{\mathbb{R}^3} (\partial_{u_n}\pmb{\xi})w_{B,A}du_Adu_Bdu_n.
\end{align*}
 Therefore, we have for any smooth compactly supported function $\pmb{\xi}$,
 \begin{equation*}
     \int_{\mathbb{R}^3}(\partial_{u_n}\pmb{\xi})(w_{A,B}-w_{B,A})du_Adu_Bdu_n=0.
 \end{equation*}
Hence, $w_{A,B}-w_{B,A}$ does not depend on $u_n$. 
Using the asymptotic estimates from Lemma \ref{fwa} and the estimates $|\partial_A|=u_n^{-1}=O(y_n^{-1})$ which follows from Lemma \ref{u-y}(3),  we obtain
\begin{equation*}
   w_{A,B}-w_{B,A}=O(y_n^{-1}r^{-q})
\end{equation*}
which vanishes at infinity.
Therefore, $w_{A,B}=w_{B,A}$. 
Consequently, and with the help of equation \eqref{wAn}, there exists a function $w$ with $w_\alpha=\partial_{u_\alpha} w$. 
    
\medskip

Let $\hat{\tau}:=d_{\Sigma_1}((u_1,0,\cdots,0,1),\cdot)$.
Then we have $\hat{t}=\cosh \hat{\tau}$.
Recall that $u_i-y_i\to 0$ and $t=O(e^{\hat{\tau}})$.
Since $w_\alpha=O_1(r^{-q})$ by Lemma \ref{fwa}, it follows that $\nabla^{\Sigma_1}w=O_1(e^{-q\hat{\tau}}u_n)$. Replacing $y_1$ by $w$ and $\tau$ by $\hat{\tau}$ in Lemma \ref{y1-u1}, and arguing as in Lemma \ref{y1-u1}, we obtain $w=\mathcal{O}_2(e^{-q\hat{\tau}}u_n)$ after imposing $w\to 0$ at infinity on $\Sigma_1$.
    
\medskip

Let $\gamma$ be the integral curve of $\nabla^{\Sigma_1}\hat{\tau}$ connecting $p$ and infinity. Then
\[w(p)=-\int_{\gamma} \langle \nabla^{\Sigma_1}w, \nabla^{\Sigma_1}\hat{\tau}\rangle.\]
Since $|\partial u_1|=|fe_1+w_\alpha u_n^{-1} e_\alpha|=O(y_n^{-1})$, there exists a constant $C>0$ such that
\begin{equation*}
    |\partial_{u_1}w|\le \int_\gamma (|\nabla_{\partial_{u_1}}\nabla^{\Sigma_1} w|+|\nabla^{\Sigma_1}w||\nabla_{\partial_{u_1}}\nabla^{\Sigma_1} \hat{\tau}|)\le C\int_\gamma e^{-q\hat{\tau}}=\frac{C}{q}e^{-q\hat{\tau}(p)}.
\end{equation*}
Using $e_1=f^{-1}(\partial u_1-w_\alpha \partial u_\alpha)$, $w_\alpha=O_1(e^{-q\hat{\tau}})$ and the first equation in \eqref{nuealpha}, we obtain $e_1(w)=\mathcal{O}_1(u_ne^{-q\hat{\tau}})$. 
Thus, we have $w=O_1(u_ne^{-q\hat{\tau}})$.
Combining this with $\hat{t}=\cosh\hat{\tau}$ yields the first equation in \eqref{wrq}.

\medskip

Furthermore, equation \eqref{kialpha}, Lemma \ref{lemma christoffel symbols} and the decay assumption $k=O_1(r^{-q})$ imply that
\begin{align*}
    k_{A\alpha}=&-\langle \nabla_{A} e_1, e_\alpha\rangle=\frac{1}{2}f^{-1}(w_{A,\alpha}+w_{\alpha,A})-f^{-1}w_n u_n^{-1}\delta_{A\alpha}=O_1(r^{-q}),
\\ 
    k_{nn}=&-\langle \nabla_{n}e_1,e_n\rangle=f^{-1}w_{n,n}-f^{-1}u_n^{-1}w_n=O_1(r^{-q}).
\end{align*}
Therefore, using  $w_\alpha=O_1(r^{-q})$ from Lemma \ref{fwa}, we obtain $\partial_{u_\alpha}\partial_{u_\beta}w=O_1(r^{-q}u_n^{-1})$ which yields the second equation in \eqref{wrq}.
Combining this with $e_1(w)=\mathcal{O}_1(u_ne^{-q\hat{\tau}})$ gives $e_1(w)=\mathcal{O}_2(u_ne^{-q\hat{\tau}})$.
Since $\hat{t}=\cosh \hat{\tau}$, the claim follows.

\end{proof}

\begin{theorem} \label{kXL}
    $(M^n,g,k)$ isometrically embeds into a Siklos wave spacetime $(\overline M^{n+1},\overline g)$ with second fundamental form $k$ satisfying 
    $$k_{ij}=-\frac{1}{2}fu_n^{2}(\nabla_i X_j+\nabla_j X_i),\quad X=u_n^{-2}\nabla u_1$$
    where the wave profile function $L$ is given by
    $$L=u_n^2 f^2-2w_1+\sum_{\beta=2}^{n}|w_\beta|^2.$$
    Moreover, $(M^n,g,k)$ is the $(t=-w)-$graph over the $(t=0)$-slice in $(\overline M^{n+1},\overline g)$.
\end{theorem}

\begin{proof}
Recall that
    \begin{align*}
        g=(f^2+|\mathbf{w}|^2u_n^{-2})du_1^2+2 \sum_{\alpha=2}^n w_\alpha u_n^{-2}du_\alpha du_1+
u_n^{-2}(du_2^2+\cdots +du_n^2).
    \end{align*}
    The equations of $k$ and $X$ follow from equation  \ref{NXYw} and Lemma \ref{Nf}.

    \medskip
    
    Let $\overline M^{n+1}=M^n\times\mathbb R$.
    Define the standard Killing development metric
    \begin{equation*}
     \overline{g}=-N^2d\tau^2+g_{ij}(dx^i+X^id\tau)(dx^j+X^jd\tau)=2u_n^{-2}d\tau du_1+g
\end{equation*}
A short computation yields 
\begin{align*}
    \overline g=&2u_n^{-2}d\tau du_1+(f^2+|\mathbf{w}|^2u_n^{-2})du_1^2+2 \sum_{\alpha=2}^n w_\alpha u_n^{-2}du_\alpha du_1+
u_n^{-2}(du_2^2+\cdots +du_n^2)\\
=&\frac{1}{u_n^2}( 2du_1dt +\delta_{\alpha\beta}du_{\alpha}du_\beta+ Ldu_1^{2})
\end{align*}
where $t=-\tau-w$ and $L=u_n^2 f^2-2w_1+\sum_{\beta=2}^{n}|w_\beta|^2$. 
\end{proof}


\section{Asymptotic analysis}\label{S:asymptotic analysis}

We use the notation $w_1=\partial_{u_1}w$ on $\Sigma_1$, and recall that $\Delta_{\Sigma_1}=u_n^{2}\partial_{\alpha\alpha}-(n-3)u_n\partial_{u_n}$.
Hence, Lemma \ref{Je1} implies that
\begin{equation}
    0\ge -\mu= \frac{1}{2}f^{-2}u_n^{-2}(\Delta_{\Sigma_1} L-2u_n\partial_{u_n} L)
\end{equation}
where $L=u_n^2 f^2-2w_1+\sum_{\beta=2}^{n}|w_\beta|^2$.
Moreover, Lemma \ref{fwa} and Equation \eqref{wrq} show  $L=1+O(\hat{t}^{-q})$. 
Let $\mathcal{L}:=u_n^{-1}(L-1)$ and $h:=\Delta_{\Sigma_1}\mathcal{L}-(n-1)\mathcal{L}$. Combined with $\mu=O(r^{-q})$, the estimates of  $f$ and $u_n$ by Lemma \ref{fwa} and \ref{u-y},   then $\mathcal{L}$ and $h$ satisfy 
\begin{equation} \label{hL}
h=-2f^2u_n\mu \le 0,\quad h=O(\hat{t}^{-q}u_n^{-1}),\quad
\mathcal{L}=O(\hat{t}^{-q}u_n^{-1})
.
\end{equation}
To obtain estimates for the function $\mathcal{L}$, we first need an estimate for the corresponding Green's function:

\begin{lemma} \label{Green}
    Let $G(p,x)$ be the Green's function of the operator $Q:=\Delta_{\Sigma_1}-(n-1)$.
    Then $G(p,0)=O(\hat{t}(p)^{1-n})$. 
\end{lemma}
\begin{proof}
  Note that $Q \hat{t}^a=(a-1)(n+a-1)\hat{t}^a-a(a-1)\hat{t}^{a-2}$. We make the following Ansatz for the Green's function:
  \begin{equation*}
      G_0=\sum_{i=0}^\infty A_i \hat{t}^{-2i+1-n},\quad A_0=1, \quad A_{i+1}=\frac{(2i+n-1)(2i+n)}{(2i+n+2)(2i+2)}A_i.
  \end{equation*}
  This implies $$A_{i+1}=[(1+i^{-1})^\frac{n-5}{2}+O(i^{-2})]A_i.$$ 
  Consequently, there exists a constant $C_0>0$ such that    $\lim_{i\to \infty} i^{-\frac{n-5}{2}}A_{i}=C_0$. 
Hence, the power series converges to $G_0$, when $\hat{t}>1$. In particular, when $\hat{t}\to \infty$, $G_0=O(\hat{t}^{1-n})$. 

\medskip

Next, we note that the Green's function around the origin satisfies $G(p,0)=O(d(p,0)^{3-n})=O((\hat{t}^2-1)^{\frac{3-n}{2}})$ when $n> 3$.
Moreover, $G(p,0)=O(\ln(d(p,0)))=O(\ln(\hat{t}^2-1))$ when $n=3$.  
The Taylor series in both cases are
  \begin{align*}
      &(1-\hat{t}^{-2})^\frac{3-n}{2}=\sum_{i=0}^\infty B_i \hat{t}^{-2i}, \quad B_i=\frac{\left(\frac{n-3}{2}\right)\left(\frac{n-3}{2}+1\right)\cdots\left(\frac{n-3}{2}+i-1\right)}{i!}, \quad \text{ when } n\ge 4,
   \\& \ln(1-\hat{t}^{-2})= \sum_{i=0}^\infty \frac{1}{i+1}\hat{t}^{-2i} \quad \text{ when } n=3.  
  \end{align*}
  Applying Stirling's formula,  $\lim_{i\to \infty}i^{\frac{5-n}{2}}B_i=[\Gamma(\frac{n-3}{2})]^{-1}$, we obtain that there exists a constant $C_1>0$ such that $\lim_{i\to \infty}B_i^{-1}A_i=C_1$. 
  Thus, the blow-up rate of $G_0$ at origin is equal to the Green's function, i.e., $G_0$ is a Green's function of $Q$ up to a scaling.
\end{proof}

Note that the convolution $G*h(p)$ satisfies $$Q(G*h)(x)=\left(\Delta_{\Sigma_1}-(n-1)\right)\int_{\Sigma_1}G(x,y)h(y)dy=h(x).$$

\begin{lemma}\label{G*h} Let $h$ be a non-positive function on $\Sigma_1$. Then the following holds:
\begin{enumerate}
    \item  If $h$ does not vanish at some point, there exists a constant $C_h$ depending on $h$ such that
   $G*h\le -C_h \hat{t}^{1-n}$.
   \item If $q\in(\frac{n}{2},n-2]$ and $h$ is compactly supported, we have $G*h=\mathcal{O}_3(\hat{t}^{1-n})$.
\end{enumerate}
 
\end{lemma}
\begin{proof}
(1) Suppose that $h\le-c<0$ on $B_{x_0}(r_0)$. Then there exists $C>0$ such that
    \begin{equation*}
        \int_{\Sigma_1} G(p,x)h(x) dx\le -c\int_{B_{x_0}(r_0)} G(p,x) dx\le -C\hat{t}(p)^{1-n},
    \end{equation*}
    which proves the first claim.

    \medskip

(2) Recall that  $\hat{\tau}(p)= d_{\Sigma_1}(p, (u_1(p),0,\cdots,0,1))$, as defined in Lemma \ref{w estimate}, and $\hat{t}=\cosh \hat{\tau}$. We also have $G(p,x)=G_0(\cosh d_{\Sigma_1}(p,x))$ and $G_0(\hat{t})=O(\hat{t}^{1-n})$ by Lemma \ref{Green}.
Since $h$ is compactly supported,  we have
\begin{equation*}
    \begin{split}
\int_{\Sigma_1}|G(p,x)h(x)|dx\le& (\sup |h|)\int_{\text{supp} h} G_0(\cosh d_{\Sigma_1}(x,p))dx
\\ \le & C(\sup |h|)\int_{\text{supp} h}[\cosh d_{\Sigma_1}(x,p)]^{1-n}dx
\\ = & O(\hat{t}^{1-n}(p)),
    \end{split}
\end{equation*}
where $C>0$ is a constant such that $G_0(\hat{t})\le C\hat{t}^{1-n}$ for $\hat{t}$ sufficiently large. Finally, the $C^3$ estimate for $G*h$ follows from applying standard elliptic estimates to the equation $Q(G*h)=h$.
\end{proof}

The following result completes the proofs of Theorem \ref{Thm Intro rigidity} and the first part of Theorem \ref{Thm Intro examples}.
 
\begin{theorem}\label{theorem decay rates}
The following holds:
    \begin{enumerate}
        \item If $q\in (n-2,n]$, there are no non-trivial asymptotically AdS  metrics with $C^{2,\alpha}_{-q}$ decay rate with $\mathcal H(N(\psi^\infty), X(\psi^\infty))=0$ for a type I null spinor $\psi^\infty$.
        \item When $q\in (\frac{n}{2},n-2]$, there exist such non-trivial examples. 
    \end{enumerate}
\end{theorem}

\begin{proof}
(1) According to Lemma \ref{G*h}, if $h\not\equiv0$ somewhere, there exists $C>0$ such that $\mathcal{L}\le -C\hat{t}^{1-n}$. 
However, when $q>n-2$, this contradicts $\mathcal{L}=\mathcal{O}(\hat{t}^{-q}u_n^{-1})$ by Equation \eqref{hL}. Therefore, $h\equiv 0$, i.e., $\mathcal{L}=0$ and $L=1$. Thus, $(M^n,g,k)$ isometrically embeds into AdS spacetime with second fundamental form $k$ which proves the first claim.

\medskip

(2) Let $\bar{u}=(u_2,\cdots, u_n)$, $\mathbf{u}=(u_1,\cdots, u_n)$ and $\bar{h}(\bar{u})$ be a nontrivial, smooth, nonpositive compactly supported function. Let $\eta(u_1)$ be a smooth compactly supported function and $h(\mathbf{u})=\eta(u_1)\bar{h}(\bar{u})$. 
In view of Lemma \ref{G*h}, we have 
$\mathcal{L}:=G* h= 
\mathcal{O}_3(\hat{t}^{1-n})$.
Since $\eta(u_1)$ is compactly supported, we obtain 
\begin{equation} \label{decay of L}
    \mathcal{L}=O_3(t^{1-n})\quad \text{and}\quad L=1+u_nL=O_3(t^{2-n}).
\end{equation}
Consider the graph function $w=0$ and define $f$ via $u_n^2f^2=L$. 
Then the metric $g$ and the second fundamental form $k$ of the initial data set $(M^n,g,k)$ contained in the Siklos wave spacetime $(\overline M^{n+1},\overline g)$ corresponding to the wave profile function $L$,  are given by Theorem \ref{kXL}
    \begin{equation*}
        g=f^2du_1^2+u_n^{-2}(du_2^2+\cdots+du_n^2)\quad \text{and}
        \quad k_{ij}=-\frac{1}{2}fu_n^{2}(\nabla_i X_j+\nabla_j X_i), 
    \end{equation*}
    where $X=u_n^{-2}\nabla u_1$.
    Hence, $(M^n,g,k)$ is asymptotically AdS of order $n-2$ by Equation \ref{decay of L}.
\end{proof}

Recall that asymptotically AdS initial data sets have a decay rate $q> \frac n2$.

\begin{corollary}
The following holds:
\begin{enumerate}
  \item  For $n\le 4$ there are no non-trivial Siklos wave spacetimes containing an asymptotically AdS initial data set $(M^n,g,k)$.
   \item  For $n>5$ there are such non-trivial examples.
   \end{enumerate}
\end{corollary}

\begin{proof}
    This follows since asymptotically AdS initial data sets have a decay rate $q>\frac n2$.
\end{proof}

Remarkably, this yields the same dimension threshold as for asymptotically flat initial data sets contained in Siklos waves \cite{HZ24}.
Interestingly, although most arguments become more complicated for Siklos waves compared to their pp-wave counterparts, constructing this example is actually easier. 
The reason for this is that in the AdS setting, taking derivatives does not improve the decay rate at infinity.
This feature in the asymptotically flat setting forced us to construct a graph in the spacetime in a subtle way in \cite{HZ24}.

\begin{remark}
    The $(n+1)$-dimensional Schwarzschild-AdS metric is given by
    \begin{align*}
        \overline g=-\left(1+r^2-\frac{2m}{r^{n-2}}\right)dt^2+\left(1+r^2-\frac{2m}{r^{n-2}}\right)dr^2+r^2g_{S^{n-1}}
    \end{align*}
    has a decay rate of $q=n$.
    Hence, according to Theorem \ref{theorem decay rates}, all asymptotically AdS initial data sets contained in non-trivial Siklos wave spacetimes have less decay.
    Consequently, they can be ruled out by assuming Schwarzschild-AdS asymptotics at infinity.
    The same holds in the AF setting with vanishing cosmological constant.
    On the other hand, the extremal black hole in Section \ref{S:ultraspinning} does satisfy the stronger decay condition $q=n$ and can appear under such restrictive asymptotics.
\end{remark}

We also have a Lorentzian version of Theorem \ref{Thm Intro rigidity}:

\begin{theorem}
Let $(\overline M^{n+1},\overline g)$ be a Lorentzian manifold admitting a non-trivial imaginary Killing spinor.
Suppose that $(\overline M^{n+1},\overline g)$ contains an asymptotically AdS initial data set $(M^n,g,k)$.
Then $(\overline M^{n+1},\overline g)$ is either stationary vacuum or a Siklos wave.
\end{theorem}


\section{Ultraspinning limits of Kerr-AdS}\label{S:ultraspinning}
The study of extremal black holes plays a central role in theoretical physics \cite{banados1992black, kostelecky1996solitonic, gutowski2004supersymmetric, kunduri2006supersymmetric, kunduri2007near, Kallosh1992}. \emph{Extremal} refers to black holes that maximize their conserved charges (such as angular momentum or electromagnetic charge) relative to their mass, thereby saturating the corresponding BPS bound. In particular, extremal black holes are supersymmetric in the sense that they admit an imaginary Killing spinor, possibly with respect to a modified connection as in the Majumdar–Papapetrou family. In recent decades, attention has increasingly turned to higher-dimensional examples. For instance, in the context of the AdS/CFT correspondence, a $(4+1)$-dimensional asymptotically AdS spacetime is dual to a 
$(3+1)$-dimensional conformal field theory on its boundary.

\medskip

In this section we analyze the ultraspinning black hole found in \cite{cvetivc2005supersymmetric}.
This spacetime is stationary, vacuum, and admits a timelike imaginary Killing spinor.
Moreover, it admits asymptotically AdS initial data sets $(M,g,k)$ (of the optimal decay rate $q=n$) which saturate the BPS bound of Theorem \ref{Thm Intro PMT}.
Essentially, all the results of this section are contained in \cite{cvetivc2005supersymmetric}. However, since computational details were omitted, we will provide here the details.

\subsection{Derivation of the metric}

Let $\Lambda$ be the cosmological constant (which in the other sections was fixed to $\Lambda=-\frac{n(n-1)}2$), and denote by $\ell=\sqrt{-\frac{n(n-1)}{2\Lambda}}$ the cosmological length scale.

\medskip

The $5$-dimensional Kerr--AdS metric in Boyer--Lindquist coordinates is given by:
\begin{align*}
\overline g_{\text{Kerr--AdS}} &= -\frac{\Delta}{\rho^2} \left[ dt - \frac{a \sin^2 \theta}{\Xi_a} d\phi - \frac{b \cos^2 \theta}{\Xi_b} d\psi \right]^2 
+ \frac{\Delta_\theta \sin^2 \theta}{\rho^2} \left[ a\, dt - \frac{r^2 + a^2}{\Xi_a} d\phi \right]^2  \\
&\quad + \frac{\Delta_\theta \cos^2 \theta}{\rho^2} \left[ b\, dt - \frac{r^2 + b^2}{\Xi_b} d\psi \right]^2
+ \frac{\rho^2 dr^2}{\Delta} + \frac{\rho^2 d\theta^2}{\Delta_\theta}  \\
&\quad + \frac{(1 + r^2 \ell^{-2})}{r^2 \rho^2} 
\left[ ab\, dt - \frac{b(r^2 + a^2) \sin^2 \theta}{\Xi_a} d\phi - \frac{a(r^2 + b^2) \cos^2 \theta}{\Xi_b} d\psi \right]^2
\end{align*}
where
\begin{align*}
\Delta &= \frac{1}{r^2} (r^2 + a^2)(r^2 + b^2)(1 + r^2 \ell^{-2}) - 2m,  \\
\Delta_\theta &= 1 - a^2 \ell^{-2} \cos^2 \theta - b^2 \ell^{-2} \sin^2 \theta,  \\
\rho^2 &= r^2 + a^2 \cos^2 \theta + b^2 \sin^2 \theta,  \\
\Xi_a &= 1 - a^2 \ell^{-2}, \qquad \Xi_b = 1 - b^2 \ell^{-2}. 
\end{align*}
Here $m\ge0$ is the mass and $a,b\ge0$ are the angular momentum parameters.
Note that we need $a,b<\ell$ to avoid any coordinate singularities.
As pointed out in \cite{ChruscielMaertenTod}, the restriction on the allowed range of parameters in the Kerr–AdS family is not due to an incomplete understanding of the full solution space; rather, in view of the BPS bound, Theorem \ref{Thm Intro PMT}, it is a necessary feature of asymptotically AdS spacetimes satisfying the dominant energy condition.
  Note that in dimension $3$, we have $\mathcal E \ell\ge |\mathcal A|$ and in dimension $4$ we obtain $m\ge a+b$.

  \medskip

Setting $a=b$ the above metric simplifies to
\begin{align*}
\overline g_{\text{Kerr--AdS}} &= -\frac{\Delta}{\rho^2} \left[ dt - \frac{a \sin^2 \theta}{\Xi} d\phi - \frac{a \cos^2 \theta}{\Xi} d\psi \right]^2 
+ \frac{\Xi \sin^2 \theta}{\rho^2} \left[ a\, dt - \frac{\rho^2}{\Xi} d\phi \right]^2  \\
&\quad + \frac{\Xi \cos^2 \theta}{\rho^2} \left[ a\, dt - \frac{\rho^2}{\Xi} d\psi \right]^2
+ \frac{\rho^2 dr^2}{\Delta} + \frac{\rho^2 d\theta^2}{\Xi}  \\
&\quad +\frac1{\rho^2}(\Delta +2m)
\left[ a^2\rho^{-2} dt - \frac{a\sin^2 \theta}{\Xi} d\phi - \frac{a \cos^2 \theta}{\Xi} d\psi \right]^2
\end{align*}
where
\begin{align*}
\Delta &= \frac{1}{r^2} (r^2 + a^2)^2(1 + r^2 \ell^{-2}) - 2m,  \\
\rho^2 &= r^2 + a^2 ,  \\
\Xi &= 1 - a^2 \ell^{-2}.
\end{align*}
Next, we substitute $y^2=\rho^2\Xi^{-1}$ and $M=m\Xi^{-3}$.
Then $\Delta =\frac1{r^2}\rho^4l^{-2}(y^2+\ell^2)\Xi-2m$ and
\begin{align*}
\overline g_{\text{Kerr--AdS}} &= -\frac{\ell^{-2}y^2+1-2My^{-4}r^2}{y^2r^2} \left[ \Xi y^2dt - {y^2a \sin^2 \theta} d\phi - y^2{a \cos^2 \theta} d\psi \right]^2 
\\&+ \frac{ \sin^2 \theta}{y^2} \left[ a\, dt - y^2 d\phi \right]^2  \\
&\quad + \frac{ \cos^2 \theta}{y^2} \left[ a\, dt - y^2 d\psi \right]^2
+ \frac{dy^2}{\ell^{-2}y^2+1-2y^{-4}r^2M} + {y^2 d\theta^2}  \\
&\quad +\frac{\ell^{-2}y^2+1}{y^2r^2}
\left[ a^2 dt - {ay^2\sin^2 \theta}d\phi - {a y^2\cos^2 \theta} d\psi \right]^2.
\end{align*}
Doing another substitution $\hat\phi=\phi+a\ell^{-2}t$ and $\hat\psi=\psi+a\ell^{-2}t$, we obtain
\begin{align*}
\overline g_{\text{Kerr--AdS}} &= -\frac{\ell^{-2}y^2+1-2My^{-4}r^2}{y^2r^2} \left[y^2dt - {y^2a \sin^2 \theta} d\hat \phi - y^2{a \cos^2 \theta} d\hat \psi \right]^2 
\\&+ \frac{ \sin^2 \theta}{y^2} \left[ a(1+\ell^{-2}y^2)\, dt - y^2 d\hat\phi \right]^2  \\
&\quad + \frac{ \cos^2 \theta}{y^2} \left[ a(1+\ell^{-2}y^2)\, dt - y^2 d\hat\psi \right]^2
+ \frac{dy^2}{\ell^{-2}y^2+1-2y^{-4}r^2M} + {y^2 d\theta^2}  \\
&\quad +\frac{\ell^{-2}y^2+1}{y^2r^2}
\left[ a^2 (1+\ell^{-2}y^2)dt - {ay^2\sin^2 \theta}d\hat\phi - {a y^2\cos^2 \theta} d\hat\psi \right]^2.
\end{align*}
Suppose now that $a=b=\ell-\varepsilon$ for some small $\varepsilon>0$.
Then $r^2=y^2\Xi-a^2=-\ell^2+\mathcal O(\varepsilon)$ and
\begin{align*}
\overline g_{\text{Kerr--AdS}} &= \frac{\ell^{-2}y^2+1+2My^{-4}\ell^2}{y^2\ell^2} \left[y^2dt - {y^2\ell \sin^2 \theta} d\hat \phi - y^2{\ell \cos^2 \theta} d\hat \psi \right]^2 
\\&+ \frac{ \sin^2 \theta}{y^2} \left[ \ell(1+\ell^{-2}y^2)\, dt - y^2 d\hat\phi \right]^2  \\
&\quad + \frac{ \cos^2 \theta}{y^2} \left[ \ell(1+l^{-2}y^2)\, dt - y^2 d\hat\psi \right]^2
+ \frac{dy^2}{\ell^{-2}y^2+1+2y^{-4}\ell^2M} + {y^2 d\theta^2}  \\
&\quad -\frac{\ell^{-2}y^2+1}{y^2\ell^2}
\left[ \ell^2 (1+\ell^{-2}y^2)dt - {\ell y^2\sin^2 \theta}d\hat\phi - {\ell y^2\cos^2 \theta} d\hat\psi \right]^2+\mathcal O(\varepsilon).
\end{align*}
Finally, expanding the quadratic terms gives
\begin{align*}
\overline g_{\text{Kerr--AdS}} &= -\left( \frac{y^2}{\ell^2} + 1 \right) dt^2 
+ \frac{2M}{y^2} \left( dt - \ell \sin^2 \theta\, d\hat\phi - \ell \cos^2 \theta\, d\hat\psi \right)^2 \\
&\quad + \frac{dy^2}{1 + \frac{2M \ell^2}{y^4} + \frac{y^2}{\ell^2}} 
+ y^2 \left( d\theta^2 + \sin^2 \theta\, d\hat\phi^2 + \cos^2 \theta\, d\hat\psi^2 \right)+\mathcal O(\varepsilon).
\end{align*}
Letting $\varepsilon$ go to zero, we obtain
\begin{align*}
\overline g:=\lim_{\varepsilon\to0}\overline g_{\text{Kerr--AdS}}&= -\left( \frac{y^2}{\ell^2} + 1 \right) dt^2 
+ \frac{2M}{y^2} \left( dt - \ell \sin^2 \theta\, d\hat\phi - \ell \cos^2 \theta\, d\hat\psi \right)^2 \\
&\quad + \frac{dy^2}{1 + \frac{2M \ell^2}{y^4} + \frac{y^2}{\ell^2}} 
+ y^2 \left( d\theta^2 + \sin^2 \theta\, d\hat\phi^2 + \cos^2 \theta\, d\hat\psi^2 \right).
\end{align*}
This is the ultraspinning limit of Kerr AdS as first found in \cite{cvetivc2005supersymmetric}.
Moreover, this metric is smooth outside the origin $y=0$.

\subsection{Special cases}

Setting \( M = 0 \), the metric simplifies significantly. All terms involving \( M \) drop out and the metric becomes
\begin{align*}
\overline g &= -\left( \frac{y^2}{\ell^2} + 1 \right) dt^2 
+ \frac{dy^2}{\frac{y^2}{\ell^2} + 1} 
+ y^2 \left( d\theta^2 + \sin^2 \theta\, d\hat\phi^2 + \cos^2 \theta\, d\hat\psi^2 \right).
\end{align*}
We note that the angular part
    \[
    d\theta^2 + \sin^2 \theta\, d\hat\phi^2 + \cos^2 \theta\, d\hat\psi^2
    \]
    is simply the round metric on $S^3$ in Hopf coordinates.
Hence, $\overline g$ is simply the AdS metric in this case.

\medskip

Another special case is given by $\ell=\infty$ which corresponds to a vanishing cosmological constant.
To understand this scenario, we substitute $\hat\phi=\phi+\ell^{-1}t$, $\hat\psi=\psi+\ell^{-1}t$, $\mu=M\ell$ to find
\begin{align*}
\overline g=& -dt^2 
+ \frac{2\mu}{y^2} \left( -  \sin^2 \theta\, d\phi -  \cos^2 \theta\, d\psi \right)^2 \\
&\quad + \frac{dy^2}{1 + \frac{2\mu}{y^4} } 
+ y^2 \left( d\theta^2 + \sin^2 \theta\, d\phi^2 + \cos^2 \theta\, d\psi^2 \right)+\mathcal O(\ell^{-1}).
\end{align*}
If $M$ is fixed, the metric diverges for $\ell\to\infty$, but if we fix $\mu$ instead, it converges to
\begin{align}\label{l to infinity}
& -dt^2 
+ \frac{2\mu}{y^2} \left( -  \sin^2 \theta\, d\phi -  \cos^2 \theta\, d\psi \right)^2  + \frac{dy^2}{1 + \frac{2\mu}{y^4} } 
+ y^2 \left( d\theta^2 + \sin^2 \theta\, d\phi^2 + \cos^2 \theta\, d\psi^2 \right).
\end{align}
This metric is asymptotically flat and smooth outside of $y=0$.

\medskip

On the other hand, the Myers-Perry metric 
\begin{align*}
\overline g_{\text{Myers-Perry}} &= -\frac{\Delta}{\rho^2} \left[ dt - {a \sin^2 \theta} d\phi - {b \cos^2 \theta} d\psi \right]^2 
+ \frac{ \sin^2 \theta}{\rho^2} \left[ a\, dt - (r^2 + a^2) d\phi \right]^2  \\
&\quad + \frac{ \cos^2 \theta}{\rho^2} \left[ b\, dt - (r^2 + b^2) d\psi \right]^2
+ \frac{\rho^2 dr^2}{\Delta} + {\rho^2 d\theta^2}  \\
&\quad + \frac{1}{r^2 \rho^2} 
\left[ ab\, dt - b(r^2 + a^2) \sin^2 \theta d\phi - a(r^2 + b^2) \cos^2 \theta d\psi \right]^2
\end{align*}
where
\begin{align*}
\Delta &= \frac{1}{r^2} (r^2 + a^2)(r^2 + b^2) - 2m,  \\
\rho^2 &= r^2 + a^2 \cos^2 \theta + b^2 \sin^2 \theta,  
\end{align*}
simplifies for $a=b$ to
\begin{align*}
\overline g_{\text{Myers-Perry}} &= -(\rho^2r^{-2}-2m\rho^{-2}) \left[ dt - {a \sin^2 \theta} d\phi - {a \cos^2 \theta} d\psi \right]^2 
+ \frac{ \sin^2 \theta}{\rho^2} \left[ a\, dt - \rho^2 d\phi \right]^2  \\
&\quad + \frac{ \cos^2 \theta}{\rho^2} \left[ a\, dt - \rho^2 d\psi \right]^2
+ \frac{ d\rho^2}{1-2mr^{2}\rho^{-4}} + {\rho^2 d\theta^2}  \\
&\quad + \frac{1}{r^2 \rho^2} 
\left[ a^2\, dt - a\rho^2 \sin^2 \theta d\phi - a\rho^2 \cos^2 \theta d\psi \right]^2
\end{align*}
where
\begin{align*}
\rho^2 &= r^2 + a^2 .
\end{align*}
If $a\gg1$, we roughly have $r^2=-a^2$ and we recover the metric \eqref{l to infinity} with $\mu=ma^2$.
Therefore, for $l\to\infty$, we obtain the superextremal Kerr metric where the limit $\frac ma\to0$ is taken.

\subsection{Saturating the mass inequality}

Next, we consider the $t=0$ slice $(M,g,k)$ of $(\overline M^{n+1},\overline g)$ where we recall
\begin{align*}
\overline g=& -\left( \frac{y^2}{\ell^2} + 1 \right) dt^2 
+ \frac{2M}{y^2} \left( dt - \ell \sin^2 \theta\, d\hat\phi - \ell \cos^2 \theta\, d\hat\psi \right)^2 \\
& + \frac{dy^2}{1 + \frac{2M \ell^2}{y^4} + \frac{y^2}{\ell^2}} 
+ y^2 \left( d\theta^2 + \sin^2 \theta\, d\hat\phi^2 + \cos^2 \theta\, d\hat\psi^2 \right).
\end{align*}
The following is contained in \cite{cvetivc2005supersymmetric}, also see \cite[Remark 3.3]{ChruscielMaertenTod}:

\begin{theorem}
    Let $(M,g,k)$ be as above and set $\ell=1$. Then $(M,g,k)$ satisfies the dominant energy condition, is asymptotically AdS, and we have 
    \begin{align*}
      \mathcal   E=|  \mathcal A|
    \end{align*}
    for the energy and angular momentum.
    Moreover, $(\overline M^{n+1},\overline g)$ admits an imaginary Killing spinor, and therefore $(M,g,k)$ admits a spinor $\psi$ satisfying $\nabla_i\psi=-\frac12k_{ij}e_je_0\psi-\frac12\mathbf{i}e_i\psi$.
\end{theorem}

This completes the proof of the second part of Theorem \ref{Thm Intro examples}.
We emphasize that, strictly speaking, we cannot use the BPS bound / PMT inequality from Theorem \ref{Thm Intro PMT} to the initial data set $(M^n,g,k)$ above due to the naked singularity present at the origin $y=0$.
However, we expect that non-smooth versions of Theorem \ref{Thm Intro PMT} can be proven which allow for such types of singularities.
See for instance \cite{DaiSunWang2024_PMT_spin_conical,LiMantoulidis2019_skeleton,CecchiniFrenckZeidler2024, LeeLeFloch2015_distributional} for substantial progress made in the non-smooth asymptotically flat setting.

\medskip

It would be of interest to further investigate such ultraspinning black holes both in the asymptotically AdS and in the asymptotically flat setting.
In particular, Theorem \ref{Thm Intro rigidity} and Theorem \ref{Thm Intro examples} ask the question whether such examples exist in the smooth setting, or whether only AdS occurs in part (2) of Theorem \ref{Thm Intro rigidity}. 
In either case, it would be desirable to fully classify such extremal black hole solutions.
Finally, we would like to point out that when a non-trivial Maxwell field is present, there are numerous examples of smooth supersymmetric black holes \cite{kostelecky1996solitonic, gutowski2004supersymmetric}, even occurring in dimension $2+1$ \cite{banados1992black}.


\appendix

\section{Christoffel symbols and momentum densities}

Throughout this section, let $\alpha=2,\dots,n$ and $A,B=2,\dots,n-1$.  We use the notations $\nabla_i:=\nabla_{e_i}$, $w_{i,j}=\partial_{u_j} w_i$ and $k_{ij}=k(e_i,e_j)$. For simplicity, we write $\partial_i$ as shorthand for $\partial_{u_i}$.

\medskip

 Let $g$ be a metric on $\mathbb{R}^{n-1}\times \mathbb{R}_+$ given by
\begin{equation} \label{metric}
    g=(f^2+\sum_{\alpha=2}^{n}w_\alpha^2u_n^{-2})du_1^2+2 \sum_{\alpha=2}^n w_\alpha u_n^{-2}du_\alpha du_1+
u_n^{-2}(du_2^2+\cdots +du_n^2), \quad u_n>0.
\end{equation}
Consider on $(\mathbb{R}^{n-1}\times \mathbb{R}_+,g)$ the orthonormal frame
\begin{equation*}
    e_1=f^{-1}(\partial u_1-w_\alpha\partial u_\alpha), \quad e_\alpha= u_n \partial u_\alpha.
\end{equation*}
We have 
\begin{equation} \label{nuealpha}
    [e_1, e_\alpha]=-f^{-1}w_n \partial u_\alpha-u_n (f^{-1})_\alpha fe_1 +f^{-1}u_n w_{\beta,\alpha} \partial u_\beta,
    \quad 
    [e_A, e_n]=-e_A.
\end{equation}

\begin{lemma}\label{lemma christoffel symbols}
    Let $g$ be as above. Then all non-vanishing connection coefficients $\gamma_{ij}^l:=\langle\nabla_{e_i}e_j,e_l\rangle $ (up to the symmetry $\gamma_{ij}^l=-\gamma_{il}^j$) are given by
    \begin{align*}
         \gamma_{11}^\alpha=&u_n f\partial_\alpha(f^{-1}),\\
         \gamma_{1A}^\alpha=&\frac{1}{2}f^{-1}(w_{\alpha,A}-w_{A,\alpha}),\\
         \gamma_{A1}^B=&f^{-1}w_n u_n^{-1}\delta_{AB}-\frac{1}{2}f^{-1}(w_{A,B}+w_{B,A}),\\
         \gamma_{A1}^n=& \gamma_{n1}^A=-\frac{1}{2}f^{-1}(w_{A,n}+w_{n,A}),\\
          \gamma_{AB}^n=&\delta_{AB},\\
          \gamma_{n1}^n=& f^{-1}u_n^{-1} w_n-f^{-1}w_{n,n}.
    \end{align*}
\end{lemma}

\begin{proof}
Applying the formula 
\[\gamma_{ij}^l=\frac{1}{2}(\langle[e_i,e_j],e_l\rangle+\langle[e_l,e_i],e_j\rangle- \langle[e_j,e_l],e_i\rangle),\] 
we compute
\begin{align*}
    \gamma_{11}^\alpha=& \langle[e_\alpha, e_1],e_1\rangle=u_n f\partial_\alpha(f^{-1})
    \\ 
    \gamma_{1A}^B=& \frac{1}{2} (\langle [e_1,e_A], e_B\rangle+\langle [e_B,e_1], e_A\rangle-\langle [e_A,e_B], e_1\rangle)=
    \frac{1}{2}f^{-1}(w_{B,A}-w_{A,B})
    \\ \gamma_{1A}^n=& 
    \frac{1}{2} (\langle [e_1,e_A], e_n\rangle+\langle [e_n,e_1], e_A\rangle-\langle [e_A,e_n], e_1\rangle)
    =\frac{1}{2}f^{-1}(w_{n,A}-w_{A,n})
    \\ \gamma_{A1}^B=&\frac{1}{2} (\langle [e_A,e_1], e_B\rangle+\langle [e_B,e_A], e_1\rangle-\langle [e_1,e_B], e_A\rangle)
    \\=& 
f^{-1}w_n u_n^{-1}\delta_{AB}
    -\frac{1}{2}f^{-1}(w_{A,B}+w_{B,A})
    \\ \gamma_{A1}^n=& \frac{1}{2} (\langle [e_A,e_1], e_n\rangle+\langle [e_n,e_A], e_1\rangle-\langle [e_1,e_n], e_A\rangle)=-\frac{1}{2}f^{-1}(w_{A,n}+w_{n,A}).
    \end{align*}
    Moreover,
    \begin{align*}
        \gamma_{AB}^n=&
    \frac{1}{2} (\langle [e_A,e_B], e_n\rangle+\langle [e_n,e_A], e_B\rangle-\langle [e_B,e_n], e_A\rangle)=\delta_{AB}
    \\ \gamma_{n1}^A=&  \frac{1}{2} (\langle [e_n,e_1], e_A\rangle+\langle [e_A,e_n], e_1\rangle-\langle [e_1,e_A], e_n\rangle)=-\frac{1}{2}f^{-1}(w_{A,n}+w_{n,A})  
    \\ \gamma_{n1}^n=& \langle [e_n,e_1], e_n\rangle=f^{-1}u_n^{-1} w_n-f^{-1}w_{n,n},\\
    \gamma_{nA}^n=& \langle [e_n,e_A], e_n\rangle=0
    \\ \gamma_{n A}^B=& 
    \frac{1}{2} (\langle [e_n,e_A], e_B\rangle+\langle [e_B,e_n], e_A\rangle-\langle [e_A,e_B], e_n\rangle)=0, \\
    \gamma_{AB}^C=&0.
\end{align*}
\end{proof}

\begin{lemma}\label{Je1}
    Let   
   $g$ be the metric given in Equation \eqref{metric}
and $k_{ij}=-\frac{1}{2}fu_n^{2}[\nabla_i X_j+\nabla_j X_i]$, where $X=u_n^{-2}\nabla u_1$. Denote $L= u_n^2 f^2 
         -2w_1+\sum_{\alpha=2}^n|w_\alpha|^2$. Then 
\begin{equation*}
    \langle J,e_1 \rangle=-\mu= \frac{1}{2}f^{-2}u_n^{-2}(\Delta_{\Sigma_1} L-2u_n\partial_n L) 
\end{equation*}
 
\end{lemma}
\begin{proof}
We compute
\begin{align*}
     \langle J,e_1\rangle
        =&e_i(k_{i1})-k(\nabla_{e_i}e_i,e_1)-k(e_i,\nabla_{e_i}e_1)-e_1(k_{ii})
        \\=& u_n \partial_\alpha [-u_n \partial_\alpha(f^{-1})f]-k_{j 1}\gamma_{ii}^j-k_{ij}\gamma_{i1}^j
        -e_1(k_{\alpha\alpha})
        \\=&  u_n^{2}\partial_\alpha\partial_\alpha \log f+u_n \partial_n \log f
        -e_1(\log f+2\log u_n) (f^{-1}w_{\alpha\alpha}-(n-1)f^{-1}w_nu_n^{-1})
        \\&-(u_n \partial_\alpha \log f+\delta_{\alpha n} )((n-2)\delta_{n\alpha}-u_n\partial_\alpha \log f)
        \\ &+ (u_n\partial_\alpha\log f+\delta_{\alpha
         n})u_n\partial_\alpha \log f+ 
         |f^{-1}w_{\alpha\beta}-f^{-1}w_nu_n^{-1}\delta_{\alpha\beta}|^2
         \\& -f^{-1}(\partial_1-w_\beta\partial_\beta)(f^{-1}w_{\alpha\alpha}-(n-1)f^{-1}u_n^{-1}w_n)
\end{align*}
where we used the formulas for the Christoffel symbols from Lemma \ref{lemma christoffel symbols}.
Simplifying, we find that the above term equals
\begin{align*}
    & u_n^2 f^{-1}f_{\alpha\alpha}-u_n^{2}f^{-2}|f_\alpha|^2+u_n f^{-1}f_n
         \\&- (f^{-2}f_1-f^{-2}f_\alpha w_\alpha -2f^{-1}w_n u_n^{-1})(f^{-1}w_{\beta\beta}-(n-1)f^{-1}w_n u_n^{-1})
         \\ & +2y_n^2f^{-2}|f_\alpha|^2-(n-4)u_nf^{-1}f_n-(n-2) 
         + f^{-2} |w_{\alpha\beta}|^2-2f^{-2}w_n u_n^{-1}w_{\alpha\alpha} 
         \\& +(n-1)f^{-2}w_n^2y_n^{-2}- f^{-2}w_{1\alpha\alpha}+f^{-3}f_1w_{\alpha\alpha}+f^{-2}w_{\beta}w_{\alpha\alpha\beta}-f^{-3}f_\beta w_{\beta}w_{\alpha\alpha}
         \\& -(n-1)f^{-3}f_1 u_n^{-1}w_n+(n-1)f^{-2}u_n^{-1}w_{1n}+(n-1)f^{-3}f_\beta w_\beta u_n^{-1}w_n
         \\& -(n-1)f^{-2}u_n^{-1}w_\beta w_{n\beta}+(n-1)f^{-2}w_n^2 u_n^{-2}
         \\=& 
         u_n^{2}f^{-1}f_{\alpha\alpha}+u_n^2f^{-2}|f_\alpha|^2
         -(n-5)u_n f^{-1}f_n
         -(n-2)-f^{-2}w_{1\alpha\alpha}+(n-1)f^{-2}u_n^{-1}w_{1n}
         \\ & +f^{-2}|w_{\alpha\beta}|^2+f^{-2}w_\beta w_{\alpha\alpha\beta}-(n-1)f^{-2}w_\beta w_{n \beta}u_n^{-1}.
         \end{align*}
                Finally, we use the definition of $L$ to obtain
         \begin{align*}
          \langle J,e_1\rangle =& f^{-2}\partial_{\alpha\alpha}\left(\frac{1}{2} u_n^2 f^2 
         -w_1+\frac{1}{2}|w_\beta|^2
         \right)
         -(n-1)f^{-2}u_n^{-1}\partial_n \left(\frac{1}{2}u_n^2 f^2-w_1+\frac{1}{2}|w_\beta|^2\right)
         \\=&\frac{1}{2}\left[\sum_{\alpha=2}^n f^{-2}\partial_{\alpha\alpha}L
         -(n-1)f^{-2}u_n^{-1}\partial_n L\right]
         \\=& \frac{1}{2}f^{-2}u_n^{-2}(\Delta_{\Sigma_1} L-2u_n\partial_n L) .
\end{align*}
\end{proof}

 \section{Charges of Siklos waves}\label{A:Siklos}

 For convenience we restate Theorem \ref{Thm Siklos wave charges}:
 \begin{theorem}
     Let $(M^n,g,k)$ be the $t=0$-slice of the Siklos wave $(\overline{M}^{n+1},\overline{g})$, where
     \begin{equation*}
          \overline g=\frac{1}{u^{2}_n}\Big( 2du_1dt+ Ldu_1^{2}  +\delta_{\alpha \beta}du_{\alpha}du_{\beta}
          \Big)\quad \text{and}\quad L=1+O_2(r^{-q}).
     \end{equation*}
     Then we have for any lapse-shift pair $(N,X)$
     \begin{equation*}
       \mathcal{H}(N,X)= \frac{1}{(n-1)\omega_{n-1}}\int_{\mathbb{H}^n}[N\mu+\langle X,J\rangle_b]LdV_b.
     \end{equation*}
 \end{theorem}

 \begin{proof}
 Let $\partial_i:=\partial_{u_j}$, $e_i:=u_n\partial_{i}$ and $\mathring{\gamma}_{ij}^k:=\langle\mathring{\nabla}_{e_i}e_j,e_k\rangle$.
 The non-vanishing coefficients of $\mathring{\gamma}$ are given by
\begin{equation*}
    \mathring{\gamma}_{\alpha\beta}^n=\delta_{\alpha\beta},\quad \mathring{\gamma}_{\alpha n}^\beta=-\delta_{\alpha\beta}.
\end{equation*}
 Let $F:=L-1$. 
   Then the non-vanishing gradient terms are
   \begin{equation*}
       \mathring{\nabla}_j g_{11}=e_j(F),\quad 
       \mathring{\nabla}_1 g_{n1}=\mathring{\nabla}_1 g_{1n}=F.
   \end{equation*}
   Using $F=O_2(r^{-q})$, we obtain 
   \begin{align*}
       \mathbb{U}^i(N) =&2 \left( Ng^{i[k} g^{j]l} \mathring\nabla_j g_{kl} + \mathring{\nabla}^{[i}Ng^{j]k} (g_{jk} - b_{jk}) \right)
       \\=&N(\delta_{i1}e_1(F)+\delta_{in}F-e_i(F))+e_i(N)F-e_1(N)\delta_{i1}F+O_1(Nr^{-2q}).
   \end{align*}
 Denoting by $\operatorname{div}_b$ be the divergence with respect to the hyperbolic metric $b$, we compute
   \begin{align*}
       \operatorname{div}_b [\mathbb{U}(N)]=&\sum_{i=1}^n e_i[\mathbb{U}^i(N)]-(n-1)\mathbb{U}^n(N)
       \\=&\sum_{A=2}^{n-1}e_A[-Ne_A(F)+e_A(N)F]+e_n[NF-Ne_n(F)+e_n(N)F]
       \\ &-(n-1)[NF-Ne_n(F)+e_n(N)F]+\operatorname{div}_b(O_1(Nr^{-2q}))
   \end{align*}
   The identities $\nabla_{ij}N=Ng_{ij}$ and $\nabla_{\alpha\alpha}F=e_\alpha^2(F)-e_n(F)$ imply
   \begin{align*}
       \operatorname{div}_b [\mathbb{U}(N)]=&\sum_{A=2}^{n-1}\left[-Ne_A^2(F)+e_A^2(N)F\right]+e_n(N)F+Ne_n(F)-Ne_n^2(F)+e_n^2(N)F
      \\& -(n-1)[NF-Ne_n(F)+e_n(N)F]+\operatorname{div}_b(O_1(Nr^{-2q}))
      \\=&-N\Delta_{\Sigma_1}F+2Ne_n(F)+\operatorname{div}_b(O_1(Nr^{-2q})).
   \end{align*}
   Recall that $\mu=-\frac{1}{2}L^{-1}[\Delta_{\Sigma_1} L-2e_n (L)]$ and $q>\frac{n}{2}$. 
   Hence,
   \begin{equation*}
     \lim_{R\to \infty}\int_{r=R}\mathbb{U}^i(N)\nu_idA=\int_{\mathbb{H}^n} 2NL\mu dV_b
   \end{equation*}
   Note that $k_{ij}=-\frac{1}{2}L^\frac{1}{2}u_n[\nabla_i (L^{-1}\partial_{1})_j+\nabla_j(L^{-1}\partial_{1})_i]$. 
   Since
   \begin{equation*}
       \nabla_{\partial_ {i}}\partial_{1}=\sum_{k=2}^n -\frac{1}{2}u_n^2\delta_{i1}[\partial_k (u_n^{-2}L)] \partial_{k}+\frac{1}{2}u_n^2L^{-1}[\partial_i(u_n^{-2}L)]\partial_{1},
   \end{equation*}
   we have
   \begin{equation*}
       \langle\nabla_{\partial_{i}} (L^{-1}\partial_{1}), \partial_{j}\rangle=  -\frac{1}{2}L^{-1}(1-\delta_{j1})\delta_{i1}\partial_{j}(u_n^{-2}L)+\frac{1}{2}L^{-1}\delta_{j1}\partial_{i} (u_n^{-2}L)+\delta_{j1}Lu_n^{-2}\partial_{i}(L^{-1})
   \end{equation*}
   and
   \begin{equation*}
       \langle\nabla_{\partial_{i}} X, \partial_{j}\rangle+\langle\nabla_{\partial_{j}} X, \partial_{i}\rangle=\delta_{i1}\delta_{j1}L^{-1}\partial_1(u_n^{-2}L)+\delta_{j1}Lu_n^{-2}\partial_i(L^{-1})+\delta_{i1}Lu_n^{-2}\partial_j(L^{-1})
   \end{equation*}
   Consequently, $$k(e_i,e_j)=-\frac{1}{2}u_n\left(\delta_{i1}\delta_{j1}\partial_{1}L-\delta_{j1}\partial_i L-\delta_{i1}\partial_j L\right)+O_1(r^{-2q})$$ and
   \begin{equation*}
       2(\operatorname{div}_bk)_j=
       e_1e_j L+\sum_{\alpha=2}^n\delta_{j1}e_\alpha^2L-n\delta_{j1}e_nL+\delta_{jn}e_1L+\operatorname{div}_bO_1(r^{-2q}).
   \end{equation*}
   We also have
   \begin{align*}
        \operatorname{div}_b\mathbb{V}(X) =& 2\mathring{\nabla}_i\left[ \left( k_{ij} - \tr_g(k) g_{ij} \right) X^j\right]
        \\=& X^j\left( e_1e_j L+\sum_{\alpha=2}^n\delta_{j1}e_\alpha^2L-n\delta_{j1}e_nL+\delta_{jn}e_1L-e_je_1L\right)+\operatorname{div}_b(O_1(r^{-2q+1}))
        \\=& X^j\delta_{j1}\left(\sum_{\alpha=2}^n\delta_{j1}e_\alpha^2L-n\delta_{j1}e_nL\right)+\operatorname{div}_b(O_1(r^{-2q+1}))
        \\=& 2L\langle X,J\rangle_b+\operatorname{div}_b(O_1(r^{-2q+1}))
   \end{align*}
   where the second step uses $\mathring{\nabla}_i X^j+\mathring{\nabla}_j X^i=0$ and the last step uses $J=-\mu e_1=\frac{1}{2}L^{-1}(\Delta_{\Sigma_1}L-2e_n L)e_1$. 
   Moreover,
   \begin{equation*}
     \lim_{R\to \infty}\int_{r=R}\mathbb{V}^i(X)\nu_idA=\int_{\mathbb{H}^n} 2L\langle X,J\rangle_b dV_b.
   \end{equation*}
   Therefore, 
   \begin{align*}
       \mathcal{H}(N,X)=& \lim_{R \to \infty} \frac{1}{2(n-1)\omega_{n-1}} \int_{r=R} \left( \mathbb{U}^i(N) + \mathbb{V}^i(X) \right) \nu_i dA
       \\=& \frac{1}{(n-1)\omega_{n-1}}\int_{\mathbb{H}^n}[N\mu+\langle X,J\rangle_b]LdV_b
   \end{align*}
   which finishes the proof.
\end{proof}

\bibliographystyle{plain}
\bibliography{literature.bib}

\end{document}